\documentclass[moor,nonblindrev]{informs4}
\makeatletter
\if@NONBLINDREV
  \def\theARTICLETOPLEFT{}
  \def\theARTICLETOPRIGHT{}
\fi
\makeatother
\RequirePackage{tgtermes}
\RequirePackage{newtxtext}
\RequirePackage{newtxmath}
\RequirePackage{bm}
\RequirePackage{endnotes}
\usepackage{hyperref}

\OneAndAHalfSpacedXII

\usepackage{graphicx}
\usepackage{dsfont}
\usepackage{tikz}
\usepackage[capitalize]{cleveref}
\usepackage{color}
\usepackage{mathtools}
\usepackage{algorithm}
\usepackage{algpseudocode}
\usepackage{xcolor}
\usepackage{subcaption}
\usepackage{wrapfig}

\usepackage{diagbox}
\usepackage{makecell}

\usepackage[ezlist,ezrestate, ezeq]{ezlib_OR}

\usepackage{natbib}
 \bibpunct[, ]{(}{)}{,}{a}{}{,}%
 \def\bibfont{\small}%

\EquationsNumberedThrough
\TheoremsNumberedThrough
\ECRepeatTheorems

\renewcommand{\leq}{\leqslant}
\renewcommand{\geq}{\geqslant}
\renewcommand{\le}{\leqslant}
\renewcommand{\ge}{\geqslant}

\newcommand{\E}{\mathbb{E}}
\undef\P
\newcommand{\P}{\mathbb{P}}
\newcommand{\inner}[3][2]{\left\langle #2 , #3 \right\rangle}
\newcommand{\I}{\mathcal{I}}

\newcommand{\origduration}{\textnormal{\textsf{original priority index}}}
\newcommand{\remduration}{\textnormal{\textsf{remaining priority index}}}

\newcommand{\SEB}{\mathrm{SEB}}
\newcommand{\PSJF}{\mathrm{PSJF\textnormal{-}1}}
\newcommand{\SRPT}{\mathrm{SRPT\textnormal{-}1}}
\newcommand{\OPT}{\mathrm{OPT}}
\newcommand{\mgone}{\mathrm{M/G/1}}

\allowdisplaybreaks

\begin{document}

\TITLE{Optimal Scheduling in Generalized Switch in Heavy Traffic}

\RUNAUTHOR{Xie, Scully, Righter, and Grosof}
\RUNTITLE{Optimal Scheduling in Generalized Switches in Heavy Traffic}

\ARTICLEAUTHORS{%
\AUTHOR{Runhan Xie}
\AFF{Department of Industrial Engineering and Operations Research, University of California, Berkeley, Berkeley, CA, USA, \EMAIL{runhan\_xie@berkeley.edu}}

\AUTHOR{Ziv Scully}
\AFF{School of Operations Research and Information Engineering, Cornell University, Ithaca, NY, USA, \EMAIL{zivscully@cornell.edu}}

\AUTHOR{Rhonda Righter}
\AFF{Department of Industrial Engineering and Operations Research, University of California, Berkeley, Berkeley, CA, USA, \EMAIL{rrighter@berkeley.edu}}

\AUTHOR{Izzy Grosof}
\AFF{Department of Industrial Engineering and Management Sciences, Northwestern University, Evanston, IL, USA, \EMAIL{izzy.grosof@northwestern.edu}}
}

\ABSTRACT{
The generalized switch is a highly flexible queueing model, covering multiclass, multiserver, and multiresource queueing systems as well as a wide variety of stochastic networks. Although many scheduling policies have been developed for this model, they almost entirely address \emph{unknown} job duration settings. How to optimally use \emph{known} job durations in the generalized switch has remained open. Moreover, optimizing mean response time remains open in both settings. We introduce the first policy to guarantee heavy-traffic optimal mean response time in the generalized switch, our Smallest Equalizing Bucket (SEB) policy. The key challenge in designing an optimal scheduling policy is that we must simultaneously prioritize small jobs and also minimize resource waste, all while fitting within the generalized switch's service options. SEB overcomes this challenge by grouping jobs into duration-based buckets and enforcing an “equalizing’’ service structure that keeps each bucket balanced while still prioritizing the smallest jobs. We prove SEB’s heavy-traffic optimality. Simulations further confirm the effectiveness of SEB-inspired heuristics.}

\maketitle

\section{Introduction}
The generalized switch is a versatile multi-class queueing model in which, at any time, a scheduler can choose one from a set of service options that determines the number of jobs from each class that can be served simultaneously. Since its introduction \citep{stolyar2004maxweight}, it has been an important, highly general model which encompasses many standard queueing models as special cases. Notable examples of theoretical importance include: the multiserver job (MSJ) model, where different jobs require different numbers of servers; the multiresource job model, where jobs require combinations of resources such as servers, memory, and disk space; the compatibility model, where each job class can be served only by a subset of servers; and redundancy models, where identical job copies are sent to multiple servers. Beyond its theoretical significance, the generalized switch has found applications in modeling input-queued switches in wireless networks \citep{mckeown2002achieving}, virtual machines in cloud computing \citep{maguluri2014heavy}, and multi-core jobs in datacenters \citep{harchol2022multiserver}, among others.

Designing and analyzing effective scheduling algorithms has always been a central topic for the generalized switch. A well-chosen scheduling policy can substantially improve key performance metrics (e.g. throughput, mean response time%
\footnote{A job's response time, also known as delay or sojourn time, is the amount of time between its arrival and its completion.})
with no additional resources. Many scheduling policies have been proposed and studied for generalized switch models, such as MaxWeight \citep{tassiulas1990stability,stolyar2004maxweight}, Randomized Timers \citep{ghaderi2016randomized}, and Markovian Service Rate (MSR) policies \citep{chen2025improving}, as well as special cases such as ServerFilling-SRPT/DivisorFilling-SRPT in MSJ systems \citep{grosof2022optimal}, Least-Redundancy-First (LRF) in redundancy systems \citep{anton2022scheduling}, multiserver SRPT in M/G/k systems \citep{grosof2018srpt}, etc. 

However, this extensive body of work overwhelmingly focuses on \emph{unknown-duration} settings, rather than settings where duration is known in advance. Existing  policies either (1) ignore individual job durations (e.g., MaxWeight, MSR policies, LRF), or (2) address known-duration scheduling only in specific subclasses of the generalized switch (e.g., ServerFilling-SRPT/DivisorFilling-SRPT in MSJ systems, multiserver SRPT in M/G/k systems). \emph{Known-duration} scheduling in the generalized switch is an open problem. Motivated by this gap, we propose and analyze a new policy called Smallest Equalizing Bucket (SEB), the first policy which provably achieves heavy-traffic optimal mean response time in the generalized switch (\Cref{thm:heavy-traffic_optimality}).
We now discuss challenges in achieving this result, and the intuition behind the SEB policy.

\subsection{Challenges and Policy Intuition}
When minimizing mean response time, the primary goal of a scheduling policy is to prioritize the jobs of smallest (remaining) duration. In a homogeneous job model like the M/G/k, straightforwardly prioritizing such jobs with the SRPT-$k$ policy is sufficient to achieve an asymptotically optimal mean response time \citep{grosof2018srpt}.
However, in the generalized switch, greedily serving small jobs may fail to use the full capacity of the system, or may leave us with unused resources that can only be utilized by large jobs. This occurs when there is a dramatic imbalance in the resource demands of the small jobs. For instance, in the compatibility scheduling setting, if all of the small jobs are the same class, which only half the servers can serve, there is no way to serve the small jobs with the entire capacity of the system.

We navigate this challenge by \emph{preventing imbalance from arising in the first place}.
Our SEB policy balances the set of small jobs in advance, and more generally balances all sets of jobs of similar durations (\cref{sec:policy}). In particular, SEB works as follows:

\* In place of the usual notion of job durations, we define, for each job in the system, a novel \emph{priority index}, which is a carefully chosen scaling of the job duration.
\* We divide possible job priority indices into disjoint intervals (buckets). Each job is placed into a bucket based on its original priority index.
\* Within each bucket, we find the ideal ``equalizing'' service option that maintains its balance.
\* We find the smallest bucket (in priority-index ordering) for which the ideal service option is available, and serve that bucket. We only move on to a larger bucket if the current bucket cannot fulfill its ideal service option because there are too few jobs in a smaller bucket.
\*/

\subsection{Paper Outline}
The rest of the paper is organized as follows:
\*\Cref{sec:prior_work} reviews prior work. 
\*\Cref{sec:model} formally defines the generalized switch model we consider. 
\* \Cref{sec:resource-pooled-lower-bound} describes the construction of a response time lower bounding system and priority indices of jobs and introduces other assumptions for the results.
\*\Cref{sec:policy} defines the SEB policy and discusses the intuition behind its design. 
\*\Cref{sec:results} states our main result, namely SEB's heavy-traffic optimality.
\*\Cref{sec:analysis_1} proves stability and state-space collapse (bucket balance) lemmas.
\*\Cref{sec:analysis_2} upper bounds mean response time and proves our optimality main result.
\*\Cref{sec:dynamic-SEB} introduces three dynamic SEB policies, all of which are inspired by SEB but overcome SEB's shortcomings in practice.
\*\Cref{sec:simulation} numerically evaluates dynamic SEB policies via simulations.
\*/

\section{Prior work}
\label{sec:prior_work}
We first review prior work on generalized switch scheduling (\cref{sec:prior-gs}), then
discuss prior work on special cases of the generalized switch (\cref{sec:prior-special-case-gs}). A comparison of existing scheduling policies and our policy can be found in \Cref{tab:optimality}.

\subsection{Generalized Switch Scheduling}
\label{sec:prior-gs}
Many scheduling policies focus on throughput optimality in generalized switch models. The most studied is the MaxWeight policy (e.g. \citep{tassiulas1990stability, stolyar2004maxweight}), which assigns each class a weight according the total number of jobs in that class, then serves the class with the highest weight. It is shown that MaxWeight achieves both throughput optimality and mean work optimality. 

Despite being throughput optimal, MaxWeight has high complexity, as a new schedule must be computed by solving a combinatorial optimization problem whenever the system state changes. As a result, other policies have been invented to reduce complexity while maintaining throughput optimality. Low-complexity MaxWeight variants have been shown to achieve mean-work-optimality in the same heavy traffic regime \citep{jhunjhunwala2022low}.
The Randomized Timers policy\footnote{We note that although Randomized Timers and MSR policies are designed for VM scheduling, they can be adapted for the generalized switch model.}  is a non-preemptive policy that adjusts the schedule according to exponential clocks to achieve throughput optimality \citep{ghaderi2016randomized, psychas2018randomized}.
The Markovian Service Rate (MSR) policy \citep{chen2025improving} first constructs a continuous-time Markov chain (CTMC) offline, then uses this CTMC to determine a schedule, resulting in a throughput-optimal policy. The schedules can be preemptive or non-preemptive and avoid solving complicated optimization problems online.

Mean response time results are known for MaxWeight and for MSR, and \citep{chen2025improving} prove that the two policies are constant-competitive.

In contrast to our work, none of these policies use duration information of individual jobs.

\subsubsection{Complete Resource Pooling (CRP) Condition}
Roughly speaking, the CRP condition for heavy-traffic trajectories refers to cases where a system converges to a high-dimensional face of the stability region, rather than a lower-dimensional edge between faces. Under the CRP condition \citep{harrison1999heavy}, MaxWeight exhibits state-space collapse to a single dimension,
maximizes throughput, and minimizes mean work in heavy traffic \citep{stolyar2004maxweight}. \citep{hurtado2022heavy} studied the behavior of MaxWeight when the CRP condition does not hold, demonstrating collapse to a higher-dimensional subspace and characterizing the heavy-traffic stationary behavior.

In this paper, we focus on the CRP setting: See \cref{sec:assumptions}.

\subsection{Scheduling in Special Cases of the Generalized Switch}
\label{sec:prior-special-case-gs}
Here we focus primarily on the MSJ model in  \Cref{sec:prior-MSJ} and the compatibility model in \Cref{sec:prior-compat}.

\begin{table}
\begin{center}
\resizebox{\textwidth}{!}{
\begin{tabular}{|c|c|c|c|c|}
\hline
\diagbox[innerleftsep=12pt]{Policies}{Results}    &  Throughput & Job durations & Mean Response Time & Model Generality \\
\hline
FCFS & Not optimal & Unknown & Analyzed, not optimal & MSJ and MRJ models\\
\hline
MaxWeight & Optimal & Unknown & Analyzed, not optimal & Generalized Switch\\
\hline
Randomized Timers & Optimal & Unknown & Not analyzed & Generalized Switch\\
\hline
\makecell{ServerFilling/ \\ DivisorFilling} & Optimal & Unknown & Analyzed, not optimal & \makecell{Restrictive MSJ Model}\\
\hline
Idle-Avoid $c\mu/m$ Rule &  Optimal & Unknown & \makecell{Optimal in many-server \\ limit over a finite horizon} & MSJ Model\\
\hline
\makecell{ServerFilling-SRPT/ \\ DivisorFilling-SRPT} & Optimal & Known & \makecell{Heavy-traffic optimal} & \makecell{Restrictive MSJ Model}\\
\hline
\makecell{MSR} & Optimal & Unknown & \makecell{Analyzed, not optimal} & \makecell{Generalized Switch}\\
\hline
\makecell{\textbf{Smallest Equalizing} \\ \textbf{Bucket (SEB)}} & \textbf{Optimal} & \textbf{Known} & \makecell{\textbf{Heavy-traffic optimal}} & \textbf{Generalized Switch}\\
\hline
\end{tabular}}
\end{center}
\caption{Comparison of optimality results for our paper and for prior work}
\label{tab:optimality}
\end{table}

\subsubsection{MSJ Scheduling}
\label{sec:prior-MSJ}
MSJ scheduling has recently become a topic of interest in the queueing theory community.
Despite some recent advances, theoretical results on MSJ scheduling remain limited \citep{harchol2022multiserver}.
We now give a brief overview of MSJ scheduling policies with theoretical guarantees.
In \cref{tab:optimality}, we compare our results in this paper to prior policies and results.

\textbf{First-Come-First-Served (FCFS):} FCFS with head-of-line-blocking is the most straightforward policy. However, FCFS is generally not throughput nor mean response time optimal due to blocking at the front of the queue. The stability region under FCFS has been characterized under restrictive assumptions (e.g. \citep{morozov2016stability, rumyantsev2017stability, afanaseva2020stability}).
Mean response time under FCFS is known exactly only in very specific settings \citep{kim1979m, brill1984queues, filippopoulos2007m}.
A general bound and heavy-traffic characterization of the mean response time under FCFS was established by \citet{grosof2023reset}. This general bound extends to the MRJ model as well.

\textbf{ServerFilling/DivisorFilling:} \citet{grosof2022wcfs} introduce a queueing framework called Work Conserving Finite Skip (WCFS) and propose the ServerFilling/DivisorFilling policies, which consider the minimal set of jobs necessary to fill all the servers. A critical assumption in \citep{grosof2022wcfs} is that the server needs of jobs must divide the number of servers, ensuring that it is possible to fill all of the servers whenever enough jobs are present. It is shown that ServerFilling/DivisorFilling is throughput optimal and the heavy-traffic mean response time is comparable to that in M/G/1/FCFS.

\textbf{ServerFilling/DivisorFilling-SRPT:} \citet{grosof2022optimal} propose the ServerFilling/DivisorFilling-SRPT policies, which prioritize jobs of shortest remaining duration. These are the first MSJ scheduling policies that are provably both throughput and heavy-traffic mean response time optimal. Similar to this paper, ServerFilling/DivisorFilling-SRPT uses individual job durations for scheduling. However, here we relax the restrictive assumption that the server needs of jobs must divide the number of servers. 

\textbf{Other policies:}
Other scheduling policies have been proposed: variants of Backfilling e.g. \citep{jones1999scheduling, wang2009application, carastan2019one}, Smallest Area First \citep{carastan2019one}, reinforcement learning based Shortest Job First \citep{guo2018optimal}, and the idle-avoid $c\mu/m$ rule \citep{zylchlinski2023managing}, among many others. Mean response time optimality results have been shown in the many-servers limit \citep{zylchlinski2023managing,hong2022sharp}, but heavy-traffic optimality remains open in general, see \cref{tab:optimality}.

\subsubsection{Compatibility Scheduling}
\label{sec:prior-compat}
In recent years, queues with compatibility between jobs and servers have been extensively studied. In these queues, a job can only enter service at a subset of servers as specified by a compatibility graph. 

On the scheduling side, most existing work in similar settings treats jobs as indistinguishable and focuses on load-balancing algorithms. Optimality results are established in the mean-field limit or many-server limit (e.g. \citep{tsitsiklis2013queueing,mukherjee2018asymptotically,weng2020optimal,rutten2023load}). \citep{anton2024efficient} focus on scheduling based on the compatibility structures.
We are aware of no prior study of scheduling known-duration jobs in the compatibility scheduling setting.

\section{Model}
\label{sec:model}

The generalized switch model we consider is defined as follows:
There are $n_c$ classes of jobs, and a set of service options $\mathcal{R} = \{\mathbf{r}\}$ to choose from.
A service option $\mathbf{r}$ specifies a number of jobs of each class which can be served at once.
We assume that no more than $n_s$ jobs may be served simultaneously for any service option.

At any given time, the scheduling policy $\pi$ selects any service option $\mathbf{r} \in \mathcal{R}$, and serves up to $r_i$ class-$i$ jobs, which can be any of the class-$i$ jobs in the system. If the current service option is $\mathbf{r}$, then the total remaining duration of jobs in class $i$ is decreasing with rate $r_i$\footnote{Note that this can be straightforwardly generalized to the case where different jobs are served at different rates. We use the current description for simplicity.}.

Jobs of class $i$ arrive according to a Poisson process with rate $\lambda_i$, with an overall arrival rate of $\boldsymbol{\lambda}$. Each job has a service duration sampled i.i.d. from some general class-specific distribution with random variable $D_i$.
The remaining service time and class of each job in the system is known to the scheduling policy at all times.
The scheduler may choose any service option from the list $\mathcal{R}$ at any moment in time. 
Jobs may be preempted and resumed with no overhead or loss of work.

Using the job duration and arrival rates, we define the system load vector $\boldsymbol{\rho}$ with class-$i$ load $\rho_i=\lambda_i\E[D_i]$. Then the stability region $\mathcal{S}$ of the system is the open interior of the convex hull formed from all service options:
$
    \mathcal{S} = \text{Interior}(\text{ConvexHull}(\mathcal{R}))
$. It has been shown that the system is stabilizable if and only if $\boldsymbol{\lambda} \in \mathcal{S}$ \citep{stolyar2004maxweight}.

A special subclass of generalized switch models that we will use for examples is the MSJ models. Here, the system has a fixed number of servers, and each job occupies multiple servers concurrently during service. Job classes are defined by the number of servers a job needs. A service option consists of a vector of numbers of jobs of each class in service, such that the total server need never exceeds the total number of servers. For example, consider an MSJ setting with $n_s=7$ servers, and where jobs can have server needs $1, 2,$ or $3$. One possible service option is $\mathbf{r}=[2, 1, 1]$, where two 1-server jobs, one $2$-server job, and one $3$-server job are served.

\section{Comparison with a Resource-Pooled System}
\label{sec:resource-pooled-lower-bound}
Our goal in this paper is to design a scheduling policy that minimizes mean response time in heavy traffic. To prove the optimality of our policy, we establish an upper bound on the mean response time of our policy and a lower bound on the mean response time of \emph{any} scheduling policy and show that the upper and lower bounds match in heavy traffic.

Our approach for establishing such a lower bound is to study the response time in a properly chosen resource-pooled single-server queue. We use the resource-pooled system to prove a universal lower bound for any scheduling policy in the original system (\Cref{sec:lower-bound}). 

In this section, we first introduce a family of resource-pooled systems, all of which can serve as mean response time lower bounds for the generalized switch, then we choose a tight lower bound that we will show converges to our upper bound in heavy traffic. Finally, we specify the class of generalized switch models we focus on in this paper.

\subsection{Lower Bounds from Resource-Pooled Systems}
\label{sec:lower-bound}

We lower bound mean response time in the generalized switch via a carefully selected resource-pooled single-server queue.

A resource-pooled single-server system is an M/G/1 with a related but distinct duration distribution to that of the original generalized switch, where every scheduling policy in the generalized switch can be mapped onto an equivalent scheduling policy in the resource-pooled system, resulting in the same response time distribution. Since this resource-pooled system has many more service options that are not realizable in the generalized switch model, the minimum possible mean response time in the resource-pooled system, given by the single-server SRPT policy, is a lower bound on the minimum possible mean response time in the original generalized switch.

To define this resource-pooled single-server queue, we create a mapping converting a job's duration in the original system to a duration in the resource-pooled system.
We prove the corresponding lower bound as \cref{thm:rp-lower-bound} below.

The key property of our duration mapping is that for any service option in the original system, the corresponding service rate in the resource-pooled system is no more than the capacity of the resource-pooled system, which we normalize to 1. 

\begin{figure}[h]
    \centering
    \begin{tikzpicture}[scale=0.7]
\tikzstyle{every node}=[font=\LARGE]
\draw  [fill=blue!50] (12.5,14.25) circle (0.5cm);
\draw  [fill=blue!50](12.5,13) circle (0.5cm);
\draw  [fill=blue!50](12.5,11.75) circle (0.5cm);
\draw  [fill=blue!50](12.5,10.5) circle (0.5cm);
\draw  [fill=blue!50](12.5,9.25) circle (0.5cm);
\draw  [fill=blue!50](12.5,8) circle (0.5cm);
\draw  (12.5,6.75) circle (0.5cm);

\draw  (1.5,11.25) -- (11.25,11.25);
\draw  (1.5,9.75) -- (11.25,9.75);
\draw  (11.25,11.25) -- (11.25,9.75);
\draw  (9.5,11.25) -- (9.5,9.75);
\draw  (7.75,11.25) -- (7.75,9.75);
\draw  (6,11.25) -- (6,9.75);
\draw  (4.25,11.25) -- (4.25,9.75);
\draw  (2.5,11.25) -- (2.5,9.75);

\draw [ rounded corners = 8.4] (11.75,14.75) rectangle (13.25,11.25);
\draw [ rounded corners = 8.4] (11.75,11) rectangle (13.25, 7.5);
\draw  [fill=blue!50](10.375,12.75) circle (0.5cm);
\draw  [fill=blue!50](10.375,11.5) circle (0.5cm);
\draw  [fill=blue!50](10.375,10.25) circle (0.5cm);
\draw [ rounded corners = 8.4] (9.625,13.25) rectangle (11.125,9.75);

\draw  [fill=green!50](8.625,10.25) circle (0.5cm);
\draw  [fill=green!50](8.625,11.5) circle (0.5cm);
\draw [ rounded corners = 8.4] (7.875,12) rectangle (9.375,9.75);

\draw  [fill=green!50](6.875,10.25) circle (0.5cm);
\draw  [fill=green!50](6.875,11.5) circle (0.5cm);
\draw [ rounded corners = 8.4] (6.125,12) rectangle (7.625,9.75);

\draw  [fill=blue!50](5.125,12.75) circle (0.5cm);
\draw  [fill=blue!50](5.125,11.5) circle (0.5cm);
\draw  [fill=blue!50](5.125,10.25) circle (0.5cm);
\draw [ rounded corners = 8.4] (4.375,13.25) rectangle (5.875,9.75);
\end{tikzpicture}
    \caption{A multiserver-job model, forming a special case of the generalized switch. Green jobs require 2 servers, and blue jobs require 3 servers, out of 7 total servers.}
    \label{fig:msj-example}
\end{figure}
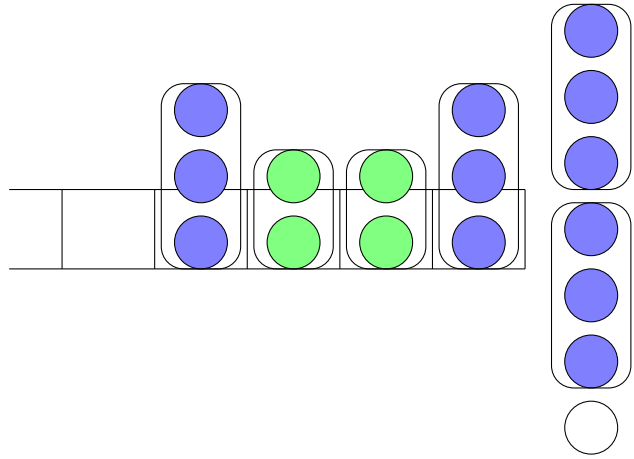

We call a mapping ``valid" when all service options in the original system have mapped service rate less than or equal to 1 in the resource pooled system, which is the capacity of the resource pooled system. However, there might be many valid mappings. For example, consider an MSJ system with 7 servers and two job classes shown in \cref{fig:msj-example}: green class-1 jobs require 2 servers and blue class-2 jobs require 3 servers. Two examples of valid mappings are as follows: In mapping A, a class-1 job with duration $d_1$ is mapped to a duration of $2d_1/7$ in the resource pooled system, and a class-2 job with duration $d_2$ to a duration of $3d_2/7$ in the resource-pooled system. Since the service rate in the resource pooled system is equal to $\frac{1}{7}$ times the number of servers occupied in the original system and the original system only has 7 servers, this is never more than 1. In mapping B, a class-1 job with duration $d_1$ is mapped to a duration of $d_1/3$ in the resource-pooled system and a class-2 job with duration $d_2$ to a duration of $d_2/3$ in the resource-pooled system. Since the service rate in the resource pooled system is equal to $\frac{1}{3}$ times the number jobs in service in the original system and the original system can never serve more than $\lfloor \frac{7}{2} \rfloor = 3$ jobs, this is never more than 1. To define a resource-pooled system, we choose a mapping up front and fix it for all service options.

We can generalize and make explicit the concept of a mapping, and of a valid mapping:
\begin{definition}
    We define a \emph{conversion mapping vector}, $\mathbf{k}$, so that a job of class $i$ with a duration $d\sim D_i$ in the original system has duration $k_id$ in the resource-pooled system. We will refer to the resource-pooled system under conversion multiplier vector $\mathbf{k}$ as the $\mathbf{k}$-mapped resource-pooled system. Moreover, we define a conversion multiplier vector $\mathbf{k}$ to be \emph{valid} if for all $\mathbf{r}\in\mathcal{R}$, $\langle \mathbf{k}, \mathbf{r}\rangle \leq 1$.
\end{definition}

We now prove that the mean response time under any policy in the original system can be matched by an appropriately chosen policy in any valid $\mathbf{k}$-mapped resource-pooled system.
In \cref{sec:converging-bounds}, we discuss how to find a valid conversion mapping vector which gives a tight lower bound on the mean response time of the optimal policy.

\begin{theorem}
    \label{thm:rp-lower-bound}
Let $\mathbf{k}=(k_1,\ldots,k_n)$ be a \emph{valid} conversion multiplier vector. Then for any scheduling policy $\pi$ in the original system, there exists a scheduling policy $\pi'$ in the $\mathbf{k}$-mapped resource-pooled system such that the mean response time is identical in both systems: $\E[T^\pi]$ in the original system matches $\E[T^{\pi'-RP-\mathbf{k}}]$ in the $k$-mapped resource-pooled system.

As a result, for any scheduling policy $\pi$ in the original system, $\E[T^\pi] \ge \E[T^{SRPT-RP-\mathbf{k}}]$.
\end{theorem}
\begin{proof}{\textit{Proof.}}
We use a sample path coupling argument. 
At time $t=0$, let both systems be empty.
We couple the arrival sequence as follows:
Jobs arrive to both systems according to a global Poisson process with rate $\lambda$.
Whenever the Poisson process increments,
a job arrives to each system.
The duration $d$ and class $i$ of the arriving job in the original system are sampled i.i.d. from the joint class-duration distribution, $(d, i) \sim (D, I)$.
In the $k$-mapped resource-pooled system, the job which arrives at the same moment has duration
$k_i d$.
We call the jobs that arrive at the same point in time ``coupled jobs''.

We now specify the scheduling policy $\pi'$ in the $\mathbf{k}$-mapped resource-pooled system,
based on the policy $\pi$ in the original system.
Policy $\pi'$ is a generalized-processor-sharing policy,
meaning that at any given point in time, it serves several jobs at fractional rates,
where those rates sum to at most 1.

At each point in time $t$, for every job that policy $\pi$ is serving at time $t$ in the original system,
the policy $\pi'$ serves the coupled job in the $\mathbf{k}$-mapped resource-pooled system: the job which arrived at the same time.
If a job of class $i$ is served by $\pi$ in the original system,
then the coupled job is served by $\pi'$ at rate $k_i$.

Note that because vector $\mathbf{k}$ is valid, the total of these service rates is at at most 1,
under any service option $r$ in the original system.
Thus, $\pi'$ is a valid generalized-processor-sharing policy.

With policy $\pi'$ defined in this fashion,
if a class-$i$ job in the original system under policy $\pi$ has received
an amount $a$ of time in service,
then the coupled job in the $\mathbf{k}$-mapped resource-pooled system
must have received $k_ia$ service, being served at rate $k_i$ for $a$ time.
In particular, at the moment the job in the original system
completes, it must have been in service for $d$ duration,
and the coupled job must have received $k_i d$ service,
resulting in the coupled job completing at the same time.

As a result, under this coupling and under the policies $\pi$ and $\pi'$, every job experiences the same response time as its coupled job in the other system, which then implies that the mean response time is identical in both systems.
Moreover, because we know that in a single-server queue, Shortest-Remaining-Processing-Time (SRPT) minimizes the mean response time over all scheduling policies, including generalized-processor-sharing policies, we must have for any policy $\pi$ in the original generalized switch:    
\begin{align*}
\forall \pi, \forall \text{ valid } \mathbf{k}, \E[T^\pi] &\ge \E[T^{SRPT-RP-\mathbf{k}}].
\end{align*}
\hfill\Halmos
\end{proof}

\subsection{Tight Lower Bounds in Heavy Traffic}
\label{sec:converging-bounds}

\cref{thm:rp-lower-bound} proves that every valid conversion multiplier vector $\mathbf{k}$ leads to a lower bound on the mean response time in the original system, for all generalized switch scheduling policies. 
Now, we must find a vector which gives a tight lower bound on the optimal policy. As we will see, while there are many possible lower bounds, selecting a lower bound that helps us prove heavy-traffic optimality is nontrivial.

\subsubsection{Motivation and Example}
To obtain a tight lower bound,
as the original system approaches heavy traffic,
the $\mathbf{k}$-mapped resource-pooled system must also approach heavy traffic.

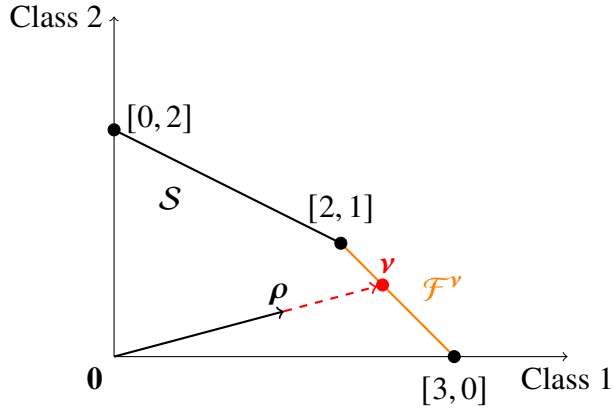
\begin{figure}
\label{fig: HT-trajectory}
    \begin{center}
        \begin{tikzpicture}[scale=1.5]

    \definecolor{lightmagenta}{rgb}{1.0, 0.8, 0.9}
    
    \draw[->] (0,0) -- (4,0) node[anchor=north] {Class 1};  
    \draw[->] (0,0) -- (0,3) node[anchor=east] {Class 2}; 

    \node[anchor=north east] at (0,0) {$\mathbf{0}$};


    
    
    \draw[->, thick, black] (0,0) -- (1.5,0.4);
    \draw[->, dashed, thick, red] (1.5, 0.4) -- (2.32,0.62);
    \node at (1.45, 0.55) {$\boldsymbol{\rho}$};  
    


    

    \node at (0.5, 1.4) {$\mathcal{S}$};

    \node at (3, -0.3) {$[3,0]$};

    \node at (2, 1.3) {$[2,1]$};

    \node at (0.4, 2.1) {$[0,2]$};

    \draw[thick, orange] (3, 0) -- (2, 1);
    \draw[thick] (2, 1) -- (0, 2);

    \filldraw (3, 0) circle (1.5pt);
    \filldraw (2,1) circle (1.5pt);
    \filldraw (0,2) circle (1.5pt);

    \filldraw[red] (2.368, 0.632) circle (1.5pt);

    \node[red] at (2.41, 0.8) {$\boldsymbol{\nu}$};
    

    \node[orange] at (2.9, 0.6) {$\mathcal{F}^{\boldsymbol{\nu}}$};
\end{tikzpicture}
        \caption{Heavy-traffic trajectory of the load vector $\boldsymbol{\rho}$ to an interior point $\boldsymbol{\nu}$ on the facet defined by $[3, 0]$ and $[2,1]$ in an MSJ system introduced in \Cref{fig:msj-example}}
    \end{center}
\end{figure}

There are many ways the original system can approach heavy traffic: one or a combination of different classes can contribute most of the load of the system. Different heavy-traffic regimes have different tight resource-pooled lower bounds. For example, consider the 7-server system shown in \cref{fig:msj-example} and the mappings A and B discussed in \cref{sec:resource-pooled-lower-bound}.
Suppose that the vast majority of the load consists of two server (class-1) jobs (\cref{fig: HT-trajectory}).
Then the system is typically running 3 class-1 jobs at a time.
Under mapping A, the A-mapped resource-pooled system only uses $6/7$ of the total capacity
in this common case, because $A_1 = 2/7$.
As a result, the A-mapped resource pooled does not approach the heavy traffic system when the original system approaches heavy traffic.
Thus, the lower bound of $E[T^{SRPT-RP-A}]$ is a poor lower bound in this heavy-traffic scenario -- it would be a better bound if more class-2 jobs were present.
However, under mapping B,
the B-mapped resource-pooled system uses its full capacity
when 3 class-1 jobs are in service ($\mathbf{r}=[3,0]$), because $B_1 = 1/3$.
Similarly, the B-mapped resource-pooled system also uses its full capacity
when 2 class-1 jobs and 1 class-2 job are in service ($\mathbf{r}=[2,1]$), because $2 B_1+ B_2 = 1$.
As a result, the B-mapped resource-pooled system approaches heavy traffic when the original system approaches heavy traffic in this scenario when most load comes from class-1 jobs. In fact, one can show that the B vector is the only conversion multiplier with this property in this scenario.
The B vector gives the lower bound that we use in this setting.
In the next section, we will discuss how we select,
for a given heavy-traffic trajectory,
a conversion multiplier vector $\mathbf{k}$
which achieves this tightness property.

\subsubsection{Generalization}
\label{sec:LB-generalization}

We now generalize over general sets of service options $\mathcal{R}$.

Recall that the stability region $\mathcal{S}$ is the interior of the convex hull of the service options $\mathcal{R}$. In particular, $\mathcal{S}$ is a polytope.
The boundary of the stability region can therefore be decomposed into $(n-1)$-dimensional facets. Each facet $\mathcal{F}$ is associated with a set  $\mathcal{R}^{\mathcal{F}}$ of service options that mark the extreme points of this facet.

For each facet $\mathcal{F}$, we define a corresponding conversion multiplier vector $\mathbf{k}^{\mathcal{F}}$.
We define $\mathbf{k}^{\mathcal{F}}$ as the unique solution to the system of equations $\langle \mathbf{k}^{\mathcal{F}}, \mathbf{r}\rangle = 1$ for all $\mathbf{r}\in\mathcal{R}^{\mathcal{F}}$: the uniqueness comes from the fact that an $n-1$ dimensional facet is defined by $n$ linearly independent extreme points. Note that $\mathbf{k}^\mathcal{F}$ is valid: if there were an $r$ such that $\langle\mathbf{k}^{\mathcal{F}}, \mathbf{r}\rangle > 1$, then $\mathcal{F}$ would not be on the convex hull.

For a chosen mapping vector $\mathbf{k}$, we define the \emph{$\mathbf{k}$-mapped priority index}
of a class-$i$ job with duration $d$ to be the duration $k_id$ of the matched job in the $\mathbf{k}$-mapped resource-pooled system defined in \cref{thm:rp-lower-bound}.
The $\mathbf{k}$-mapped priority index random variable for of class-$i$ jobs is $Z_i=k_iD_i$.
Similarly, for a load vector $\boldsymbol{\rho}$ where $\rho_i=\lambda_i\E[D_i]$, we define the $\mathbf{k}$-mapped load vector to be $\boldsymbol{\rho}^{\mathcal{F}}_i=\lambda_i\E[Z_i]$. 

Now we specify the heavy traffic regime we focus on in this paper. Let the load limit point $\boldsymbol{\nu}$ be a point on the surface of the stability region $\mathcal{S}$. We consider a sequence of systems, indexed by a stability gap $\varepsilon$, with load vectors $\boldsymbol{\rho}^{(\varepsilon)}=(1-\varepsilon)\boldsymbol{\nu}$.
In this sequence of systems, we hold the duration distribution $D_i$ constant,
and allow the arrival rates $\lambda_i^{(\varepsilon)}$ to scale linearly with $(1-\varepsilon)$.
We define the limiting facet $\mathcal{F}^{\boldsymbol{\nu}}$ to be the facet that contains the load limit point $\boldsymbol{\nu}$.
We will assume that $\boldsymbol{\nu}$ is located within the interior of a single facet, so there is a unique limiting facet $\mathcal{F}^{\boldsymbol{\nu}}$. This assumption is called the Complete Resource Pooling assumption, which we discuss further in \cref{sec:assumptions}. Facet $\mathcal{F}^{\boldsymbol{\nu}}$ in turn induces a conversion multiplier vector $\mathbf{k}^{\mathcal{F}^{\boldsymbol{\nu}}}$ and a set of service vectors $\mathcal{R}^{\mathcal{F}^{\boldsymbol{\nu}}}$.
We therefore define the $\mathbf{k}^{\mathcal{F}^{\boldsymbol{\nu}}}$-mapped priority index, and the $\mathbf{k}^{\mathcal{F}^{\boldsymbol{\nu}}}$-mapped load vector $\boldsymbol{\rho}^{\mathcal{F}^{\boldsymbol{\nu}}}$.
The $\mathbf{k}^{\mathcal{F}^{\boldsymbol{\nu}}}$-resource-pooled system is our tight resource-pooled system.
We define the total system load to be $\rho^{\mathcal{F}^{\boldsymbol{\nu}}}=\|\boldsymbol{\rho}^{\mathcal{F}^{\boldsymbol{\nu}}}\|_1$.
The system can be stabilized only if $\rho^{\mathcal{F}^{\boldsymbol{\nu}}}<1$, which is a necessary condition to $\boldsymbol{\rho} \in \mathcal{S}$, our standard stability condition for this system.

From now on we will work primarily with $\mathbf{k}^{\mathcal{F}^{\boldsymbol{\nu}}}$-mapped priority index (i.e. duration in the $\mathbf{k}^{\mathcal{F}^{\boldsymbol{\nu}}}$-mapped resource pooled system) unless specified otherwise. Because we will mostly work with the limiting facet with a prefixed $\boldsymbol{\nu}$,
we will subsequently replace $\mathcal{F}^{\boldsymbol{\nu}}$ with $\boldsymbol{\nu}$ in the superscripts to ease the notation.

\subsection{Setting Assumptions (CRIB)}
\label{sec:assumptions} 

In this section, we specify the class of generalized switch models we focus on in this paper. These systems of interest satisfy the following conditions, which we abbreviate as CRIB:

\* \emph{Complete Resource Pooling}. We have defined the family of systems approaching heavy traffic in \Cref{sec:converging-bounds}. We will work with the case where the load limit point $\boldsymbol{\nu}$ is in the interior of the limiting facet $\mathcal{F}^{\boldsymbol{\nu}}$. This is known as the Complete Resource Pooling (CRP) condition in literature.
\* \emph{Independence}. Over the joint class-duration distribution $(I, D)$, which induces a joint class-priority-index distribution $(I, Z)$, we assume independence between job class $I$ and job priority index $Z$ under the mapping vector $\mathbf{k}^{\boldsymbol{\nu}}$ corresponding to the limiting facet. 
\* \emph{Boundedness}. The duration distribution $D$ is bounded away from 0 and $\infty$. That is, the support of $D$ is contained in some interval $[d_{\min}, d_{\max}]$ where $0<d_{\min}<d_{\max}<\infty$. Note that as a result, the job priority index distribution $Z$ under any given facet is likewise bounded between some $z_{\min}$ and $z_{\max}$.
\*/

While the CRP and boundedness conditions are relatively standard in the literature, we now discuss the independence condition further.



The independence assumption implies that the priority index distribution for class $i$, $Z_i$, is the same across all classes. In the 7-server system shown in \Cref{fig:msj-example}, the facet with vertices $[2,1]$ and $[3,0]$ has mapping vector B, namely $(1/3, 1/3)$. Under this mapping vector, a job's priority index is proportional to its duration and it follows that the job priority index and class are independent if $D_1 \sim D_2$. On the other hand, the facet with vertices $[0,2]$ and $[2,1]$ has a mapping vector of $(1/4, 1/2)$. Under this mapping vector, $D_1 \sim 2D_2$ implies the independence of $Z$ and $I$.

Despite focusing on the CRIB setting, our result is still far more general than the prior state of the art: No previous paper has shown optimal mean response time under any generalized switch setting.
Moreover, the CRIB setting does not remove the key challenge of the setting, namely that we must simultaneously prioritize serving small jobs while keeping sufficient balance among the job classes to keep all servers well-utilized.

\subsection{Handling non-saturated classes}
\label{sec:non-saturated}

Some facets in the stability region can be parallel to an axis, causing the mapping vector to contain zeros. 
For example, in an MSJ system with 11 servers, where jobs can have server needs 2 (class 1) or 3 (class 2), the stability region $\mathcal{S}$ includes a facet with vertices $[1, 3]$ and $[0, 3]$. The corresponding mapping vector is $[0, 1/3]$. When approaching a point in the interior of this facet in heavy traffic, class-1 jobs are not saturated: The system can serve 1 class-1 job for free without reducing service to class-2 jobs. Correspondingly, the class-1 mapping vector entry is 0. This happens because the corresponding facet is parallel to an axis and indicates the corresponding class of jobs is not in heavy traffic.
As a result, the mean response times of such classes of jobs are negligible in heavy traffic.

For the purpose of our heavy-traffic results, we remove all non-saturated classes of jobs, and consider only systems in which all classes are saturated. We will henceforth assume that all classes are saturated on the limiting facet and refer to this as the saturated facet assumption. This does not reduce the generality of our result, because reintroducing such classes of jobs with an arbitrary stabilizing scheduling policy among the policies on the limiting facet results in no change to our heavy traffic results.

\section{Our Policy: Shortest Equalizing Bucket (SEB)}
\label{sec:policy}

In this section, we discuss the challenges of devising an optimal policy and
the behavior of existing policies (\cref{sec:challenges}), the intuition behind our SEB policy (\cref{sec:policy_intuition}).
We formally define the SEB policy (\cref{sec:def-SEB}),
prove basic properties of the SEB policy (\cref{sec:seb-basic}),
and define general notation for SEB (\cref{sec:SEB_notation}).

\subsection{What are the essential qualities of an optimal policy?}
\label{sec:challenges}
Given our analysis of the lower bounding system (\Cref{sec:resource-pooled-lower-bound}), an optimal policy must mimic resource-pooled SRPT in order to achieve heavy-traffic optimality. One can observe two necessary prerequisites for minimizing mean response time from resource-pooled SRPT.
\begin{enumerate}[(1)]
    \item \textbf{Mapped service rate is 1 as frequently as possible.} Before we analyze the mean response time of any policy in steady state, we must make sure that the policy stabilizes the system. Keeping the mapped service rate at 1 as frequently as possible ensures that the service capacity is fully utilized, which is key to stabilizing the system when the load approaches the total capacity.
    \item \textbf{Prioritize small jobs.} To minimize mean response time, ideally, we would always prioritize jobs from smallest to largest (remaining) priority index, along the lines of SRPT in the resource pooled system. Prioritizing smaller jobs over larger jobs whenever possible is key to achieving optimal mean response time.
\end{enumerate}

MaxWeight falls short because it focuses entirely on keeping all servers busy \citep{stolyar2004maxweight}.
SRPT-$k$, ServerFilling-SRPT, and the idle-avoid $c\mu/m$ rule
all focus on settings where keeping all servers busy is either easy or unnecessary, and focus entirely on prioritizing small jobs \citep{grosof2018srpt,grosof2022optimal,zylchlinski2023managing}.
Note however that there has been no analysis of policies
that manage to keep all servers busy as effectively as MaxWeight,
while simultaneously prioritizing jobs based on duration.
Our SEB policy achieves both goals.

\subsection{Intuition for our policy}
\label{sec:policy_intuition}
To simplify the discussion, we initially discuss a setting in which there are only two priority indices of jobs $z_1 < z_2$, namely ``small''~$z_1$ and ``large''~$z_2$. We explain how to generalize the ideas to general bounded priority index distributions towards the end of this subsection.
Throughout the section, and throughout the remainder of the paper, ``priority index'' refers to $\mathbf{k}^{\boldsymbol{\nu}}$-mapped priority index.

The key idea behind our policy is to balance both the small jobs and the large jobs separately. Balance refers to keeping the work vector of small jobs, $\mathbf{w}_{\text{small}}$, and work vector of large jobs, $\mathbf{w}_{\text{large}}$, roughly parallel to a given load vector ${\boldsymbol{\rho}^{\boldsymbol{\nu}}}$. Here, work of a class of jobs is roughly the total remaining duration of all jobs in the lower bounding system. We will define work more carefully in \Cref{sec:def-SEB}. We define balance in this manner for two reasons: (1) Because ${\boldsymbol{\rho}^{\boldsymbol{\nu}}}$ is not axis-parallel (see \cref{sec:non-saturated}), keeping the work vector aligned with ${\boldsymbol{\rho}^{\boldsymbol{\nu}}}$ prevents excessive depletion of work in certain classes. If too much work is drained from a specific class, capacity might be wasted. (2) if a bucket does not receive any service, arrivals provide a drift in the direction of ${\boldsymbol{\rho}^{\boldsymbol{\nu}}}$, so the bucket becomes more balanced.

We start by selecting a preferred small service option, and a preferred large service option. Both options are chosen from those on the limiting facet, and are chosen to move $\mathbf{w}_{\text{small}}$ and $\mathbf{w}_{\text{large}}$ towards ${\boldsymbol{\rho}^{\boldsymbol{\nu}}}$.
We serve small jobs using their preferred service option if possible. 
Otherwise, we serve large jobs using their preferred service option if possible.
Otherwise we idle the system.
By maintaining our notion of balance, we will prove that each preferred service option will always be available as long as there are more than a few jobs in the corresponding priority index category, which is key to our optimality proof.

Note that our goal is not to squeeze out every drop of performance. Our goal is merely to achieve heavy traffic optimality. We therefore do not focus too much on the scheduling decisions when a bucket is near empty, as this is rare in heavy traffic. We take a simple option, skipping over that bucket, for ease of theoretical analysis. In \Cref{sec:simulation}, we present SEB-inspired policies that show excellent empirical performance under both moderate and heavy traffic.

Our intuition for small and large jobs generalizes to more priority indices of jobs. For an arbitrary bounded priority index distribution, we divide the support of the distribution into disjoint intervals which we call ``buckets''. An arriving job falls into one of the buckets according to its \emph{original} priority index and stays in the same bucket for the entirety of its time in system. Our policy maintains balance for all buckets separately,
and gives smaller buckets (i.e. buckets corresponding to smaller priority indices)
higher priority if their preferred service option can be fulfilled.  

Finally, we must decide how to choose which jobs go in which buckets. Here, we draw inspiration from the dispatching setting \citep{grosof2019load}, where it has been found that geometric buckets, where the ratio of the smallest priority index and largest priority index in a bucket is set to be a properly chosen constant $c$, leads to heavy-traffic optimality. We show that such geometric buckets lead to heavy-traffic optimality in our setting as well.

\subsection{Defining Our Policy}
\label{sec:def-SEB}
Our Shortest Equalizing Bucket (SEB) policy is defined as follows:

Preprocessing: The scheduler finds the limiting duration-based facet $\mathcal{F}^{\boldsymbol{\nu}}$ that contains ${\boldsymbol{\rho}^{\boldsymbol{\nu}}}$ and solves for the mapping coefficients, $\mathbf{k}$, from the systems of equations $\inner{\mathbf{k}}{\mathbf{r}}=1$ for all $\mathbf{r}\in\mathcal{R}^{\boldsymbol{\nu}}$. After converting durations to priority indices, the scheduler then defines priority index buckets based on the interval $[z_{\min}, z_{\max}]$ containing the support of the priority index distribution. Bucket $i$ is defined as an interval $[b_{i-1}, b_i)$ so that $b_i / b_{i-1} = c$ (we set $b_0 = z_{\min}$), where
\[
    c = 1 + \frac{1}{1 + \log\left(\frac{1}{1-\rho^{\boldsymbol{\nu}}}\right)}.
\]
Let $n_b$ denote the total number of buckets.
The above constant-ratio definition for $b_i$ determines $b_0$ through $b_{n_b-1}$. The scheduler sets the last bucket to be $[b_{n_b-1}, b_{n_b}]$, where we define $b_{n_b} = z_{\max}$. When a job arrives, it is added to the bucket containing its original priority index and stays in the same bucket until it is completed.

We now define SEB's online behavior.
The scheduler iterates through the priority index buckets in increasing order.

\begin{enumerate}
    \item For each bucket $i$, the scheduler first computes the bucket work vector $\mathbf{w}^{(i)}$, obtained by summing remaining priority indices of jobs in bucket $i$ by class\footnote{Our definition of work can be understood as the total remaining duration of jobs in the lower bounding system. Note further that work is limiting facet-dependent.}. Then it finds the ideal service rate vector $\mathbf{r}^{(i)*}$ among all service rates on the priority index-based limiting facet $\mathcal{F}^{\boldsymbol{\nu}}$: $
    \mathbf{r}^{(i)*} = \argmax_{\mathbf{r}\in\mathcal{R}^{\boldsymbol{\nu}}} \inner{\mathbf{w}^{(i)}_{\perp{\boldsymbol{\rho}^{\boldsymbol{\nu}}}}}{\mathbf{r}}$.
    The scheduler then checks if the corresponding service option $\widetilde{\mathbf{r}}^{(i)*}$ can be fulfilled using jobs in the bucket.
    \item If service option $\widetilde{\mathbf{r}}^{(i)*}$ can be fulfilled using jobs in bucket $i$, the scheduler places such jobs in service, and the scheduler stops iterating through the buckets. The scheduler pick jobs among those in a given class within the bucket in FCFS order.
    \item Otherwise, the scheduler moves on to bucket $i+1$. If no bucket can fulfill its preferred service options, all servers idle.
\end{enumerate}

The scheduler updates the jobs in service if there is an arrival, a departure, or if the preferred service option changes. It's worth noting that the preferred service option $\widetilde{\mathbf{r}}^{(i)*}$ might change even when there are no arrivals or departures. This happens when, after some service, the work vector arrives at a point where multiple service options become equally preferable. In this case, the scheduler will rapidly alternate between those service options, resulting in a weighted-processor-sharing behavior. This emergent processor-sharing behavior is common in scheduling policies, with the notable example of Least Attained Service (Foreground-Background) \citep{harchol2013performance}.

\subsection{Basic property of SEB}
\label{sec:seb-basic}

Before we analyze how the system behaves under SEB, we establish an important property of the system in \Cref{prop:service_drift}: For any workload vector $\mathbf{w}$ that is not parallel to ${\boldsymbol{\rho}^{\boldsymbol{\nu}}}$, the preferred service option always incurs a non-vanishing drift towards the load vector ${\boldsymbol{\rho}^{\boldsymbol{\nu}}}$.

\begin{proposition}
\label{prop:service_drift}
For any $\mathbf{w}$ such that $\mathbf{w}_{\perp{\boldsymbol{\rho}^{\boldsymbol{\nu}}}}\neq\mathbf{0}$, there exists an $\varepsilon_0>0$, independent of $\mathbf{w}$, such that
$
\max_{\mathbf{r}\in\mathcal{R}^{\boldsymbol{\nu}}}\frac{\inner{\mathbf{w}_{\perp{\boldsymbol{\rho}^{\boldsymbol{\nu}}}}}{\mathbf{r}}}{\|\mathbf{w}_{\perp{\boldsymbol{\rho}^{\boldsymbol{\nu}}}}\|} \geq \varepsilon_0
$.
\end{proposition}

We make use of this proposition in \Cref{sec:bucket_SSC}. To prove \Cref{prop:service_drift}, we first prove the following lemma.

\begin{lemma}
\label{lemma:inner_product_rate_vectors}
For any $\mathbf{x}\neq\mathbf{0}$ such that $\inner{\mathbf{x}}{{\boldsymbol{\rho}^{\boldsymbol{\nu}}}} = 0$, we have
$
    \max_{\mathbf{r}\in\mathcal{R}^{\boldsymbol{\nu}}} \inner{\mathbf{x}}{\mathbf{r}} > 0
$.
\end{lemma}
\begin{proof}{\textit{Proof.~}}
We first note that the CRP assumption (\Cref{sec:assumptions}) and the saturated facet assumption (\Cref{sec:non-saturated}) imply that the service vectors on the limiting facet $\mathcal{F}^{\boldsymbol{\nu}}$, $\mathcal{R}^{\boldsymbol{\nu}}=\{\mathbf{r}_1,\ldots,\mathbf{r}_n\}$, are of full rank. Consider the convex hull of $\mathcal{R}^{\boldsymbol{\nu}}$. Since ${\boldsymbol{\rho}^{\boldsymbol{\nu}}}$ is in the interior of this convex hull, there exist strictly positive scalars $\alpha_1,\ldots,\alpha_n$ such that ${\boldsymbol{\rho}^{\boldsymbol{\nu}}}=\sum_{i=1}^n\alpha_ir_i$ (see Exercise 3.1 in \citep{brondsted2012introduction}). Since $\inner{\mathbf{x}}{{\boldsymbol{\rho}^{\boldsymbol{\nu}}}} = 0$, for any nonzero $\mathbf{x}$, either there must exist $\mathbf{r}\in\mathcal{R}^{\boldsymbol{\nu}}$ such that $\inner{\mathbf{x}}{\mathbf{r}}>0$, or $\inner{\mathbf{x}}{r_i}=0$ for all $i$. The latter implies that $\mathbf{x}$ is linearly independent of all $r_i$, which is impossible since $\mathcal{R}^{\boldsymbol{\nu}}$ is of full rank.
\hfill\Halmos
\end{proof}

With \Cref{lemma:inner_product_rate_vectors} in hand, we are read to prove \Cref{prop:service_drift}.

\begin{proof}{\textit{Proof of \Cref{prop:service_drift}.~}}
If the claim is false, there exists a sequence $\{\mathbf{w}^{[n]}\}_{n=1}^\infty$ such that 
\[\lim_{n\to\infty}\,\max_{\mathbf{r}\in\mathcal{R}^{\mathcal{F}}}\frac{\inner{\mathbf{w}^{[n]}_{\perp{\boldsymbol{\rho}^{\boldsymbol{\nu}}}}}{\mathbf{r}}}{\|\mathbf{w}^{[n]}_{\perp{\boldsymbol{\rho}^{\boldsymbol{\nu}}}}\|}=0.\]
Since the set $\{\mathbf{x}:\|\mathbf{x}\|=1\}\cap\{\mathbf{x}:\inner{\mathbf{x}}{{\boldsymbol{\rho}^{\boldsymbol{\nu}}}}=0\}$ is compact, there exists a limit point $\mathbf{y}$ of the sequence $\{\frac{\mathbf{w}^{[n]}}{\|\mathbf{w}^{[n]}\|}\}_1^\infty$ in the compact set. It follows that $\max_{\mathbf{r}\in\mathcal{R}^{\mathcal{F}}}\inner{\mathbf{y}}{\mathbf{r}}=0$, contradicting \Cref{lemma:inner_product_rate_vectors}.
\hfill\Halmos
\end{proof}

\subsection{Further Notation}
\label{sec:SEB_notation}

Having redefined the concept of job priority index for our system (\Cref{sec:LB-generalization}), we now redefine idleness for our system, a notion essential to analyze the efficiency of any scheduling algorithm.

For any chosen priority index mapping vector $\mathbf{k}$, we define the $\mathbf{k}$-mapped inefficiency at time $t$, $\I(t)$, as $\I(t) = 1 - \inner{\mathbf{k}}{\mathbf{r}(t)}$ for any valid service option $\mathbf{r}$ chosen at time $t$ in the stability region. We also define similarly, at time $t$, the $x$-relevant inefficiency $\I_{\leq x}(t)$ for a $\mathbf{k}$-mapped priority index $x$ as 1 minus the $\mathbf{k}$-mapped service rate devoted to jobs with remaining priority indices less than $x$ at time $t$. For the rest of the paper, we will mostly be working with the stationary versions $\I$ and $\I_{\leq x}$, which depend on the scheduling policy $\pi$.

Under SEB, for $i=1,\ldots, n_b$, we define the $i$-relevant inefficiency of the first $i$ buckets, $\I_{\leq i}$, as $\I_{\leq i} = \I_{\leq b_i}$, where $b_i$ is the priority index upper bound of the $i$-th bucket, as defined in \Cref{sec:def-SEB}, and $\I_{\leq b_i}$ is the $b_i$-relevant inefficiency, as defined in \Cref{sec:LB-generalization}.
We define the total work in the first $i$ buckets as $W_{\leq i}$. We define the total load into the first $i$ buckets as $\rho^{\boldsymbol{\nu}}_{\leq i}$. We define $W_{\leq i}^\mgone$ as the work in the $\mathbf{k}^{\boldsymbol{\nu}}$-mapped resource-pooled M/G/1 queue with job priority index distribution equal to $[Z \mid Z \le b_i]$ and load $\rho^{\boldsymbol{\nu}}_{\leq i}$. These definitions will be used when we use the Work Decomposition Law (\Cref{prop:WDL}) to study the response time of a job in the generalized switch.

When referring to bucket-specific quantities, we will use superscripts for the bucket number and the subscripts for the class number. For instance, $w_j^{(i)}$ denotes the work of class $j$ jobs in bucket $i$ and $S^{(i)}$ denotes the priority indices of jobs falling into bucket $i$.

\section{Main Results and Roadmap}
\label{sec:results}

Our main result in this paper is to prove the heavy-traffic optimality of our Smallest Equalizing Bucket (SEB) policy under the CRIB assumption (\Cref{sec:assumptions}), the first heavy-traffic optimality result in the generalized switch setting. We do so by comparing against the $\mathbf{k}^{\boldsymbol{\nu}}$-mapped resource-pooled M/G/1 system, with durations equal to the priority indices of the jobs, as defined in \cref{sec:LB-generalization}.

We prove that in the heavy traffic limit, the mean response time of SEB converges to that of the resource-pooled Shortest Remaining Processing Time (SRPT-1) policy, and that the SRPT-1 policy is a lower bound on the optimal policy in this setting:

\restatably\begin{theorem}
\label{thm:heavy-traffic_optimality}
The SEB policy is heavy-traffic optimal, under the assumptions in \cref{sec:assumptions,sec:non-saturated}:
For any $\boldsymbol{\nu}$ on the capacity region
that is in the interior of a facet, our SEB policy converges to optimal as
the load vector ${\boldsymbol{\rho}}$ converges to $\boldsymbol{\nu}$.
\[
    \lim_{{\boldsymbol{\rho}}\to\boldsymbol{\nu}} \frac{\E[T^\SEB]}{\E[T^\SRPT]}=\lim_{{\boldsymbol{\rho}}\to\boldsymbol{\nu}} \frac{\E[T^\SEB]}{\E[T^\OPT]} = 1,
\]
where $T^{\SRPT}$ is the response time under resource-pooled SRPT.
\end{theorem}
We discuss our proof structure in \cref{sec:roadmap},
and prove the theorem in \cref{sec:proof-optimality}.

Our optimality proof relies critically on our heavy traffic characterization of the mean response of the SEB policy:

\restatably\begin{theorem}
\label{thm:response_time_bound_PSJF}
Under SEB, in the heavy traffic limit and under the assumptions in \cref{sec:assumptions}, the
mean response time has the following asymptotic behavior:
\[
    \label{eq:MSJ_RT_order_bound}
    \E[T^\SEB] \leq c \cdot \E\left[T^{\PSJF}\right] +\Theta\left( \log^2\left(\frac{1}{1-\rho^{\boldsymbol{\nu}}}\right) \right)
\]
as ${\boldsymbol{\rho}}\to\boldsymbol{\nu}$, where $T^{\PSJF}$ is the response time under resource-pooled Preemptive-Shortest-Job-First (PSJF-1),
and where $c$ is the bucket width multiplier defined in \Cref{sec:policy}.
\end{theorem}

We note that under the boundedness assumption on the priority indices, one can show that under heavy traffic, PSJF-1 is also heavy-traffic mean response time optimal.

\Cref{thm:response_time_bound_PSJF} is the focus of the majority of the technical section of this paper, and we discuss the proof structure in detail in \cref{sec:roadmap}.
We prove the theorem in \cref{sec:response_time_analysis}.

\subsection{Roadmap to Mean Response Time Optimality}
\label{sec:roadmap}
We first prove that the system is throughput optimal (i.e. stable), for all $\rho^{\boldsymbol{\nu}}<1$, under SEB (\Cref{thm:stability}). We then establish an upper bound on mean response time under SEB for any load $\rho^{\boldsymbol{\nu}}<1$. A job's response time under policy $\pi$ can be written as
$
T^\pi = T^\pi_{\text{wait}}+T^\pi_{\text{res}}
$,
where $T^\pi_{\text{wait}}$ (wait time) is the time between when a job arrives and when it first receives any service and $T^\pi_{\text{res}}$ (residence time) is the time between when a job first receives service and when it leaves the system. In our system, bounding $\E[T^\pi_{\text{res}}]$ is a straightforward application of Little's law (\Cref{lem:residence_time_bound}), whereas bounding $\E[T^\pi_{\text{wait}}]$ is much more complicated.
Note also that $\E[T^\pi_{\text{wait}}]$ dominates under heavy traffic.

One of the key tools we use to bound $\E[T^\pi_{\text{wait}}]$ is the Work Integral Number Equality (WINE) technique \citep{righter1990extremal, scully2020gittins, banerjee2022heavy, scully2022new}, which converts the analysis of mean response time to the analysis of mean relevant work in system. Although the original theorem is stated in the context of an M/G/k queue, it is sufficiently general for the generalized switch model, with the proof almost unaltered. We include a proof in \Cref{app:WINE} for completeness.

\restatably\begin{proposition}[\citet{scully2022new} Theorem 15.3]\label{prop:WINE}
For an arbitrary scheduling policy $\pi$ that stabilizes the generalized switch,
\[
    \E[T^\pi] = \frac{1}{\lambda} \int_{0}^{\infty} \frac{\E[W_{\remduration \leq x}^\pi]}{x^2} \, dx,
\]
where $\E[W_{\remduration \leq x}^\pi]$ is the total remaining priority index of all jobs with remaining priority indices no more than $x$.
\end{proposition}

When we apply \Cref{prop:WINE} to bound the waiting time, we apply it to the subsystem consisting of jobs that have not yet received any service, in \cref{lemma:WINE_for_wait_time}.
Note that the set of jobs which have not yet received service and which have remaining priority indices $\le x$ is a subset of the jobs in the system with original priority indices $\le x$. Thus, for any given threshold $x$, it suffices to bound $\E[W_{\origduration \leq x}^\pi]$.

Recall that SEB assigns a job to a bucket based on the original priority index of the job. As a result, to bound $\E[W_{\origduration \leq x}^\pi]$, we focus on total work in buckets $\leq i$, where bucket $i$ is the bucket that a job of priority index $x$ is placed in. Our strategy is to apply the Work Decomposition Law \citep{scully2020gittins, scully2022new} to the first $i$ buckets. We state the Work Decomposition Law here in a way that is specialized to our setting.

\begin{proposition}[\citet{scully2022new} Theorem 8.2]\label{prop:WDL}
For any service policy $\pi$ that stabilizes the system,
\begin{align*}
    \E[W_{\leq i}^{\pi}] - \E[W_{\leq i}^{\mgone}] = \frac{\E[\mathcal{I}_{\leq i}W^{\pi}_{\leq i}]}{1-\rho^{\boldsymbol{\nu}}_{\leq i}}
\end{align*}
for any $i=1,\ldots, n_b$.
\end{proposition}

Although the Work Decomposition Law does not originally address duration–to-priority index conversion, all quantities in \Cref{prop:WDL} are defined on the $\mathbf{k}$-mapped system, making \Cref{prop:WDL} directly applicable in this paper.

The key term in \Cref{prop:WDL} is $\E[\mathcal{I}_{\leq i}W^{\pi}_{\leq i}]$, which we subsequently refer to as the \emph{waste}.
Waste intuitively measures how much work is not being processed by the servers. We will analyze the waste bucket by bucket. If there's little work $W^{\pi}_{\le i}$, the waste is small. The challenge is to show that the waste is small in the presence of substantial work.

Note that waste happens if we can't serve jobs in a bucket, so we need to show that when there is a lot of work in a bucket, it is very rare that the bucket is ineligible to receive service. To formalize this, we first define a cone around the load vector, then we consider two possible cases:
\begin{enumerate}
    \item When the bucket work vector is outside the cone, the vector on average drifts towards the cone. Thus, we show a state-space collapse result (\Cref{thm:bucket_SSC}). Using \Cref{thm:bucket_SSC} we show that the probability that the work vector is so far from the cone that the bucket is ineligible for service is small (\Cref{thm:waste_bound}).
    \item When the bucket work vector is inside the cone, we show that the bucket must be eligible for service if there is a lot of work (\Cref{thm:eligibility_work_amount}).
\end{enumerate}

Combining all bounds together, we obtain a bound on the mean response time relative to mean response time in the resource-pooled system (\Cref{thm:response_time_bound_PSJF}, proven in \cref{sec:response_time_analysis}). Heavy-traffic optimality then follows from this bound, and the fact that the resource-pooled system forms a lower bound on optimal mean response time, as we show in \cref{sec:proof-optimality}.

\section{SEB Analysis: Stability and Balancing Each Bucket}
\label{sec:analysis_1}

\subsection{Defining the Cone-based State-space Collapse}
\label{sec:cone_SSC}

    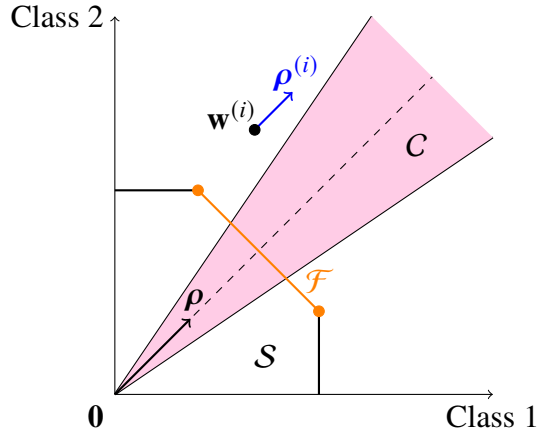
\begin{figure}[t]
        \begin{center}
            \begin{tikzpicture}[scale=1]

    \definecolor{lightmagenta}{rgb}{1.0, 0.8, 0.9}
    
    \draw[->] (0,0) -- (5,0) node[anchor=north] {Class 1};  
    \draw[->] (0,0) -- (0,5) node[anchor=east] {Class 2}; 

    \node[anchor=north east] at (0,0) {$\mathbf{0}$};

    \draw[thick] (0,0) -- (3.4, 5); 

    \draw[thick] (0,0) -- (5, 3.4);  
    
    \fill[lightmagenta] 
        (0,0) -- (3.4, 5) -- (5, 3.4) -- cycle;
    
    \draw[->, thick, black] (0,0) -- (1,1);
    \node at (1.05, 1.25) {$\boldsymbol{\rho}$};  
    
    \draw[dashed] (0,0) -- (4.2, 4.2);

    \draw[->, thick, blue] (1.85, 3.5) -- (2.35, 4);
    \node at (2.4, 4.25) {\textcolor{blue}{$\boldsymbol{\rho}^{(i)}$}};

    \filldraw (1.85, 3.5) circle (2pt);
    \node at (1.55,3.7) {$\mathbf{w}^{(i)}$};
    
    \node at (4, 3.3) {$\mathcal{C}$};

    \draw[thick, orange] (1.1, 2.7) -- (2.7, 1.1);
    \draw[thick] (0, 2.7) -- (1.1, 2.7);
    \draw[thick] (2.7, 0) -- (2.7, 1.1);
    \node at (2, 0.5) {$\mathcal{S}$};

    \filldraw[orange] (1.1, 2.7) circle (2pt);
    
    \filldraw[orange] (2.7, 1.1) circle (2pt);

    \node at (2.7, 1.5) {\textcolor{orange}{$\mathcal{F}$}};
\end{tikzpicture}
        \end{center}
        \caption{Drift of bucket-$i$ work vector in a two-class system, when service is not occurring.}
        \label{fig:drift}
    \end{figure}

In this section, we give definitions relevant for cone-based state space collapse,  which is a key idea in our optimality proof.

Recall that in \Cref{sec:policy_intuition}, we explained how SEB keeps all priority index buckets balanced by keeping the work vector of each bucket roughly parallel to the load vector ${\boldsymbol{\rho}^{\boldsymbol{\nu}}}$. To formalize this intuition, we aim to establish a state-space collapse result for the bucket work vectors $\mathbf{w}^{(i)}$. However, we note that the work vectors do not collapse to the load vector itself. This is because when a bucket receives no service, drift due to arrivals moves the work vector parallel to ${\boldsymbol{\rho}^{\boldsymbol{\nu}}}$, not towards it.
Note that the conditional load ${\boldsymbol{\rho}^{\boldsymbol{\nu}}}^{(i)}$ arriving to any bucket $i$ is parallel to ${\boldsymbol{\rho}^{\boldsymbol{\nu}}}$, due to our assumption of independence between job class and job priority index (\cref{sec:assumptions}).

The fact that the arriving load vector ${\boldsymbol{\rho}^{\boldsymbol{\nu}}}^{(i)}$ is parallel to ${\boldsymbol{\rho}^{\boldsymbol{\nu}}}$ is useful: it implies that the drift will move the work vector towards a \emph{cone} around ${\boldsymbol{\rho}^{\boldsymbol{\nu}}}$, as illustrated in \cref{fig:drift}. We will show that whether or not the bucket receives service, the system drifts towards the cone.
The cone we consider is
\begin{align*}
    \mathcal{C} = \left\{\mathbf{w}\in\mathbb{R}^{n_c}: \frac{\|\mathbf{w}_{\parallel{\boldsymbol{\rho}^{\boldsymbol{\nu}}}}\|_2}{\|\mathbf{w}\|_2} \geq \cos{\varphi}\right\},
\end{align*}
where $\varphi>0$ is chosen so that the following conditions are satisfied:
\* $\varphi$ is small enough that $\mathcal{C}$ is contained in the interior of the convex cone generated by service rate vectors in $\mathcal{R}^{\boldsymbol{\nu}}$. 
\* $\varphi$ is small compared to the constant $\varepsilon_0$ in \Cref{prop:service_drift}:
    $\left(\max_{\mathbf{r}\in\mathcal{R}^{\boldsymbol{\nu}}}\frac{\inner{{\boldsymbol{\rho}^{\boldsymbol{\nu}}}}{\mathbf{r}}}{\|{\boldsymbol{\rho}^{\boldsymbol{\nu}}}\|_2}\right)\tan\varphi<\varepsilon_0$.
\* $\varphi$ is small compared to the lowest-load class in the system: $
    \tan\varphi < \frac{{\boldsymbol{\rho}^{\boldsymbol{\nu}}}_{\min}}{\|{\boldsymbol{\rho}^{\boldsymbol{\nu}}}\|_2} 
    $, where we define ${\boldsymbol{\rho}^{\boldsymbol{\nu}}}_{\min}=\min_i\rho_i$.
\*/
Note that we use the same cone~$\mathcal{C}$ for all buckets $i$.
That is, we show in \cref{sec:bucket_SSC} that each bucket work vector $\mathbf{w}^{(i)}$ converges to the same cone~$\mathcal{C}$.

\subsection{Stability}
\label{sec:proof-stability}

\restatably\begin{theorem}
\label{thm:stability}
The system is stable under SEB for any $\rho^{\boldsymbol{\nu}}<1$.
\end{theorem}
Full proof deferred to \cref{app:stability}.
\begin{proof}{\textit{Proof outline}} Our main tool is the continuous-time Foster-Lypunov theorem in a general state space (Theorem 6 in \citep{meyn1993survey}). The key idea is to find a nonnegative Lyapunov function $V$ that has bounded drift on a compact set and negative drift outside that set. Our choices of $V$ rely on the following observations:
\begin{itemize}
    \item When bucket $i$ does not receive any service, the average drift of $\mathbf{w}^{(i)}$ is parallel to ${\boldsymbol{\rho}^{\boldsymbol{\nu}}}$, due to arrivals. As a result, job arrivals push $\mathbf{w}^{(i)}$ towards the cone $\mathcal{C}$ whenever $\mathbf{w}^{(i)}\not\in\mathcal{C}$.
    \item Since job priority index in bucket $i$ is upper bounded by $b_i$, if there is more than $n_sb_{i}$ amount of work of each job class in bucket $i$, then there are enough jobs to fulfill any service option. We show that if, for some bucket $i$, $\mathbf{w}^{(i)}\in\mathcal{C}$ and if $\|\mathbf{w}^{(i)}\|_1$ is sufficiently large, then bucket $i$ is eligible for service. Our policy then ensures that some bucket receives service at full capacity and the total work in the system decreases at rate $1-\rho^{\boldsymbol{\nu}}$ in this scenario.
\end{itemize}
As a result of these observations, our Lyapunov function is a weighted sum of two terms: a term characterizing the distance to cone $\mathcal{C}$, based on the $H_i$ function defined in \cref{sec:bucket_SSC}, and the total work in the system.
\hfill\Halmos
\end{proof}

\subsection{Bucket State-space Collapse}
\label{sec:bucket_SSC}

Our goal in this section is to prove that the workload $\mathbf{w}^{(i)}$ of each bucket $i$ collapses to the cone $\mathcal{C}$ described in \cref{sec:cone_SSC}.
We first prove a generic continuous-time state-space collapse theorem, \cref{thm:SSC}, which may be of independent interest.
We them apply this theorem to proving our bucket-based state space collapse, in \cref{thm:bucket_SSC}.

We now prove a state-space collapse theorem that can be viewed as a continuous-time equivalent of the classic drift-based discrete-time state-space collapse theorem first introduced in \citet{eryilmaz2012asymptotically}. The challenge here is to do drift analysis in continuous-time. To this end, we employ the Rate Conservation Law \citep{miyazawa1994rate} to a properly chosen test function. 

\restatably\begin{theorem}\label{thm:SSC}
Let $\mathbf{W}(t)\in\mathbb{R}^n_+$ be the workload process of a queueing system with Poisson arrivals with rate $\lambda>0$ with a stationary distribution $\mathbf{W}$. Suppose $V:\mathbb{R}_+^n\to\mathbb{R}_+$ is a differentiable nonnegative-valued function and the following conditions are satisfied:
\begin{enumerate}[(i)]
    \item There exist $\alpha>0$, $\beta>0$, and $K<\infty$ such that for any $\mathbf{W}(t)=\mathbf{w}$,
    \[\mathcal{G}V(\mathbf{w})&:=D_tV(\mathbf{w})+\lambda\E[\Delta V(\mathbf{w})]\leq -\alpha+\beta\cdot\mathbf{1}(V(\mathbf{w})\leq K)\]
    where $\mathcal{G}$ is the infinitesimal generator of the workload process, $D_tV(\mathbf{w})$ is the change in $V(\mathbf{w})$ due to service, and $\Delta V(\mathbf{w})=V(\mathbf{w}_+)-V(\mathbf{w})$ is the change in $V$ immediately after an arrival at state $\mathbf{W}(t)=\mathbf{w}$.
    \item There exist $\theta>0$ and $D<\infty$ such that for all $\mathbf{W}(t)=\mathbf{w}$, $\E\left[e^{\theta|\Delta V(\mathbf{w})|}\right]<D$.
\end{enumerate}
Then for any $0<\eta<\min\left\{\frac{\alpha\theta^2}{\lambda D},\theta\right\}$, we have
\[\E\left[e^{\eta V(\mathbf{W})}\right]\leq\frac{\theta^2\beta e^{\eta K}}{\theta^2\alpha-\lambda\eta D}<\infty\]
\end{theorem}

Full proof deferred to \cref{app:SSC}.

\begin{proof}{\textit{Proof outline}} We apply the Rate Conservation Law \citep{miyazawa1994rate} to $e^{\eta(V(\mathbf{W})\wedge n)}$ for some $0<\eta<\min\left\{\frac{\alpha\theta^2}{\lambda D},\theta\right\}$
and for a fixed positive integer $n$ such that $\E\left[e^{\eta(V(\mathbf{W})\wedge n)}\right]<\infty$.

Then we bound the resulting terms using conditions (i) and (ii) to obtain a bound on $\E\left[e^{\eta(V(\mathbf{W})\wedge n)}\right]$, using straightforward algebraic manipulations. Finally, we send $n\to\infty$ to obtain the desired result.
\hfill\Halmos
\end{proof}

To demonstrate state-space collapse in the sense of \cref{thm:SSC}, we consider the Lyapunov function $
H_i(\mathbf{w})=(\|\mathbf{w}^{(i)}_{\perp{\boldsymbol{\rho}^{\boldsymbol{\nu}}}}\|_2-\|\mathbf{w}^{(i)}_{\parallel{\boldsymbol{\rho}^{\boldsymbol{\nu}}}}\|_2\tan\varphi)^+
$.
Note that $H_i(\mathbf{w}) = 0$ whenever the workload vector $\mathbf{w}^{(i)}$ is in the cone $\mathcal{C}$. Intuitively, $H_i$ measures the distance from the workload vector to $\mathcal{C}$.

\restatably\begin{theorem}
\label{thm:bucket_SSC}
For any bucket $i$, under our SEB policy, if the assumptions in \cref{sec:assumptions} are met, then we have the following state-space collapse:
If $H_i(\mathbf{w})=(\|\mathbf{w}^{(i)}_{\perp{\boldsymbol{\rho}^{\boldsymbol{\nu}}}}\|_2-\|\mathbf{w}^{(i)}_{\parallel{\boldsymbol{\rho}^{\boldsymbol{\nu}}}}\|_2\tan\varphi)^+$ and $\eta_i=\frac{\tan\varphi}{(1+\tan\varphi)^2}\frac{e^{-2}}{\sqrt{n_c}}\frac{b_{i-1}}{b_{i}^2}$, then
\[
\E\left[e^{\eta_iH_i(\mathbf{W})}\right]\leq \frac{5(1+\tan\varphi)}{\tan\varphi}\frac{b_i}{b_{i-1}},
\]
bounding the distance from $w_i$ to the cone $\mathcal{C}$.
\end{theorem}

Proof deferred to \cref{app:bucket_SSC}. The proof is a straightforward application of \cref{thm:SSC}, verifying conditions (i) and (ii) in this setting.

\section{SEB Analysis: Mean Response Time}
\label{sec:analysis_2}

\subsection{Bounding Waste}
\label{sec:waste_analysis}

In this section, we bound the expected waste $\E[\I_{\leq i} W_{\leq i}^\SEB]$ for an arbitrary bucket $i$.
Our bound is based on our cone-state-space-collapse (SSC) result, \cref{thm:bucket_SSC}.
We show in this section that whenever the workload vector is in the cone or near the cone,
waste must be small. From our cone-SSC result, this allows us to bound expected waste.
This bound on expected waste directly translates into a bound on expected work in each bucket, giving a bound on waiting time and response time.

\restatably\begin{theorem}\label{thm:eligibility_work_amount}
If $\mathbf{w}^{(i)}\in\mathcal{C}$ and if $\|\mathbf{w}^{(i)}\|_1$ satisfies
\[
\|\mathbf{w}^{(i)}\|_1\geq\frac{n_s\sqrt{n_c} b_{i}}{\frac{{\boldsymbol{\rho}^{\boldsymbol{\nu}}}_{\min}}{\|{\boldsymbol{\rho}^{\boldsymbol{\nu}}}\|_2}\cos\varphi-\sin\varphi}
\]
then any service option in $\mathcal{R}^{\mathcal{F}}$ can be fulfilled using jobs in bucket $i$.
\end{theorem}
Proof deferred to \Cref{app:eligibility_work_amount}.

We now proceed to bound expected waste $\E[\I_{\leq i}W_{\leq i}]$, which is the key technical result of the paper. To ease the notation, we define two constants that do not scale with load.
\[
\gamma&=\frac{n_s\sqrt{n_c} b_{n_b}}{\frac{{\boldsymbol{\rho}^{\boldsymbol{\nu}}}_{\min}}{\|{\boldsymbol{\rho}^{\boldsymbol{\nu}}}\|_2}\cos\varphi-\sin\varphi}\\
\tau&=\frac{\tan\varphi}{\sqrt{n_c}}\left(\frac{{\boldsymbol{\rho}^{\boldsymbol{\nu}}}_{\min}}{\|{\boldsymbol{\rho}^{\boldsymbol{\nu}}}\|_2}\cos\varphi-\sin\varphi\right)
\]

\restatably\begin{theorem}
\label{thm:waste_bound}
Under SEB,
\[
\E[\I_{\leq i}W_{\leq i}^\SEB]\leq \left(A_1+A_2+A_2\log\left(\frac{1}{1-\rho^{\boldsymbol{\nu}}_{\leq i}}\right)\right)(1-\rho^{\boldsymbol{\nu}}_{\leq i})\frac{c^i-1}{c-1}+A_2(1-\rho^{\boldsymbol{\nu}}_{\leq i})i 
\]
where
\[
A_1 &:= \frac{e^2\sqrt{n_c}(1+\tan\varphi)^2}{\tau\tan^2\varphi}c^3b_0\\
A_2 &:= c\left(\frac{5(1+\tan\varphi)}{\tau\tan\varphi}c+1\right)\max\left\{\frac{e^2\sqrt{n_c}(1+\tan\varphi)^2}{\tau\tan\varphi}b_0c, n_s\frac{b_0}{\tau}, \gamma\right\}
\]
\end{theorem}
\begin{proof}{\textit{Proof.}}
Let $B_{ij}>\gamma$ be numbers that we will later specify, then
\[
\E[\I_{\leq i}W_{\leq i}^\SEB]&=\sum_{j=1}^i\E[\I_{\leq i}\|\mathbf{W}^{(j)}\|_1]\\
&=\sum_{j=1}^i\Big(\E[\I_{\leq i}\|\mathbf{W}^{(j)}\|_1\mathbf{1}(\|\mathbf{W}^{(j)}\|_1>B_{ij})]+\E[\I_{\leq i}\|\mathbf{W}^{(j)}\|_1\mathbf{1}(\|\mathbf{W}^{(j)}\|_1\leq B_{ij})]\Big)\\
&\stackrel{(a)}{\leq} \sum_{j=1}^i\Big(\E[\I_{\leq j}\|\mathbf{W}^{(j)}\|_1\mathbf{1}(\|\mathbf{W}^{(j)}\|_1>B_{ij})]+\E[\I_{\leq i}\|\mathbf{W}^{(j)}\|_1\mathbf{1}(\|\mathbf{W}^{(j)}\|_1\leq B_{ij})]\Big)\\
&\stackrel{(b)}{\leq} \sum_{j=1}^i\underbrace{\E[\I_{\leq j}\|\mathbf{W}^{(j)}\|_1\mathbf{1}(\|\mathbf{W}^{(j)}\|_1>B_{ij})]}_{\mathcal{T}}+\sum_{j=1}^iB_{ij}(1-\rho^{\boldsymbol{\nu}}_{\leq i})
\]
where (a) follows from the fact that for any $j\leq i$, $\mathcal{I}_{\leq i}\leq \mathcal{I}_{\leq j}$ and (b) follows from $\E[\I_{\leq i}]=1-\rho^{\boldsymbol{\nu}}_{\leq i}$. It remains to bound $\mathcal{T}$. First note that since we will set $B_{ij}>\gamma$, according to \Cref{thm:eligibility_work_amount}, if the workload vector is in the cone $\mathcal{C}$ and the total work is large, there is no waste,
\[
\E[\I_{\leq j}\|\mathbf{W}^{(j)}\|_1\mathbf{1}(\|\mathbf{W}^{(j)}\|_1\geq B_{ij})\mathbf{1}(\mathbf{W}^{(j)}\in\mathcal{C})]=0
\]
because bucket $j$ is eligible for service. This implies that
\[
\mathcal{T}=\E[\I_{\leq j}\|\mathbf{W}^{(j)}\|_1\mathbf{1}(\|\mathbf{W}^{(j)}\|_1\geq B_{ij})\mathbf{1}(\mathbf{W}^{(j)}\not\in\mathcal{C})]
\]
Define $\mathcal{E}_j=\{\mathbf{w}^{(j)}\in\mathbb{R}^{n_c}: \mathbf{w}^{(j)}_k\geq n_sb_{i}\text{ for all $k$}\}$ to be the region in the workload space such that if $\mathbf{w}^{(j)}$ is in the region, the system can fully serve all service options using only jobs from bucket $j$.

We now bound $\mathcal{T}$.
\begin{align*}
\mathcal{T}\stackrel{(a)}{\leq}&\,\E[\mathbf{1}(\text{bucket $j$  ineligible})\|\mathbf{W}^{(j)}\|_1\mathbf{1}(\|\mathbf{W}^{(j)}\|_1\geq B_{ij})\mathbf{1}(\mathbf{W}^{(j)}\not\in\mathcal{C})]\\
\stackrel{(b)}{\leq}&\,\E[\mathbf{1}(\mathbf{W}^{(j)}\not\in\mathcal{E}_j)\|\mathbf{W}^{(j)}\|_1\mathbf{1}(\|\mathbf{W}^{(j)}\|_1\geq B_{ij})\mathbf{1}(\mathbf{W}^{(j)}\not\in\mathcal{C})]\\
\stackrel{(c)}{\leq}&\,B_{ij}\cdot\P(\{\mathbf{W}^{(j)}\not\in\mathcal{E}_j\cup\mathcal{C}\}\cap\{\|\mathbf{W}^{(j)}\|_1\geq B_{ij}\})+\\
&\,\int_{u=B_{ij}}^\infty \P(\{\mathbf{W}^{(j)}\not\in\mathcal{E}_j\cup\mathcal{C}\}\cap\{\|\mathbf{W}^{(j)}\|_1\geq u)\})\,du\\
\stackrel{(d)}{\leq}&\, B_{ij}\cdot\P(H_j(\mathbf{W})\geq \tau B_{ij}-n_sb_j)+\int_{u=B_{ij}}^\infty\P(H_j(\mathbf{W})\geq \tau u-n_sb_j)\,du\\
\stackrel{(e)}{\leq}&\, B_{ij}\frac{5(1+\tan\varphi)}{\tan\varphi}\frac{b_j}{b_{j-1}}e^{\eta_j(n_sb_j-\tau B_{ij})}+\int_{u=B_{ij}}^\infty \frac{5(1+\tan\varphi)}{\tan\varphi}\frac{b_j}{b_{j-1}}e^{\eta_j(n_sb_j-\tau B_{ij})}\,du\\
=&\,\left(B_{ij}+\frac{1}{\tau\eta_j}\right)\frac{5(1+\tan\varphi)}{\tan\varphi}\frac{b_j}{b_{j-1}}e^{\eta_j(n_sb_j-\tau B_{ij})}
\end{align*}
(a) follows from $\mathcal{I}_{\leq j}\leq \mathbf{1}(\text{bucket $j$  ineligible})$ as a result of the following observations:
\begin{itemize}
    \item If bucket $j$ is ineligible for service, then the inequality holds because $\I_{\leq j}\leq 1$.
    \item If bucket $j$ is eligible for service, then the inequality holds because $\I_{\leq j}=0$.
\end{itemize}
(b) follows from the observation that the set on which bucket $j$ is ineligible is a subset of $\mathcal{E}_j^c$. (c) follows from tail-sum formula. (d) follows from the following lemma, the proof of which is deferred to \Cref{app:H_lower_bound}, and (e) follows from the state-space collapse result in \Cref{thm:bucket_SSC}.

\restatably\begin{lemma}
\label{lem:H_lower_bound}
    For any $\mathbf{w}^{(j)}\not\in\mathcal{C}\cup\mathcal{E}_i$, $H_j(\mathbf{w}) = \|\mathbf{w}^{(j)}_{\perp{\boldsymbol{\rho}^{\boldsymbol{\nu}}}}\|_2-\|\mathbf{w}^{(j)}_{\parallel{\boldsymbol{\rho}^{\boldsymbol{\nu}}}}\|_2\tan\varphi \geq \tau \|\mathbf{w}^{(j)}\|_1 - n_s b_j$.
\end{lemma}

We note that $B_{ij}>\gamma$ can be arbitrary. We set $B_{ij}$ as follows for any $j$ such that $j\leq i$:
\[
    B_{ij}&=\frac{1}{\tau\eta_j}\log\left(\frac{1}{1-\rho^{\boldsymbol{\nu}}_{\leq i}}\right)+\frac{n_sb_{j}}{\tau}+\gamma
\]
which gives the following bound on $\mathcal{T}$:
\[
\mathcal{T}\leq\left(B_{ij}+\frac{1}{\tau\eta_j}\right)\frac{5(1+\tan\varphi)}{\tan\varphi}\frac{b_{j}}{b_{j-1}}(1-\rho^{\boldsymbol{\nu}}_{\leq i})
\]

Since $b_{j}/b_{j-1} = c$, $b_{j}=b_0c^{j}$ and $\eta_j=\frac{\tan\varphi}{(1+\tan\varphi)^2}\frac{e^{-2}}{\sqrt{n_c}}\frac{b_{j-1}}{b_{j}^2}$. We combine all bounds above to obtain
\[
\E[\I_{\leq i} W_{\leq i}^{\SEB}] \leq&\, \sum_{j=1}^i \left(B_{ij}+\frac{1}{\tau\eta_j}\right)\frac{5(1+\tan\varphi)}{\tan\varphi}\frac{b_{j}}{b_{j-1}}(1-\rho^{\boldsymbol{\nu}}_{\leq i}) + (1-\rho^{\boldsymbol{\nu}}_{\leq i})\sum_{j=1}^iB_{ij}\\
\leq&\, \frac{e^2\sqrt{n_c}(1+\tan\varphi)^2}{\tau\tan^2\varphi}c^2b_0(1-\rho^{\boldsymbol{\nu}}_{\leq i})\sum_{j=1}^ic^j+\\
&\,\left(\frac{5(1+\tan\varphi)}{\tau\tan\varphi}c+1\right)(1-\rho^{\boldsymbol{\nu}}_{\leq i})\left[\frac{e^2\sqrt{n_c}(1+\tan\varphi)^2}{\tau\tan\varphi}b_0c\log\left(\frac{1}{1-\rho^{\boldsymbol{\nu}}_{\leq i}}\right)\sum_{j=1}^ic^j +\right.\\
&\,\left.n_s\frac{b_0}{\tau}\sum_{j=1}^ic^j+i\gamma\right]
\]
Using $A_1$ and $A_2$ defined in the theorem statement, we obtain
\begin{align*}
\E[\I_{\leq i}W_{\leq i}^\SEB ]\leq \left(A_1+A_2\log\left(\frac{1}{1-\rho^{\boldsymbol{\nu}}_{\leq i}}\right)+A_2\right)(1-\rho^{\boldsymbol{\nu}}_{\leq i})\frac{c^i-1}{c-1}+A_2(1-\rho^{\boldsymbol{\nu}}_{\leq i})i.
\end{align*}
\hfill\Halmos
\end{proof}

\subsection{Bounding Mean Response Time}
\label{sec:response_time_analysis}
In this section, we bound the mean response time under SEB. We begin by bounding the mean residence time, which is a straightforward application of Little's law. 
\begin{lemma}\label{lem:residence_time_bound}
Under SEB,
\[
    \E[T^\SEB_{\text{res}}] \leq \frac{1}{\lambda} n_b n_s n_c
\]
\end{lemma}
\begin{proof}{\textit{Proof.}}
By Little's law, it suffices for us to bound the mean number of jobs in residence. Since we assumed FCFS for jobs in each class in each bucket, $n_s$ jobs per class per bucket have received service, and thus at most $n_b n_s n_c$ jobs in total have received service.
As a result, $\E[T^\SEB_{\text{res}}] = \frac{1}{\lambda} \E[N^\SEB_{\text{res}}] \leq \frac{1}{\lambda} n_b n_s n_c$.
\hfill\Halmos
\end{proof}

We now bound mean waiting time $\E[T^\SEB_{\text{wait}}]$, which is the dominant component of mean response time in heavy traffic. We first use WINE (\cref{prop:WINE}) to relate mean waiting time to the mean amount of work in each bucket.

\restatably\begin{lemma}
\label{lemma:WINE_for_wait_time}
For any policy $\pi$ that stabilizes the system,
\[
    \E[T^\pi_{\text{wait}}] \leq \frac{1}{\lambda} \sum_{i=1}^{n_b} \frac{c-1}{b_0 c^i} \E[W^\pi_{\leq i}]
        + \frac{1}{\lambda} \frac{\E[W^\pi]}{b_{n_b}}
\]
where $W^\pi_{\leq i}$ is work in the first $i$ buckets and $W^\pi$ is the total work in the system.
\end{lemma}

Proof deferred to \Cref{app:WINE_for_wait_time}.
Now, we use \cref{lem:residence_time_bound,lemma:WINE_for_wait_time}, as well as the work decomposition law \cref{prop:WDL} and our bound on waste \cref{thm:waste_bound},
to bound response time relative to the expected work in the resource pooled M/G/1 at or below each bucket cutoff.
To simplify our bound, we apply \cref{lemma:WINE_for_wait_time} to the resource pooled $\PSJF$ policy, and incorporate a bound on its response time into our bound on $SEB$.

\restatably\begin{proposition}
\label{thm:response_time_bound}
Under SEB,
\[
    \E[T^\SEB]\leq c \E[T^\PSJF] +\frac{A}{\lambda} \log\left(\frac{1}{1-\rho^{\boldsymbol{\nu}}}\right) \left(2n_b + \frac{1}{c - 1}\right) + \frac{1}{\lambda} n_b n_s n_c
\]
where
\[
    A = \frac{A_1 + 3A_2}{b_0} \frac{b_{n_b} - b_0}{b_0},
\]
and where $A_1$ and $A_2$ are defined in \Cref{thm:waste_bound}.
\end{proposition}

Proof deferred to \Cref{app:response_time_bound}.
Now, we examine \cref{thm:response_time_bound} in the heavy traffic:

\restate*\ref{thm:response_time_bound_PSJF}
\begin{proof}{\textit{Proof.}}
Recall from \Cref{sec:def-SEB} that we set
\[
    c = 1 + \frac{1}{1 + \log\left(\frac{1}{1-\rho^{\boldsymbol{\nu}}}\right)}.
\]
Since $b_{n_b}= b_0c^{n_b}$, and $b_0=z_{\min}$ and $b_{n_b}=z_{\max}$ are constants not depending on load, we have
\[
    n_b = \frac{\log\frac{z_{\max}}{z_{\min}}}{\log c} = \Theta\left(\log\left(\frac{1}{1-\rho^{\boldsymbol{\nu}}}\right)\right) \quad \text{as } \rho \to 1.
\]

The theorem follows from \Cref{thm:response_time_bound} by noting that as ${\boldsymbol{\rho}^{\boldsymbol{\nu}}}\to\boldsymbol{\nu}$, 
$c$ converges $1$, and $A$ converges to a fixed constant that does not scale with $\rho$.
\hfill\Halmos
\end{proof}

\subsection{Heavy-traffic Optimality}
\label{sec:proof-optimality}

\restatably\begin{theorem}
\label{thm:response-time_lower_bound}
Under any system load $\rho < 1$, and for any scheduling policy $\pi$,
$\E[T^\pi] \geq \E[T^\SRPT]$.
\end{theorem}

\begin{proof}{\textit{Proof.}}
We consider the resource-pooled M/G/1 queue defined in \cref{sec:lower-bound}. Since service rate in the original system after the priority index conversion is no more than 1, any scheduling policy in our original multi-server system with service constraints can be realized in the resource-pooled M/G/1 queue. The result then follows from the optimality of single-server SRPT \citep{schrage1968proof}.
\hfill\Halmos
\end{proof}

\restate*\ref{thm:heavy-traffic_optimality}
\begin{proof}{\textit{Proof.}}
Because the job priority index distribution is bounded, by Theorem 1 in \citet{lin2011heavy},
\[
\E[T^{\PSJF}] \geqslant \E[T^{\SRPT}] = \Theta\left(\frac{1}{1-\rho^{\boldsymbol{\nu}}}\right)
\]
Thus, it follows from \Cref{thm:response_time_bound_PSJF}
\[
\lim_{{\boldsymbol{\rho}^{\boldsymbol{\nu}}}\to\boldsymbol{\nu}}\frac{\E[T^\SEB]}{\E[T^\PSJF]}=\lim_{{\boldsymbol{\rho}^{\boldsymbol{\nu}}}\to\boldsymbol{\nu}}\left(c+\frac{\Theta\left(\log^2\left(\frac{1}{1-\rho^{\boldsymbol{\nu}}}\right)\right)}{\E[T^\PSJF]}\right)=1
\]

The last step in our optimality proof relies on the heavy-traffic optimality of PSJF-1 for bounded job priority index distributions:
$
\lim_{{\boldsymbol{\rho}^{\boldsymbol{\nu}}}\to\boldsymbol{\nu}}\frac{\E[T^\PSJF]}{\E[T^\SRPT]}=1
$.
This is a standard result that follows from the following waiting time inequality: $\E[W^\PSJF] \le \E[W^\SRPT]$ \citep{harchol2013performance}. For a detailed proof, see Theorem 5 in \citet{grosof2018srpt}. By \cref{thm:response-time_lower_bound}, our policy is thus heavy traffic optimal.
\hfill\Halmos
\end{proof}

\section{Practical variant: Dynamic SEB}
\label{sec:dynamic-SEB}

The SEB policy, as we have defined in \Cref{sec:def-SEB}, prioritizes analytical tractability over practicality. Indeed, SEB's performance is unsatisfactory at most practical loads because it idles buckets when there are not enough jobs to fulfill the preferred service option. When the system is underloaded, many nonempty buckets idle under SEB. In this section, we introduce three SEB-inspired policies that improve over SEB's performance, which we can \emph{Dynamic SEB-Best}, \emph{Dynamic SEB-good-enough}, and \emph{Dynamic SEB-any}. We call the original version \emph{Static SEB} to disambiguate.

One of static SEB's key insights is constructing priority index buckets so that jobs with small priority indices are prioritized. To avoid idling and skipping small jobs when there are too few jobs in the buckets, the dynamic SEB policies consider a sequence of \emph{dynamic buckets} in place of the original buckets in static SEB. 

\emph{Dynamic Buckets:} We construct these buckets as follows:
\begin{enumerate}
    \item Sort all jobs in the system in increasing remaining $\mathbf{k}^{\boldsymbol{\nu}}$-mapped priority index order.
    \item Bucket $i$ contains exactly jobs $1$ to $i$ in the sorted list.
\end{enumerate}

We highlight two modifications to static SEB: (1) dynamic buckets are constructed based on the \emph{remaining priority indices} instead of the original priority indices. (2) dynamic buckets are nested, not disjoint.

Another important insight from static SEB is that job classes must be equalized, by serving the preferred service option in a given bucket, as defined in \cref{sec:def-SEB}. We thus define the Dynamic SEB-best policy similarly:

\begin{definition}[Dynamic SEB-best]
    The scheduler examines dynamic buckets in increasing order. For each dynamic bucket $i$, the scheduler calculates the preferred facet service option $\mathbf{r}^*=\argmax_{r\in\mathcal{F}^{\boldsymbol{\nu}}}\inner{\mathbf{w}_{\perp{\boldsymbol{\rho}^{\boldsymbol{\nu}}}}}{\mathbf{r}}$.
    If that preferred service option can be served with the jobs in dynamic bucket $i$, the scheduler then serves jobs in the current bucket according to $r^*$, prioritizing jobs of least remaining priority index. Otherwise, the scheduler examines the next dynamic bucket.
\end{definition}

However, even restricting to the single most preferred service option can already be too restrictive. Note from \Cref{prop:service_drift} that the key property of the service option chosen by static SEB was that it is a facet option which achieves $\mathbf{r}$ satisfying $\frac{\inner{\mathbf{w}_{\perp{\boldsymbol{\rho}^{\boldsymbol{\nu}}}}}{\mathbf{r}}}{\|\mathbf{w}_{\perp{\boldsymbol{\rho}^{\boldsymbol{\nu}}}}\|_2}\geq\varepsilon_0$. We thus define Dynamic SEB-good enough accordingly:

\begin{definition}[Dynamic SEB-good-enough]
    At initialization, the scheduler first chooses a configuration parameter $\varepsilon_0>0$.
    At each point in time, the scheduler examines the dynamic buckets in increasing order and stops at the first bucket where a facet service option $\mathbf{r}$ satisfying $\frac{\inner{\mathbf{w}_{\perp{\boldsymbol{\rho}^{\boldsymbol{\nu}}}}}{\mathbf{r}}}{\|\mathbf{w}_{\perp{\boldsymbol{\rho}^{\boldsymbol{\nu}}}}\|_2}\geq\varepsilon_0$ becomes feasible. The scheduler serves the jobs in the current bucket according to that service option, prioritizing jobs of least remaining priority index.
\end{definition}

Dynamic SEB-good-enough does not restrict to the best equalizing option. Instead, it looks for a service option with some non-vanishing drift towards the load vector ${\boldsymbol{\rho}^{\boldsymbol{\nu}}}$, although not necessarily the one with the strongest drift. We note that such a service option is good enough to incur a state-space collapse towards the cone described in \Cref{sec:cone_SSC}, in the static SEB context. 

We can further relax the equalizing constraint by considering any service option on the limiting facet that is feasible, regardless of whether that option incurs a drift towards ${\boldsymbol{\rho}^{\boldsymbol{\nu}}}$. We define Dynamic SEB-any as follows

\begin{definition}[Dynamic SEB-any]
The scheduler examines dynamic buckets in increasing order and stops at the first bucket where a facet service option becomes feasible. This variant does not try equalizing among classes at all. It simply adopts the first feasible facet service option.
\end{definition} 

Among all three dynamic SEB policies, Dynamic SEB-any is the least restrictive in terms of finding the best equalizing option. It chooses whichever facet option that becomes available first, disregarding equalization altogether. We will see in \Cref{sec:simulation} that in some cases, choosing the first feasible facet options is better than waiting for an equalizing option.

Note that even with dynamic buckets, the scheduler may iterate through every bucket without finding a sufficiently equalizing service option, especially if the total number of jobs in the system is small. In this case, we define a fallback policy to further reduce idling.

\emph{Fallback Policy:} If the scheduler iterates through all of the dynamic bucket without finding a sufficiently equalizing service option, then it chooses from all feasible service options (not just ones on the limiting facet) an option that maximizes the total $\mathbf{k}^{\boldsymbol{\nu}}$-mapped service rate.

Also note that the above scheduling rules may produce ties where several service options are equally advantageous. In such cases, we use the following tie-breaking rule.

\emph{Tie-breaking Rule:} If the scheduler needs to pick between multiple service options which are equally preferable, it picks the one with the largest $\inner{\mathbf{w}_{\perp {\boldsymbol{\rho}^{\boldsymbol{\nu}}}}}{\mathbf{r}}$.

\section{Numerical Evaluation via Simulation}
\label{sec:simulation}
In this section, we thoroughly compare our dynamic SEB's performances against a range of prior policies and heuristics via simulation.
Our evaluation shows that each of our dynamic SEB exhibits excellent performance that greatly exceeds all prior policies and heuristics.
Moreover, we demonstrate that each of our dynamic SEB maintains this excellent performance for a variety of generalized switch settings that are of great interest to both the queueing and computer systems communities: compatibility scheduling, multi-server job scheduling, and multi-resource job scheduling.

\Cref{sec:benchmark_policies} introduces comparison policies. \Cref{sec:empirical-compat, sec:empirical-msj, sec:empirical-mrj} evaluate dynamic SEB policies in compatibility, MSJ, and MRJ scheduling settings, respectively.

\subsection{Benchmark Policies}
\label{sec:benchmark_policies}
We now list several pre-existing generalized switch policies and heuristics, which we use as benchmarks. Because prior work on duration-aware scheduling in generalized switch is scarce, we use simple adaptations of previous policies that might not be duration-aware or general enough for the generalized switch.

\begin{description}
    \item[MaxWeight (MW):] The scheduler computes the inner products $\inner{\mathbf{r}}{\mathbf{q}}$ between feasible service options $\{\mathbf{r}\}$ and queue lengths $\mathbf{q}$ (number of jobs in each class). The scheduler chooses the service option that maximizes this inner product. If there are more jobs in a class than required by the scheduler, jobs are picked in FCFS order.
    \item[MaxWeight-Queue SRPT (MW-Queue SRPT):] This is the same as MaxWeight-Queue except that jobs in a class are chosen in SRPT order, where SRPT is defined based on job priority index, using the limiting facet.
    \item[MaxWeight-Work SRPT (MW-Work SRPT):] When a job arrives or departs, the inner products $\inner{\widetilde{\mathbf{r}}}{\mathbf{w}}$ between work completion rates $\{\widetilde{\mathbf{r}}\}$ and work vector $\mathbf{w}$ (remaining work in each class) are computed. The scheduler chooses the service option that maximizes this inner product. Jobs in the same class are picked in SRPT order.
    \item[ServerFilling-SRPT:] This is an MSJ-specific policy introduced by \cite{grosof2022optimal}. Jobs are sorted in least remaining area order. The candidate set is defined to consist of the minimal initial sequence of jobs with total server need at least $n_s$, and jobs are served in most-server-need-first order among the candidate set, tie-broken by lower remaining area. The policy was proven to be heavy-traffic optimal if all server needs perfectly divide all larger server needs, and the total number of servers.
\end{description}

In addition to the benchmark policies, simulating a lower-bounding policy is also helpful, as it provides an optimistic performance benchmark and illustrates how close dynamic SEB policies are to this ideal limit. Although we use resource-pooled SRPT as a lower bound when establishing the heavy-traffic optimality of SEB (\Cref{thm:heavy-traffic_optimality}), we do not adopt it as the reference lower bound in our numerical experiments, as it is significantly more optimistic than any policy under consideration across all practical load regimes.

Instead, for each setting, we construct a heterogeneous multiserver system and use SRPT in this system as the reference lower bound. The server rates are determined as follows: the first server is assigned the maximum achievable service rate of any single job, the second server is chosen so that the combined rate of the first two servers equals the maximum achievable total service rate of the two fastest jobs, and so on, until the total service rate sums to one.

We emphasize that the lower bound produced by multiserver-SRPT (MS-SRPT) policy is also not achievable in the original system, but it provides a tighter lower bound than resource-pooled SRPT.

In our simulations, for simplicity, we only update a policy's service option at moments when arrivals and completions occur. This affects the dynamic SEB and MaxWeight-Work SRPT policies. Our initial exploration indicated that this did not have an appreciable impact on response times.

In the subsequent evaluation sections, we use three types of plots to illustrate the performance comparisons: (1) the raw mean response times under dynamic SEB policies and benchmark policies, (2) ratios of mean response times under dynamic SEB policies and benchmark policies to mean response times under multiserver-SRPT, our referencing lower bound, and (3) percentage of gap closed. We notice, via simulations, that MaxWeight-Queue SRPT has the best empirical performance among all benchmark policies. We therefore consider, for each dynamic SEB policy, the following metric:
\[
\text{Percentage of Gap Closed} = \frac{\E\left[T^{\text{MaxWeight-Queue SRPT}}\right]-\E\left[T^{\text{dynamic SEB}}\right]}{\E\left[T^{\text{MaxWeight-Queue SRPT}}\right]-\E\left[T^{\text{Multiserver-SRPT}}\right]}
\]
which measures the percentage of improvement dynamic SEBs has over the best prior policy -- MaxWeight-Queue SRPT.

Throughout the evaluation sections, we use a bounded Pareto distribution with $\alpha=1.5$ and support $[1, 10^4]$ as the job priority index distribution. For all settings, we consider a wide range of system loads from 0.6 to 0.99. For the compatibility and MRJ settings, we run 30 replications, with $3\times 10^6$ job completions per load, for all loads. For the MSJ settings, we we run 30 replications, with $3\times 10^6$ job completions per load, for 0.6 through 0.98. For load 0.99, we run 60 replications, with $3\times 10^6$ job completions. 95\% confidence intervals are shown as the dotted lines.

\subsection{Evaluation of compatibility scheduling}
\label{sec:empirical-compat}

In the compatibility scheduling setting, we consider two specific system configurations:

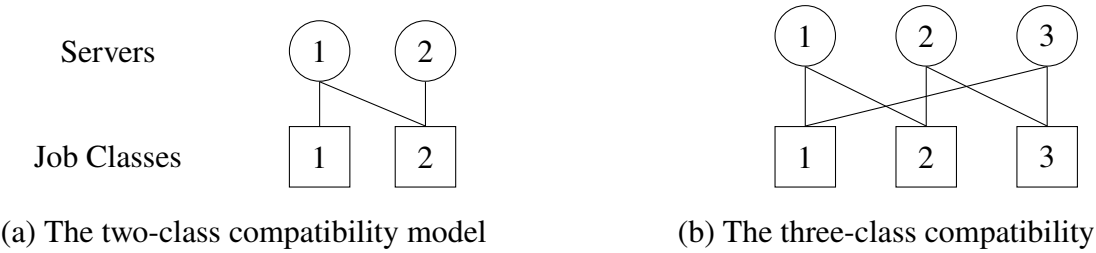
\begin{figure}[H]
    \centering
    \begin{subfigure}{0.45\textwidth}
        \centering
        \begin{tikzpicture}[scale=0.7]
\tikzstyle{style} = [draw, minimum size=0.8cm, inner sep=0pt]

\node at (-4, 2) {Servers};
\node at (-4, 0) {Job Classes};

\node[circle, style] (C1) at (0,2) {1};
\node[circle, style] (C2) at (2,2) {2};

\node[rectangle, style] (S1) at (0,0) {1};
\node[rectangle, style] (S2) at (2,0) {2};

\draw[-] (S1.north) -- (C1.south);
\draw[-] (S2.north) -- (C1.south);
\draw[-] (S2.north) -- (C2.south);

\end{tikzpicture}
        \caption{The two-class compatibility model}
        \label{fig:compatibility_2classes}
    \end{subfigure}
    \hfill
    \begin{subfigure}{0.45\textwidth}
        \centering
        \begin{tikzpicture}[scale=0.8]

\tikzstyle{style} = [draw, minimum size=0.8cm, inner sep=0pt]

\node[circle, style] (C1) at (0,2) {1};
\node[circle, style] (C2) at (2,2) {2};
\node[circle, style] (C3) at (4,2) {3};

\node[rectangle, style] (S1) at (0,0) {1};
\node[rectangle, style] (S2) at (2,0) {2};
\node[rectangle, style] (S3) at (4,0) {3};

\draw[-] (S1.north) -- (C1.south);
\draw[-] (S2.north) -- (C1.south);
\draw[-] (S2.north) -- (C2.south); 
\draw[-] (S3.north) -- (C2.south); 
\draw[-] (S3.north) -- (C3.south); 
\draw[-] (S1.north) -- (C3.south); 

\end{tikzpicture}
        \caption{The three-class compatibility model}
        \label{fig:compatibility_3classes}
    \end{subfigure}
    \caption{ Two compatibility graphs we consider for evaluations in \Cref{sec:empirical-compat}}
    \label{fig:compatibility_graphs}
\end{figure}

We start with the two-class compatibility model (\cref{fig:compatibility_2classes}). There are two job classes: class-1 jobs are compatible only with server 1 and class-2 jobs are compatible with both servers. Each job is of class 1 with probability $1/4$ and class 2 with probability $3/4$, independent of the job priority index distribution. The limiting facet is bounded by points $[0,2]$ and $[1,1]$. The multiserver SRPT lower bounding system for this setting has 2 servers with service rates $1/2$ and $1/2$. Results are shown in \Cref{fig:compatibility_2classes_mean_response_time,fig:compatibility_2classes_ratio_to_srpt,fig:compatibility_2classes_gap_closed}.

\begin{figure}
    \centering
    \begin{subfigure}[t]{0.49\textwidth}
        \centering
        \includegraphics[height=141pt, width=\textwidth, pagebox=mediabox]{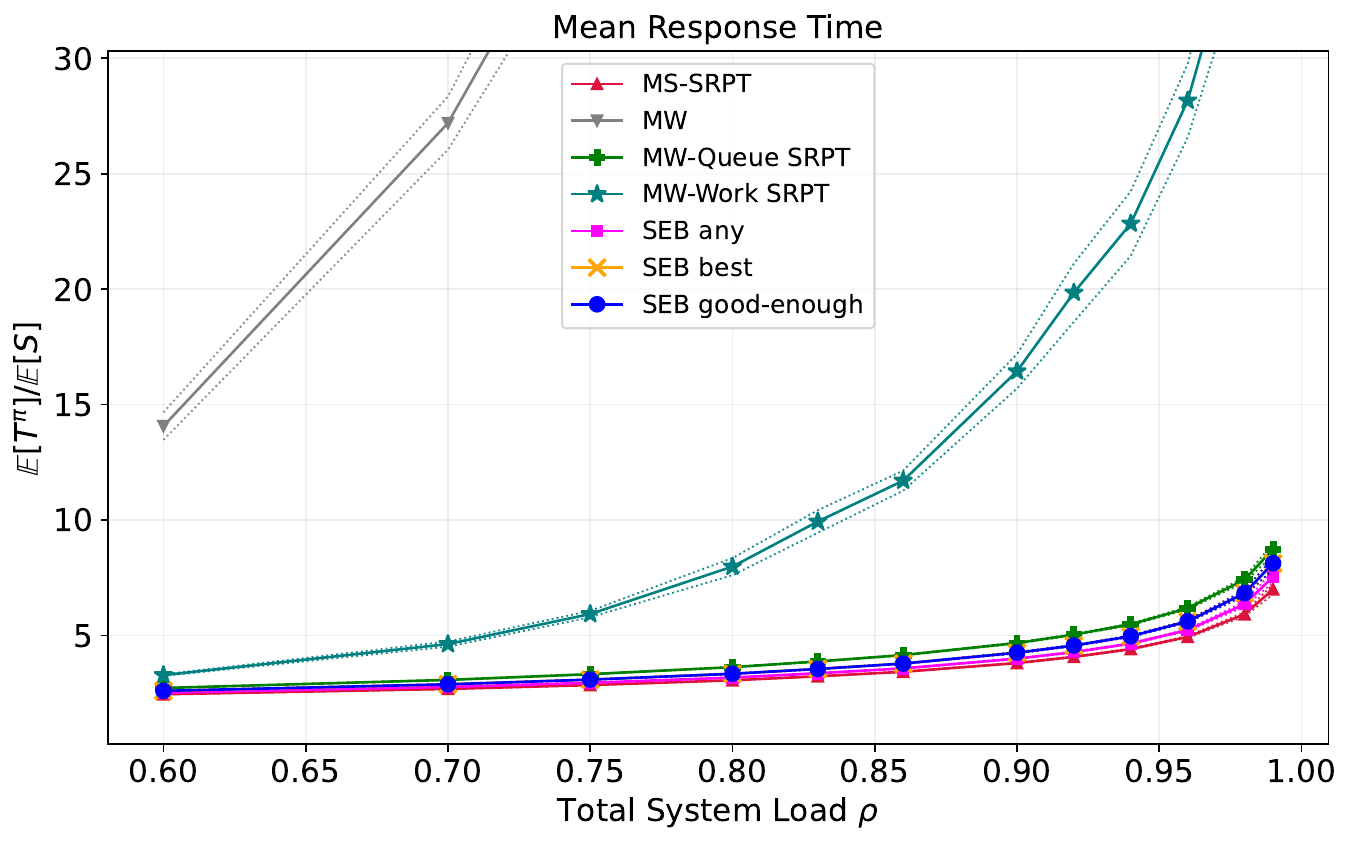}
        \caption{Mean response time. The curves for MW-Queue SRPT, SEB any, SEB best, SEB good-enough, and MS-SRPT appear clustered near the bottom of the figure because their response times are substantially smaller than those of MW and MW-Work SRPT at this scale.}
        \label{fig:compatibility_2classes_mean_response_time}
    \end{subfigure}
    \hfill
    \begin{subfigure}[t]{0.49\textwidth}
            \centering
        \includegraphics[height=140pt,width=\textwidth, pagebox=mediabox]{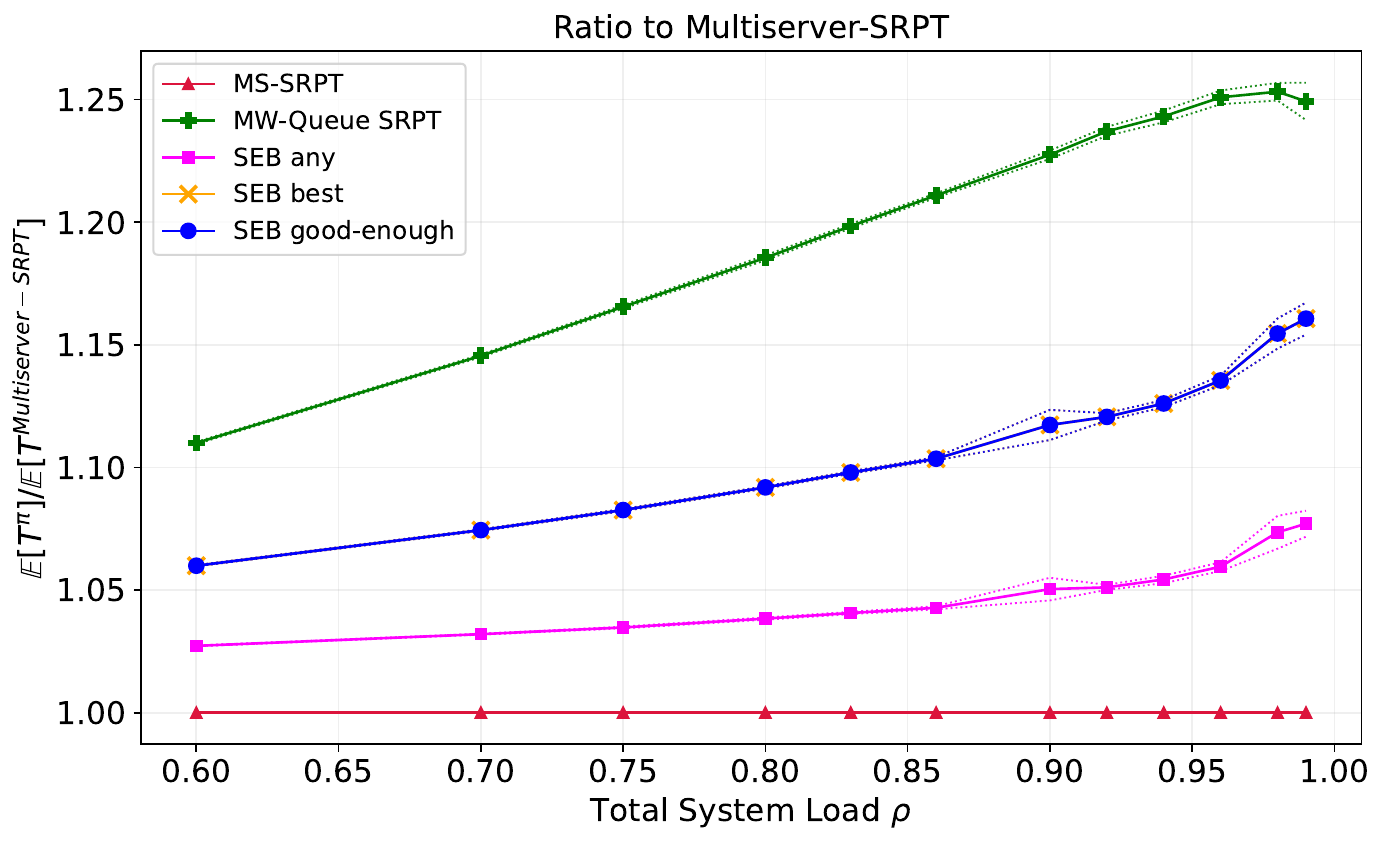}
        \caption{Ratio to Multiserver-SRPT. The curves for SEB best and SEB good-enough overlap across all loads because the two policies have nearly identical response times throughout the evaluated range.}
        \label{fig:compatibility_2classes_ratio_to_srpt}
    \end{subfigure}
    \caption{Two-server compatibility setting.}
    \label{fig:compatibility_pair1}
\end{figure}



\begin{figure}
    \centering
    \begin{subfigure}{0.49\textwidth}
        \centering
        \includegraphics[width=\textwidth]{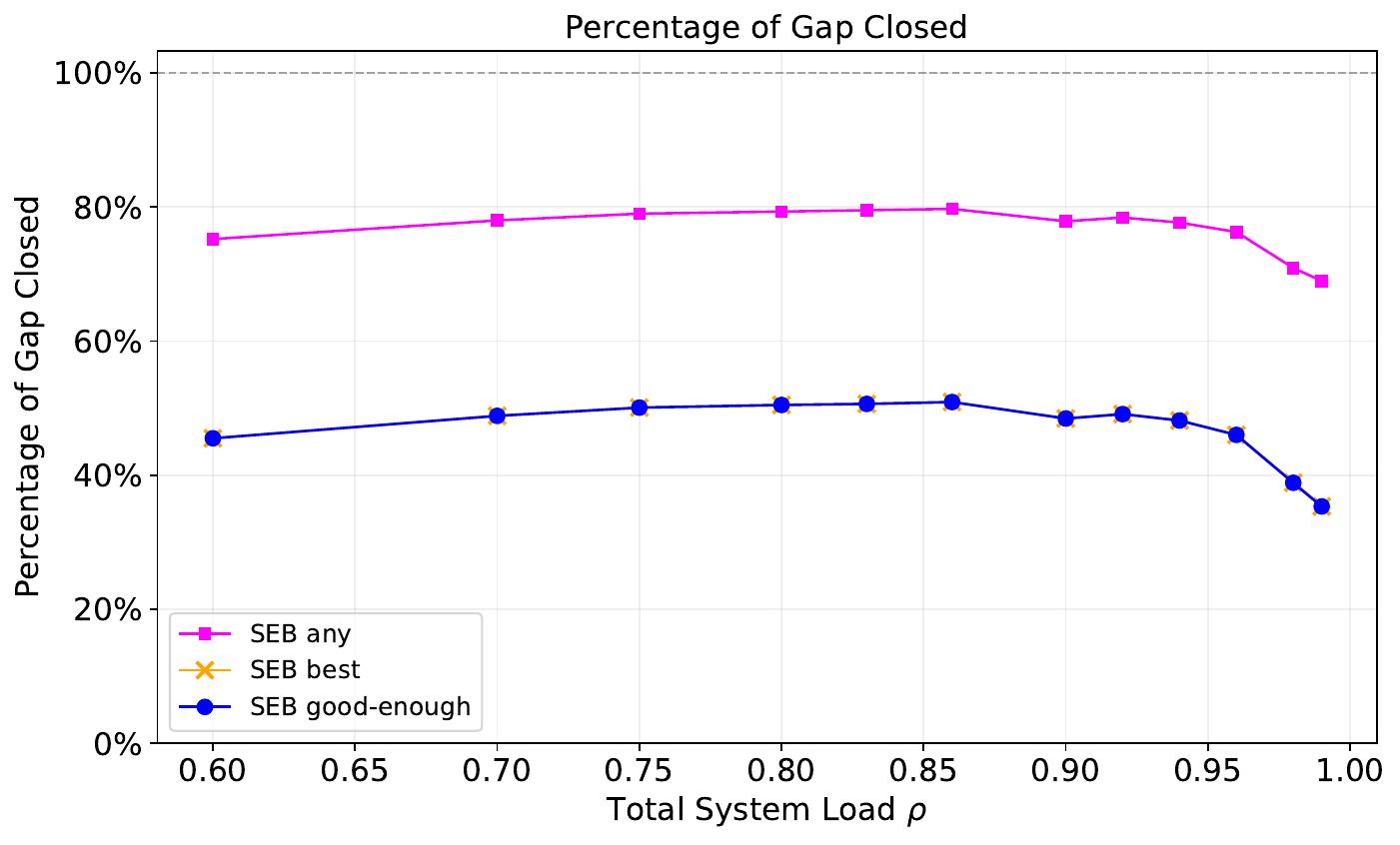}
        \caption{Two-server compatibility setting: percentage of gap closed.}
        \label{fig:compatibility_2classes_gap_closed}
    \end{subfigure}
    \hfill
    \begin{subfigure}{0.49\textwidth}
        \centering
        \includegraphics[width=\textwidth]{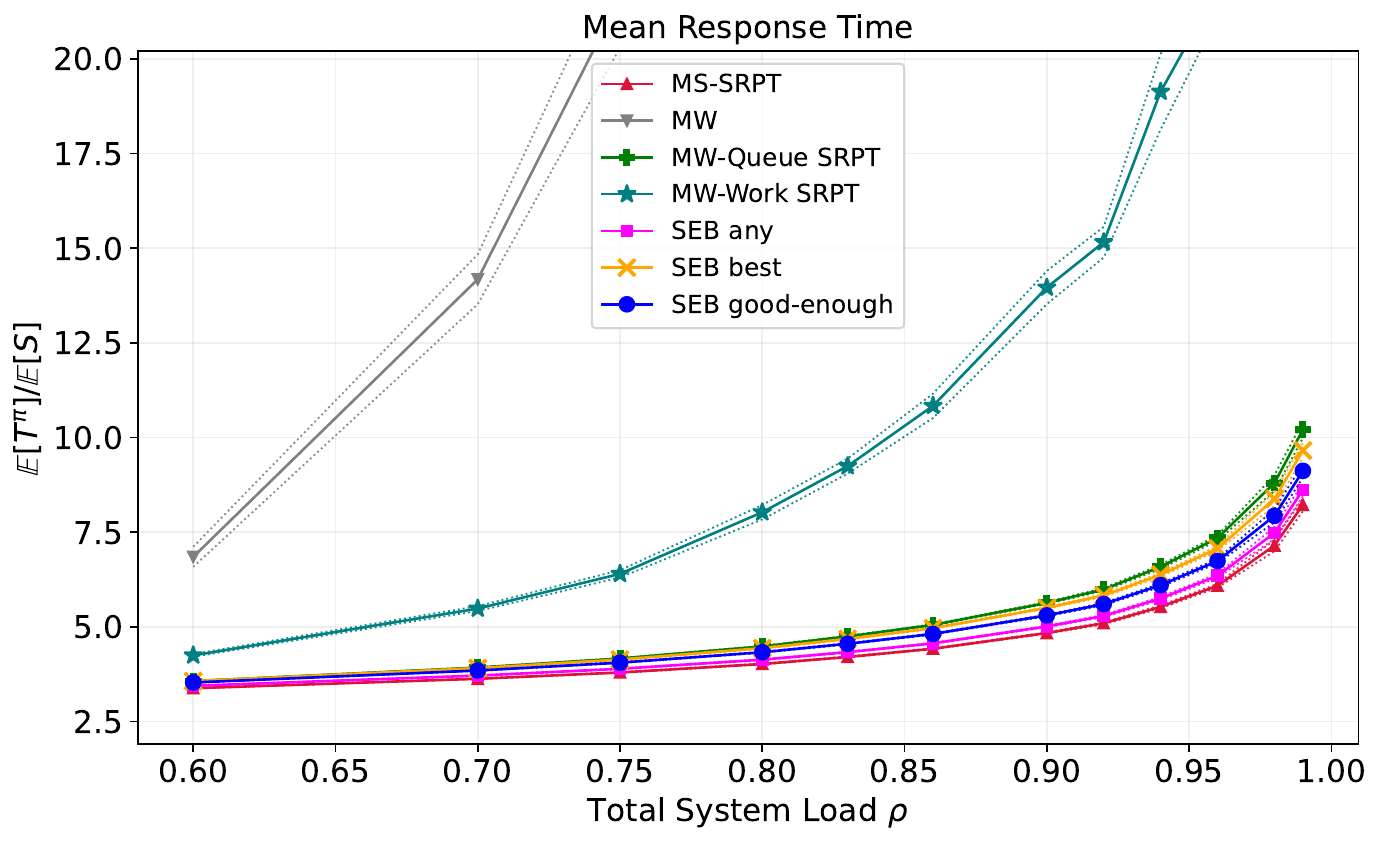}
        \caption{Three-server compatibility setting: mean response time.}
        \label{fig:compatibility_3classes_mean_response_time}
    \end{subfigure}
    \caption{Further compatibility-setting results.}
    \label{fig:compatibility_pair2}
\end{figure}

In the two-server compatibility case, all three dynamic SEB policies consistently outperform all benchmark policies across all loads. In particular, across all loads, dynamic SEB policies outperform MaxWeight-Queue SRPT by substantial margins. Dynamic SEB-any shows the largest improvement, closing more than 70\% of the gap under most loads. Dynamic SEB-best and good-enough have very similar performance, both closing more than 40\% of gaps under most loads.

The other compatibility model we consider is a three-class compatibility model (\cref{fig:compatibility_3classes}). There are three job classes: class-1 jobs are compatible with servers 1 and 3; class-2 jobs are compatible with servers 1 and 2; class-3 jobs are compatible with servers 2 and 3. Each job is of class 1 with probability $1/3$, class 2 with probability $1/3$, and class 3 with probability $1/3$, independent of the job priority index distribution. The limiting facet is bounded by points $[0, 1, 2], [0, 2, 1], [1, 0, 2], [1, 2, 0], [2, 0, 1]$, and $[2, 1, 0]$ The multiserver SRPT lower bounding system for this setting has 3 servers with service rates $1/3$, $1/3$, and $1/3$. Results are shown in \Cref{fig:compatibility_3classes_mean_response_time,fig:compatibility_3classes_ratio_to_srpt,fig:compatibility_3classes_gap_closed}.


\begin{figure}
    \centering
    \begin{subfigure}{0.49\textwidth}
        \centering
        \includegraphics[height=140pt, width=\textwidth]{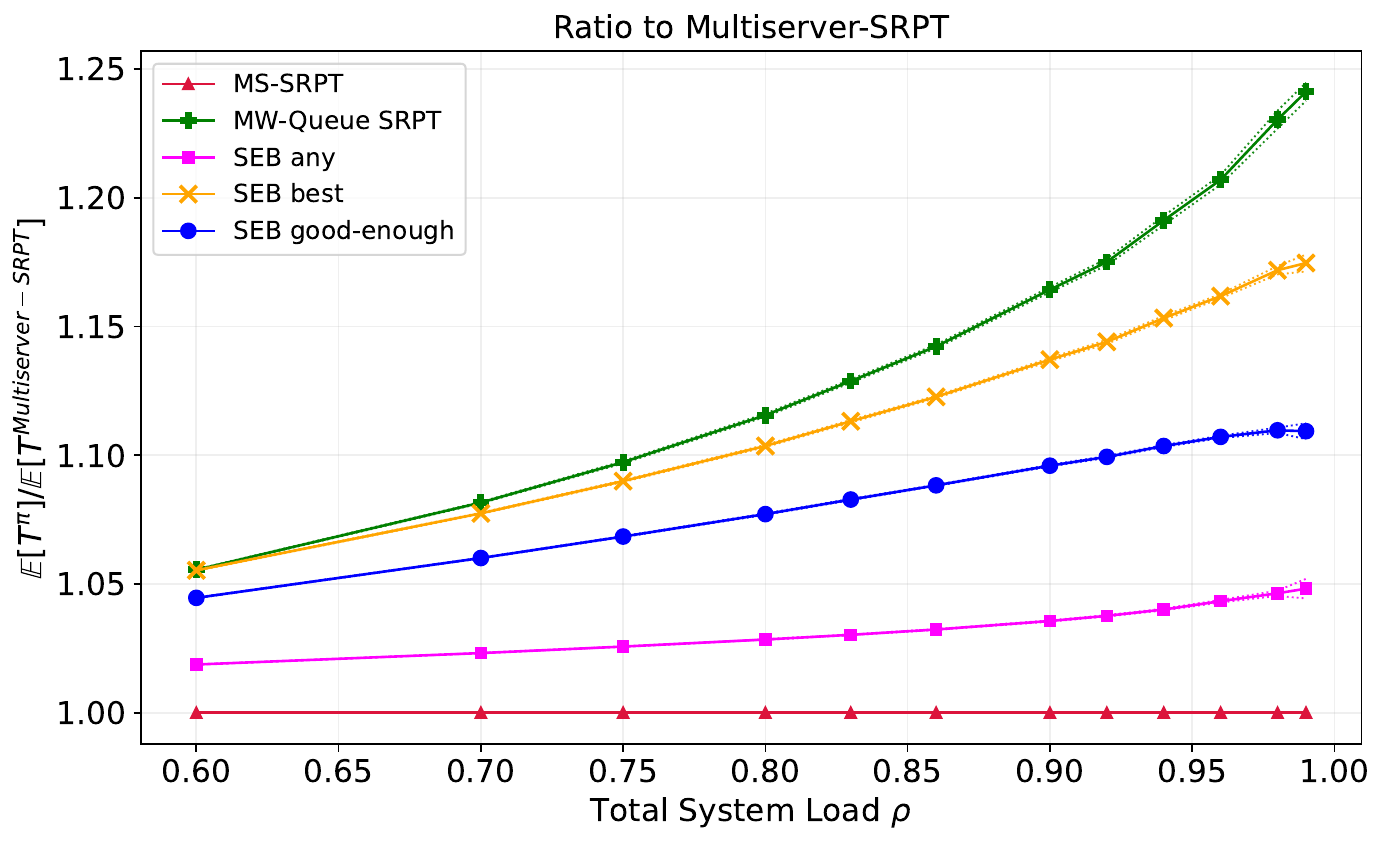}
        \caption{Ratio to Multiserver-SRPT.}
        \label{fig:compatibility_3classes_ratio_to_srpt}
    \end{subfigure}
    \hfill
    \begin{subfigure}{0.49\textwidth}
        \centering
        \includegraphics[height=140pt,width=\textwidth]{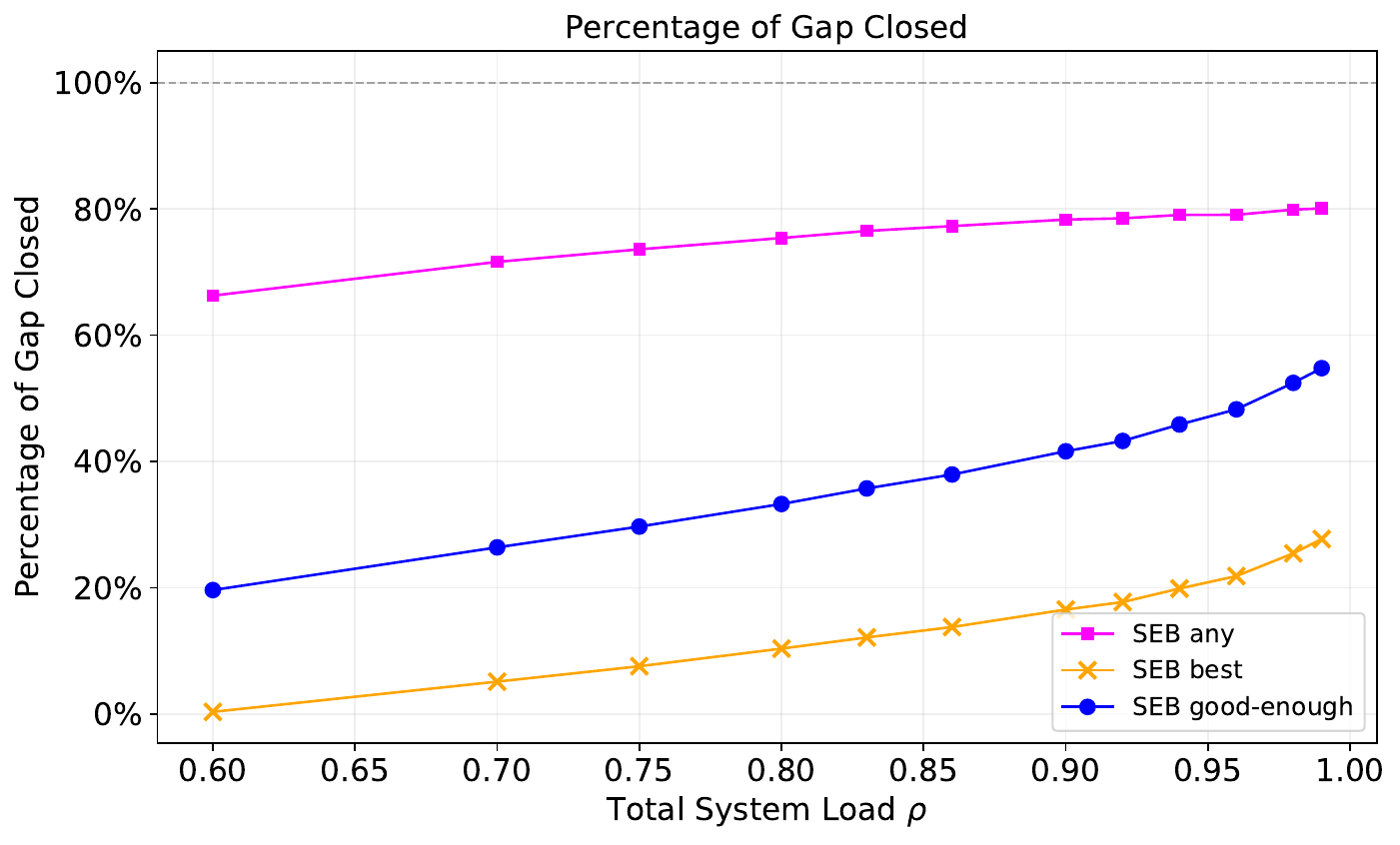}
        \caption{Percentage of gap closed.}
        \label{fig:compatibility_3classes_gap_closed}
    \end{subfigure}
    \caption{Three-server compatibility setting.}
    \label{fig:compatibility_pair3}
\end{figure}



In the three-server compatibility settings, dynamic SEB policies consistently outperform all benchmark policies across all loads. Dynamic SEB-any policy has the greatest improvement over MaxWeight-Queue SRPT, with more than 70\% gap closed under most loads. Dynamic SEB-best and dynamic SEB-good-enough also consistently outperform MaxWeight-Queue SRPT, although by a smaller margin. In particular, percentages of gap closed increase with load for all dynamic SEB policies.

\subsection{Evaluation of multiserver-job scheduling}
\label{sec:empirical-msj}

In the multiserver-job (MSJ) setting, an important benchmark is the serverFilling-SRPT. We thus consider two different 11-server systems: in the first setting, the mapping vector $\mathbf{k}^{\boldsymbol{\nu}}$ coincides with the area defined for serverFilling-SRPT \citep{grosof2022optimal}, which maps job durations to a different notion of priority indices in serverFilling-SRPT. In the second setting, $\mathbf{k}^{\boldsymbol{\nu}}$ and area are different. 

The first 11-server system has two job classes. Class-1 jobs need 2 servers and class-2 jobs need 3 servers. Each job is of class 1 with probability $5/11$ and class 2 with probability $6/11$, independent of the job priority index distribution. The limiting facet is bounded by the points $[4, 1]$ and $[1, 3]$, where both service options use all 11 servers. The multiserver SRPT lower bounding system for this setting has 4 servers with service rates $3/11$, $3/11$, $3/11$, and $2/11$. Results are shown in \Cref{fig:MSJ_11_servers_2_classes_mean_response_time,fig:MSJ_11_servers_2_classes_ratio_to_srpt,fig:MSJ_11_servers_2_classes_gap_closed}

\begin{figure}
    \centering
    \begin{subfigure}{0.49\textwidth}
        \centering
        \includegraphics[height=140pt,width=\textwidth]{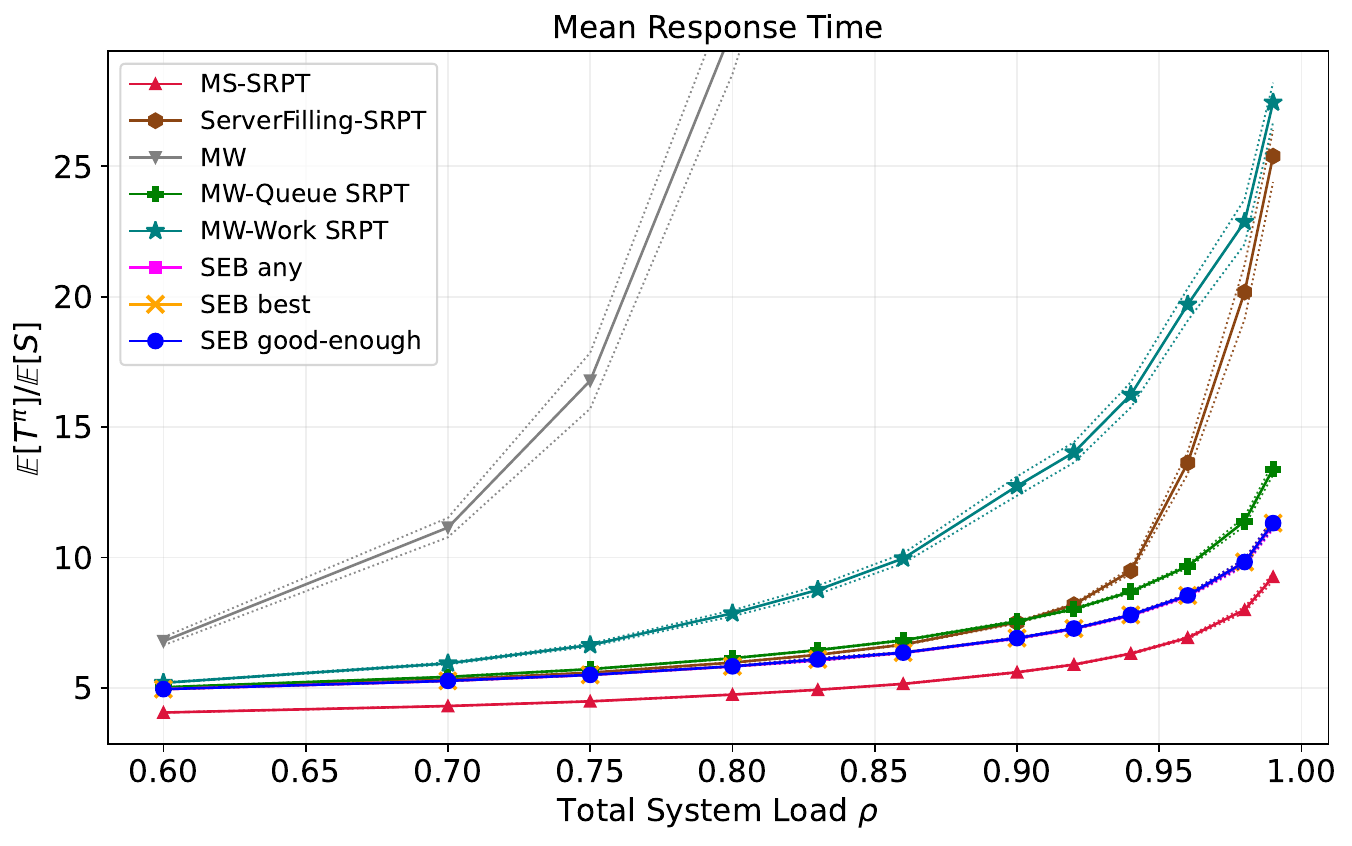}
        \caption{Mean response time.}
        \label{fig:MSJ_11_servers_2_classes_mean_response_time}
    \end{subfigure}
    \hfill
    \begin{subfigure}{0.49\textwidth}
        \centering
        \includegraphics[height=140pt,width=\textwidth]{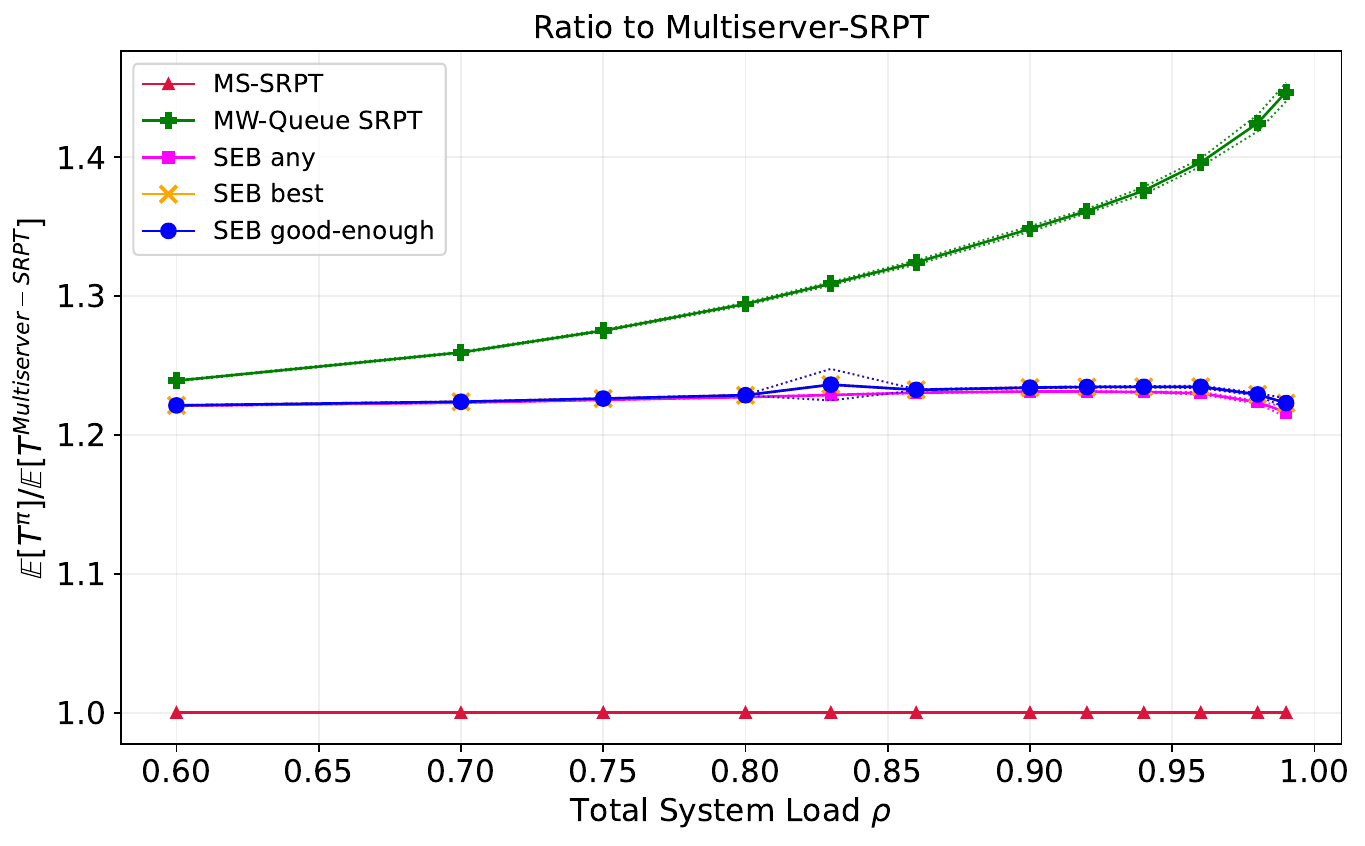}
        \caption{Ratio to Multiserver-SRPT.}
        \label{fig:MSJ_11_servers_2_classes_ratio_to_srpt}
    \end{subfigure}
    \caption{11-server 2-class MSJ setting.}
    \label{fig:msj_pair1}
\end{figure}



\begin{figure}
    \centering
    \begin{subfigure}{0.49\textwidth}
        \centering
        \includegraphics[width=\textwidth]{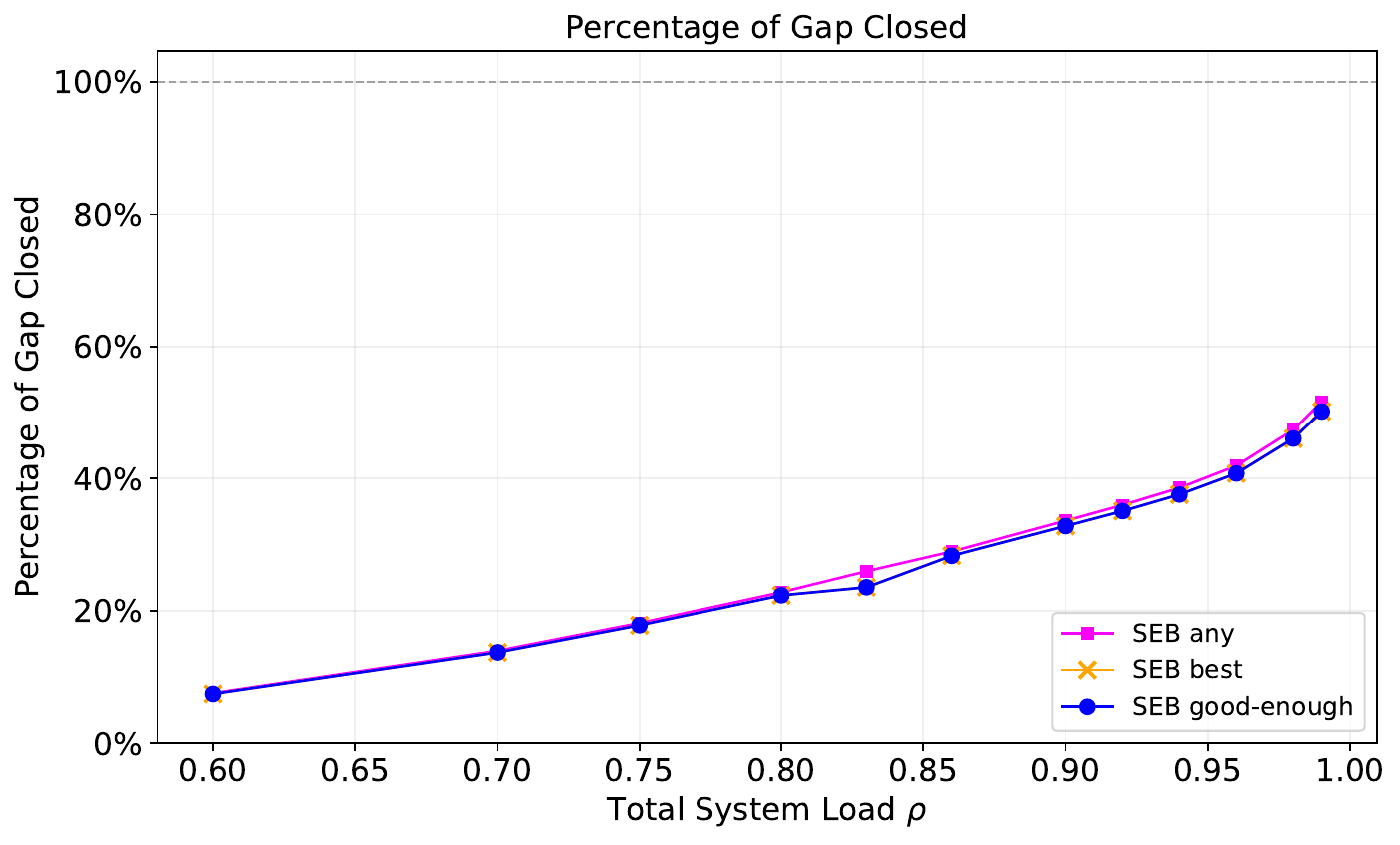}
        \caption{11-server 2-class MSJ setting: percentage of gap closed.}
        \label{fig:MSJ_11_servers_2_classes_gap_closed}
    \end{subfigure}
    \hfill
    \begin{subfigure}{0.49\textwidth}
        \centering
        \includegraphics[width=\textwidth]{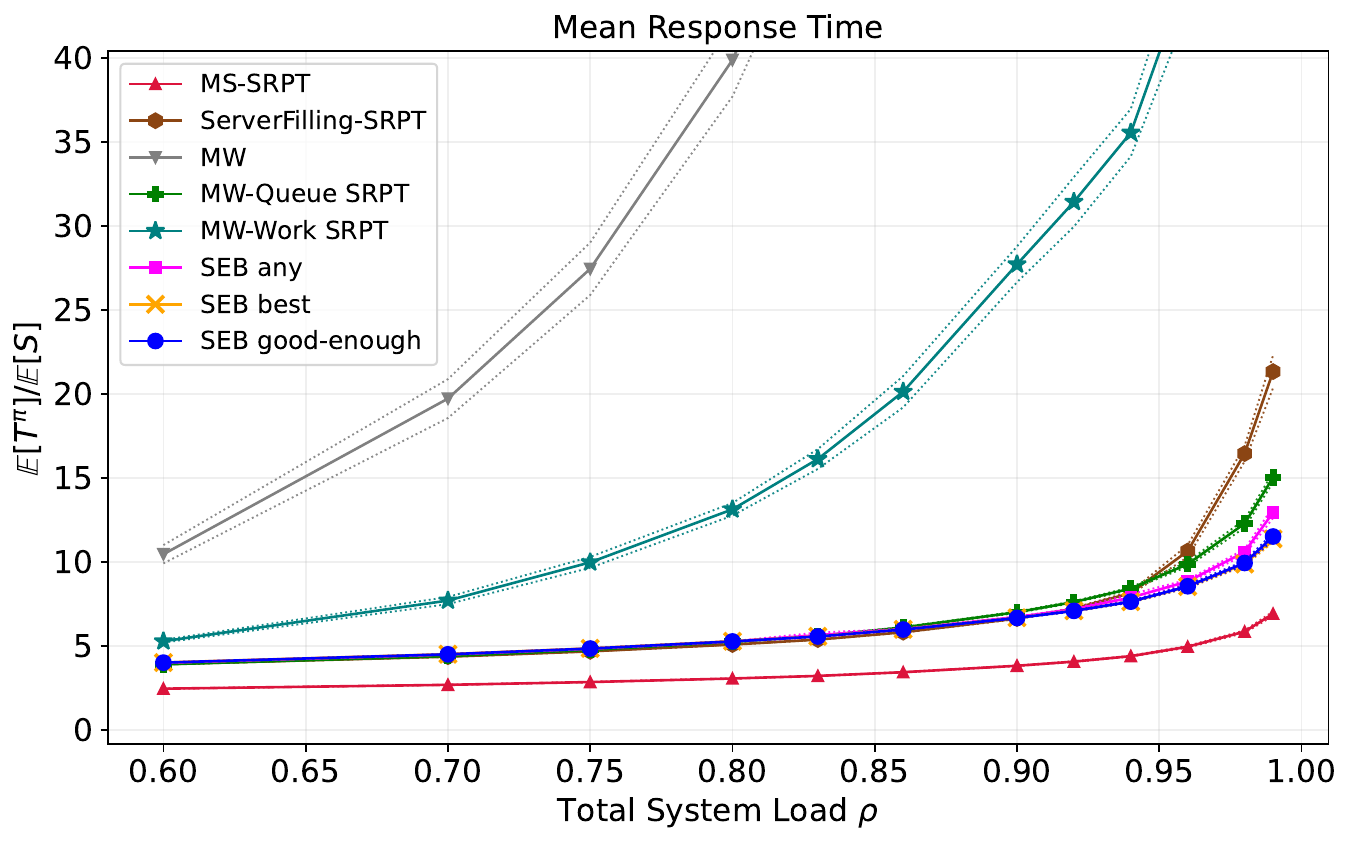}
        \caption{11-server 3-class MSJ setting: mean response time.}
        \label{fig:MSJ_11_servers_3_classes_mean}
    \end{subfigure}
    \caption{Further MSJ evaluation.}
    \label{fig:msj_pair2}
\end{figure}


In the two-class setting, the dynamic SEB policies consistently outperform all benchmark policies. Here, all three dynamic SEB policies yield similar performances across all loads. In comparison to MaxWeight-Queue SRPT, dynamic SEB policies close increasingly higher percentages of gaps as load increases, reaching around 50\% as the system becomes heavily loaded. In comparison to serverFilling-SRPT, dynamic SEB policies yield similar performance when the load is moderate, but show greater advantage as load increases. We note that serverFilling-SRPT does not always use service facet options and does not attempt to balance small and large jobs. As load increases, serverFilling-SRPT may serve large jobs before small jobs or under-utilize the system capacity.

The second 11-server system has three job classes. Class-1 jobs need 2 servers, class-2 jobs need 3 servers, and class-3 jobs need 5 servers. Each job is of class 1 with probability $1/6$, class 2 with probability $1/6$, and class 3 with probability $2/3$, independent of the job priority index distribution. The limiting facet is bounded by the points $[3, 0, 1]$, $[0, 2, 1]$, and $[0, 0, 2]$. Note that not all of these service options use all 11 servers: $[0, 0, 2]$ uses only 10. The multiserver SRPT lower bounding system for this setting has 2 servers with service rates $1/2$ and $1/2$. Results are shown in \Cref{fig:MSJ_11_servers_3_classes_mean, fig:MSJ_11_servers_3_classes_ratio, fig:MSJ_11_servers_3_classes_gap}


\begin{figure}
    \centering
    \begin{subfigure}{0.49\textwidth}
        \centering
        \includegraphics[width=\textwidth]{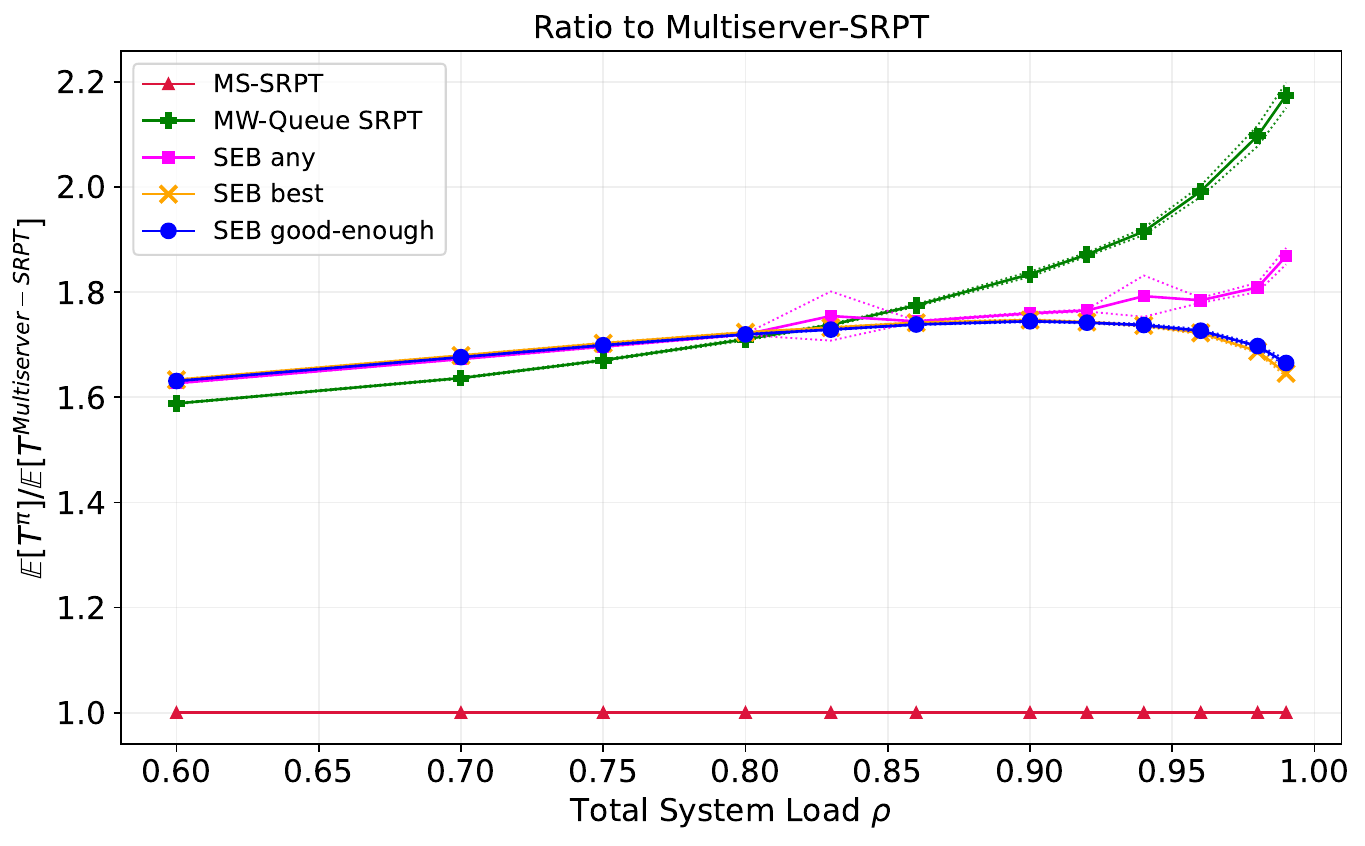}
        \caption{Ratio to Multiserver-SRPT.}
        \label{fig:MSJ_11_servers_3_classes_ratio}
    \end{subfigure}
    \hfill
    \begin{subfigure}{0.49\textwidth}
        \centering
        \includegraphics[width=\textwidth]{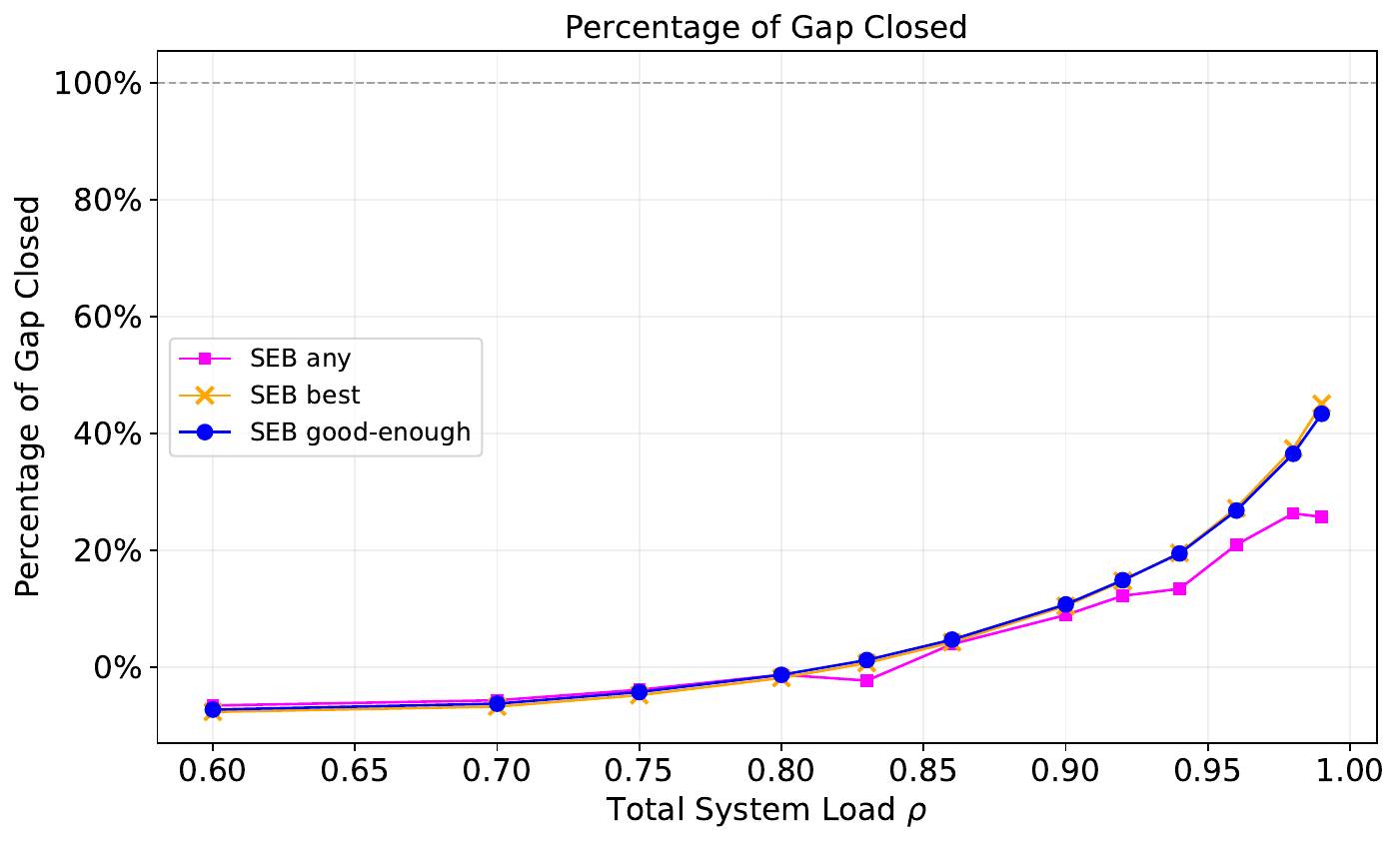}
        \caption{Percentage of gap closed.}
        \label{fig:MSJ_11_servers_3_classes_gap}
    \end{subfigure}
    \caption{11-server 3-class MSJ setting.}
    \label{fig:msj_pair3}
\end{figure}



In the three-class setting, all three dynamic SEB policies have competitive performance across all loads. Although both serverFilling-SRPT and MaxWeight-Queue SRPT outperform dynamic SEB policies by small margins when loads are moderate, dynamic SEB policies regain advantages as load increases. Dynamic SEB-best and dynamic SEB-good-enough outperform dynamic SEB-any under high loads, suggesting that in this setting, equalization benefits performance. Similar to the two-class setting, serverFilling-SRPT shows deteriorating performance at high loads, while dynamic SEB policies remain robust to increasing loads.

\subsection{Evaluation of multiresource-job scheduling}
\label{sec:empirical-mrj}

In the multiresource-job (MRJ) setting, we consider a system with two resources and two job classes. The system capacity is 10 units of resource A and 8 units of resource B. Each class-1 job requires 1 unit of resource A and 3 units of resource B. Each class-2 job requires 2 units of resource A and 1 unit of resource B. Each job is of class 1 with probability $1/2$ and class 2 with probability $1/2$, independent of the job priority index distribution. The limiting facet is bounded by points $[1,4]$ and $[2,2]$. The multiserver SRPT lower bounding system for this setting has 4 servers with service rates $1/3$, $1/3$, $1/6$, and $1/6$. Results are shown in \Cref{fig:MRJ_2_classes_mean_response_time,fig:MRJ_2_classes_ratio_to_srpt,fig:MRJ_2_classes_gap_closed}.

\begin{figure}
    \centering
    \begin{subfigure}{0.49\textwidth}
        \centering
        \includegraphics[width=\textwidth]{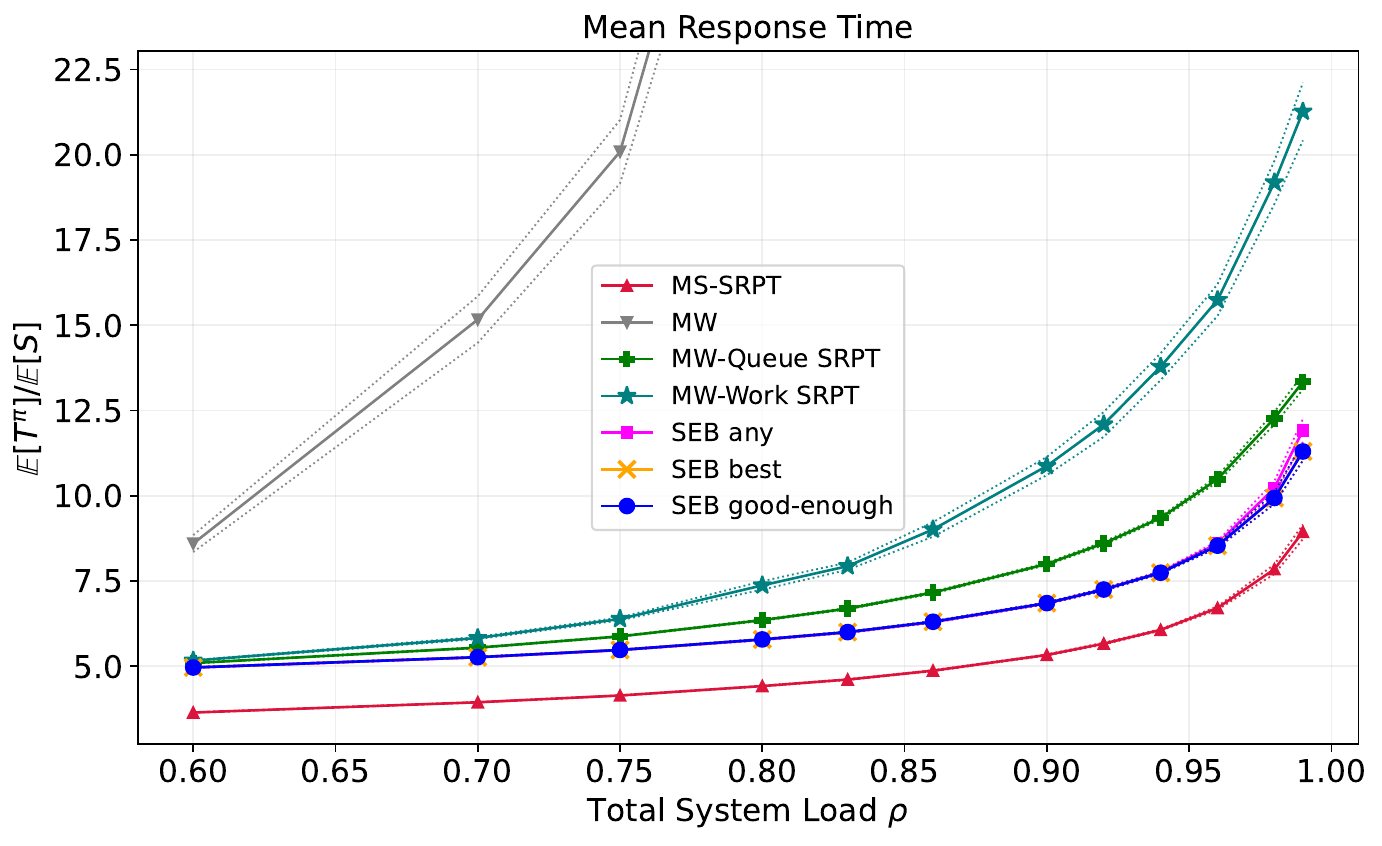}
        \caption{Mean response time.}
        \label{fig:MRJ_2_classes_mean_response_time}
    \end{subfigure}
    \hfill
    \begin{subfigure}{0.49\textwidth}
        \centering
        \includegraphics[width=\textwidth]{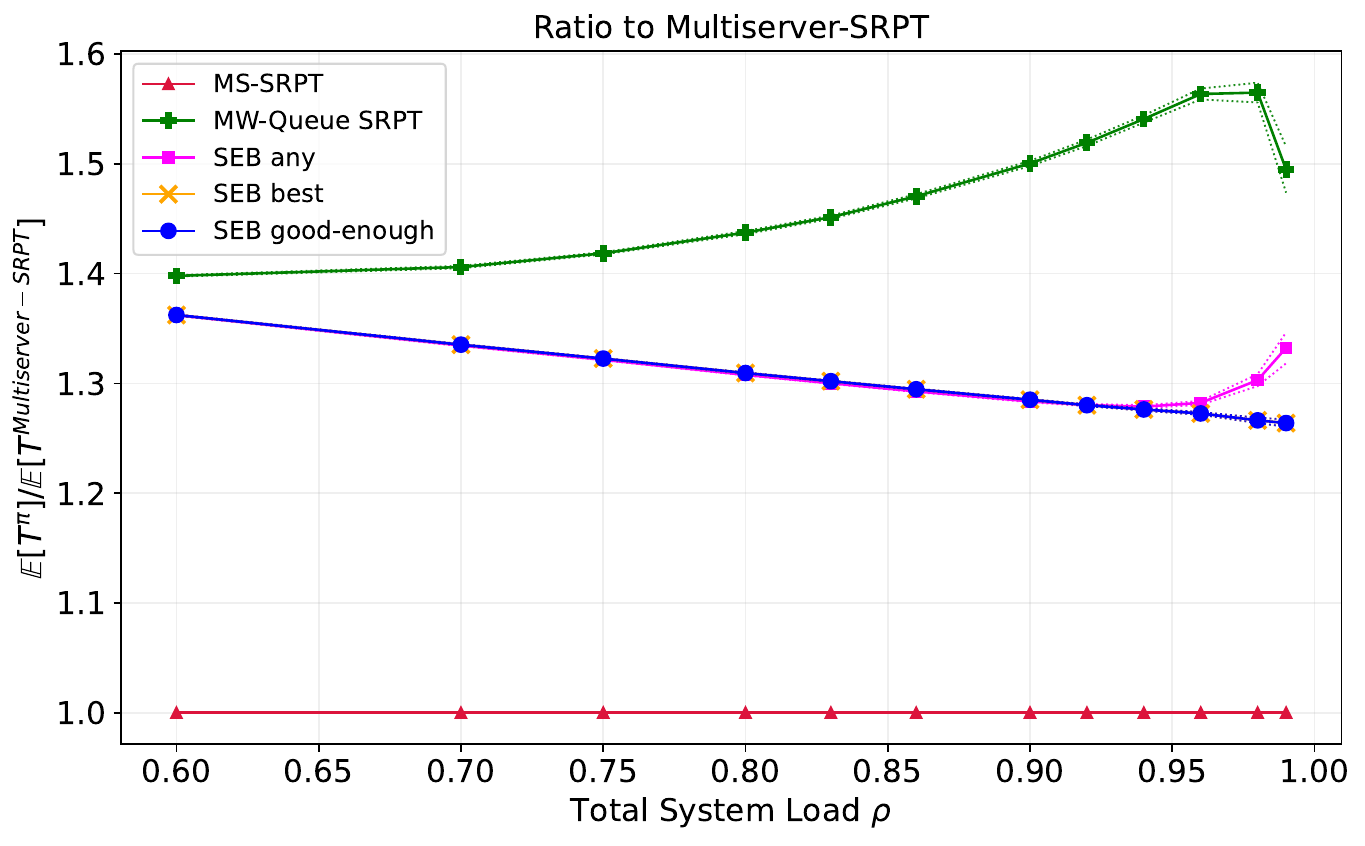}
        \caption{Ratio to Multiserver-SRPT.}
    \label{fig:MRJ_2_classes_ratio_to_srpt}
    \end{subfigure}
    \caption{2-class MRJ setting.}
    \label{fig:mrj_pair1}
\end{figure}



\begin{figure}[h]
    \centering
    \includegraphics[width=0.8\textwidth]{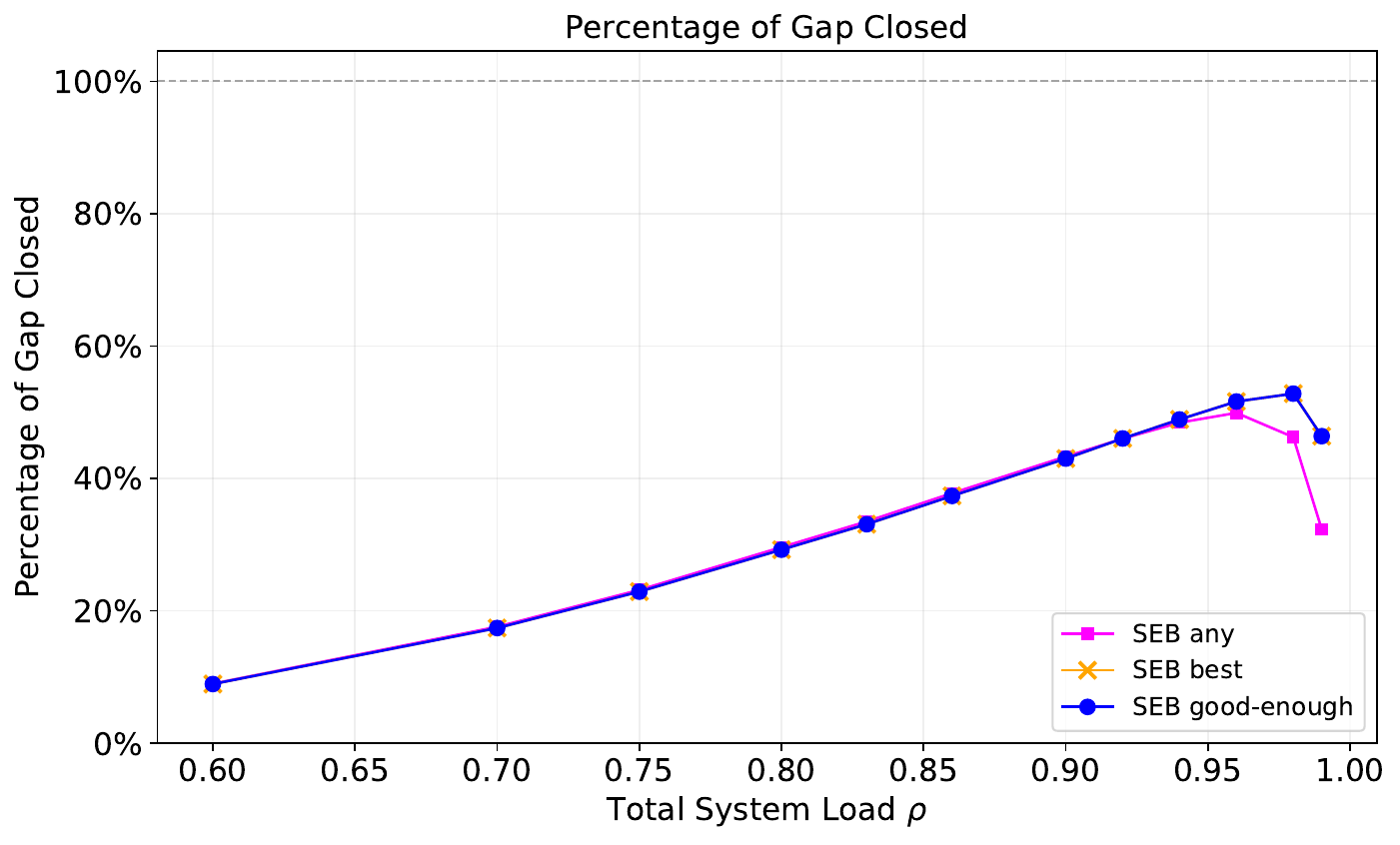}
    \caption{Evaluation results in the 2-class MRJ setting: percentage of gap closed. The curves for SEB any, SEB best and SEB good-enough overlap across most loads because the two policies have nearly identical response times throughout the evaluated range.}
    \label{fig:MRJ_2_classes_gap_closed}
\end{figure}

In this setting, all dynamic SEB policies outperform all benchmark policies across all loads. All dynamic SEB policies outperform MaxWeight-Queue SRPT by noticeable margins under most loads. Dynamic SEB-best and dynamic SEB-good-enough outperform dynamic SEB-any under high loads, suggesting that in this setting, equalization benefits performance.

In summary, our simulations show that dynamic SEB policies consistently demonstrate superior performances in the compability, MSJ, and MRJ systems that are important special cases of the generalized switch across a wide range of system load, often outperforming the benchmark policies by substantial margins. 

\section{Conclusion}
In this paper, we design a new scheduling policy, called SEB (Smallest Equalizing Bucket) and show that it is heavy-traffic optimal, thus making SEB the first proven heavy-traffic optimal scheduling policy in multi-server systems with general service constraints. SEB strikes the right balance between prioritizing small jobs and keeping all servers busy, a critical component to optimality that no previous policies are able to achieve. Finally, we propose three dynamic SEB policies and confirm their excellent performances via simulations.

Despite the success of SEB, we have left in place a few assumptions. Most importantly, we have assumed bounded job duration distribution and independence of job priority index and job class. Designing and analyzing a good policy without these assumptions is important both in theory and in practice. Our explorations suggest that substantial new ideas must be introduced to achieve optimality in the most general setting. We leave this to future work.

\bibliographystyle{abbrvnat}
\bibliography{reference.bib}

@article{eryilmaz2012asymptotically,
  title={Asymptotically tight steady-state queue length bounds implied by drift conditions},
  author={Eryilmaz, Atilla and Srikant, Rayadurgam},
  journal={Queueing Systems},
  volume={72},
  pages={311--359},
  year={2012},
  publisher={Springer}
}

@article{hajek1982hitting,
  title={Hitting-time and occupation-time bounds implied by drift analysis with applications},
  author={Hajek, Bruce},
  journal={Advances in Applied probability},
  volume={14},
  number={3},
  pages={502--525},
  year={1982},
  publisher={Cambridge University Press}
}

@article{miyazawa1994rate,
  title={Rate conservation laws: a survey},
  author={Miyazawa, Masakiyo},
  journal={Queueing Systems},
  volume={15},
  pages={1--58},
  year={1994},
  publisher={Springer}
}

@article{grosof2018srpt,
  title={{SRPT} for multiserver systems},
  author={Grosof, Isaac and Scully, Ziv and Harchol-Balter, Mor},
  journal={Performance Evaluation},
  volume={127},
  number={128},
  pages={154--175},
  year={2018}
}

@article{stolyar2004maxweight,
author = {Alexander L. Stolyar},
title = {{MaxWeight scheduling in a generalized switch: State space collapse and workload minimization in heavy traffic}},
volume = {14},
journal = {The Annals of Applied Probability},
number = {1},
publisher = {Institute of Mathematical Statistics},
pages = {1 -- 53},
year = {2004},
doi = {},
URL = {}
}

@article{grosof2022optimal,
  title={Optimal scheduling in the multiserver-job model under heavy traffic},
  author={Grosof, Isaac and Scully, Ziv and Harchol-Balter, Mor and Scheller-Wolf, Alan},
  journal={Proceedings of the ACM on Measurement and Analysis of Computing Systems},
  volume={6},
  number={3},
  pages={1--32},
  year={2022},
  publisher={ACM New York, NY, USA}
}

@article{psychas2018randomized,
  author={Psychas, Konstantinos and Ghaderi, Javad},
  journal={IEEE/ACM Transactions on Networking},
  title={Randomized Algorithms for Scheduling Multi-Resource Jobs in the Cloud},
  year={2018},
  volume={26},
  number={5},
  pages={2202-2215},
  _doi={}}

@article{grosof2022wcfs,
  title={{WCFS}: A new framework for analyzing multiserver systems},
  author={Grosof, Isaac and Harchol-Balter, Mor and Scheller-Wolf, Alan},
  journal={Queueing Systems},
  volume={102},
  number={1},
  pages={143--174},
  year={2022},
  publisher={Springer}
}

@article{grosof2023reset,
  title={The {RESET} and {MARC} techniques, with application to multiserver-job analysis},
  author={Grosof, Isaac and Hong, Yige and Harchol-Balter, Mor and Scheller-Wolf, Alan},
  journal={Performance Evaluation},
  volume={162},
  pages={102378},
  year={2023},
  publisher={Elsevier}
}

@book{brondsted2012introduction,
  title={An introduction to convex polytopes},
  author={Brondsted, Arne},
  volume={90},
  year={2012},
  publisher={Springer Science \& Business Media}
}

@article{harchol2022multiserver,
  title={The multiserver job queueing model},
  author={Harchol-Balter, Mor},
  journal={Queueing Systems},
  volume={100},
  number={3},
  pages={201--203},
  year={2022},
  publisher={Springer}
}

@inproceedings{morozov2016stability,
  title={Stability analysis of a {MAP/M/s} cluster model by matrix-analytic method},
  author={Morozov, Evsey and Rumyantsev, Alexander},
  booktitle={Computer Performance Engineering: 13th European Workshop, EPEW 2016, Chios, Greece, October 5-7, 2016, Proceedings 13},
  pages={63--76},
  year={2016},
  organization={Springer}
}

@article{rumyantsev2017stability,
  title={Stability criterion of a multiserver model with simultaneous service},
  author={Rumyantsev, Alexander and Morozov, Evsey},
  journal={Annals of Operations Research},
  volume={252},
  pages={29--39},
  year={2017},
  publisher={Springer}
}

@article{afanaseva2020stability,
  title={Stability analysis of a multi-server model with simultaneous service and a regenerative input flow},
  author={Afanaseva, Larisa and Bashtova, Elena and Grishunina, Svetlana},
  journal={Methodology and Computing in Applied Probability},
  volume={22},
  pages={1439--1455},
  year={2020},
  publisher={Springer}
}

@article{brill1984queues,
  title={Queues in which customers receive simultaneous service from a random number of servers: a system point approach},
  author={Brill, Percy H and Green, Linda},
  journal={Management Science},
  volume={30},
  number={1},
  pages={51--68},
  year={1984},
  publisher={INFORMS}
}

@article{filippopoulos2007m,
  title={An {M/M/2} parallel system model with pure space sharing among rigid jobs},
  author={Filippopoulos, Dimitrios and Karatza, Helen},
  journal={Mathematical and Computer Modelling},
  volume={45},
  number={5-6},
  pages={491--530},
  year={2007},
  publisher={Elsevier}
}

@book{kim1979m,
  title={{M/M/s} queueing system where customers demand multiple server use},
  author={Kim, Sung Shick},
  year={1979},
  publisher={Southern Methodist University}
}

@inproceedings{ghaderi2016randomized,
  title={Randomized algorithms for scheduling {VMs} in the cloud},
  author={Ghaderi, Javad},
  booktitle={IEEE INFOCOM 2016-The 35th Annual IEEE International Conference on Computer Communications},
  pages={1--9},
  year={2016},
  organization={IEEE}
}

@inproceedings{jones1999scheduling,
  title={Scheduling for parallel supercomputing: A historical perspective of achievable utilization},
  author={Jones, James Patton and Nitzberg, Bill},
  booktitle={Workshop on job scheduling strategies for parallel processing},
  pages={1--16},
  year={1999},
  organization={Springer}
}

@inproceedings{wang2009application,
  title={The application of backfilling in cluster systems},
  author={Wang, Juan and Guo, Wenming},
  booktitle={2009 WRI International Conference on Communications and Mobile Computing},
  volume={3},
  pages={55--59},
  year={2009},
  organization={IEEE}
}

@inproceedings{carastan2019one,
  title={One can only gain by replacing {EASY Backfilling}: A simple scheduling policies case study},
  author={Carastan-Santos, Danilo and De Camargo, Raphael Y and Trystram, Denis and Zrigui, Salah},
  booktitle={2019 19th IEEE/ACM International Symposium on Cluster, Cloud and Grid Computing (CCGRID)},
  pages={1--10},
  year={2019},
  organization={IEEE}
}

@article{guo2018optimal,
  title={Optimal scheduling of {VMs} in queueing cloud computing systems with a heterogeneous workload},
  author={Guo, Mian and Guan, Quansheng and Ke, Wende},
  journal={IEEE Access},
  volume={6},
  pages={15178--15191},
  year={2018},
  publisher={IEEE}
}

@article{maguluri2014heavy,
  title={Heavy traffic optimal resource allocation algorithms for cloud computing clusters},
  author={Maguluri, Siva Theja and Srikant, Rayadurgam and Ying, Lei},
  journal={Performance Evaluation},
  volume={81},
  pages={20--39},
  year={2014},
  publisher={Elsevier}
}

@article{anton2024efficient,
  title={Efficient scheduling in redundancy systems with general service times},
  author={Anton, Elene and Righter, Rhonda and Verloop, Ina Maria},
  journal={Queueing Systems},
  volume={106},
  number={3},
  pages={333--372},
  year={2024},
  publisher={Springer}
}

@article{zylchlinski2023managing,
author = {Zychlinski, Noa and Chan, Carri W. and Dong, Jing},
title = {Managing Queues with Different Resource Requirements},
journal = {Operations Research},
volume = {71},
number = {4},
pages = {1387-1413},
year = {2023},
doi = {},
}

@inproceedings{hong2022sharp,
author = {Hong, Yige and Wang, Weina},
title = {Sharp waiting-time bounds for multiserver jobs},
year = {2022},
isbn = {9781450391658},
publisher = {Association for Computing Machinery},
address = {New York, NY, USA},
url = {},
doi = {},
booktitle = {Proceedings of the Twenty-Third International Symposium on Theory, Algorithmic Foundations, and Protocol Design for Mobile Networks and Mobile Computing},
pages = {161–170},
numpages = {10},
location = {Seoul, Republic of Korea},
series = {MobiHoc '22}
}

@article{grosof2019load,
  title={Load balancing guardrails: Keeping your heavy traffic on the road to low response times},
  author={Grosof, Isaac and Scully, Ziv and Harchol-Balter, Mor},
  journal={Proceedings of the ACM on Measurement and Analysis of Computing Systems},
  volume={3},
  number={2},
  pages={1--31},
  year={2019},
  publisher={ACM New York, NY, USA}
}

@article{scully2020gittins,
  title={The Gittins policy is nearly optimal in the {M/G/k} under extremely general conditions},
  author={Scully, Ziv and Grosof, Isaac and Harchol-Balter, Mor},
  journal={Proceedings of the ACM on Measurement and Analysis of Computing Systems},
  volume={4},
  number={3},
  pages={1--29},
  year={2020},
  publisher={ACM New York, NY, USA}
}

@phdthesis{scully2022new,
  title={A New Toolbox for Scheduling Theory},
  author={Scully, Ziv},
  school={Carnegie Mellon University},
  year={2022},
  month = aug,
  address = {{Pittsburgh, PA}},
  url = {https://ziv.codes/pdf/scully-thesis.pdf},
  urldate = {2022-12-01},
}

@article{righter1990extremal,
  title={On extremal service disciplines in single-stage queueing systems},
  author={Righter, Rhonda and Shanthikumar, J George and Yamazaki, Genji},
  journal={Journal of Applied Probability},
  volume={27},
  number={2},
  pages={409--416},
  year={1990},
  publisher={Cambridge University Press}
}

@article{banerjee2022heavy,
  title={Heavy traffic scaling limits for shortest remaining processing time queues with heavy tailed processing time distributions},
  author={Banerjee, Sayan and Budhiraja, Amarjit and Puha, Amber L},
  journal={The Annals of Applied Probability},
  volume={32},
  number={4},
  pages={2587--2651},
  year={2022},
  publisher={Institute of Mathematical Statistics}
}

@article{lin2011heavy,
  title={Heavy-traffic analysis of mean response time under shortest remaining processing time},
  author={Lin, Minghong and Wierman, Adam and Zwart, Bert},
  journal={Performance Evaluation},
  volume={68},
  number={10},
  pages={955--966},
  year={2011},
  publisher={Elsevier}
}

@book{harchol2013performance,
  title={Performance modeling and design of computer systems: queueing theory in action},
  author={Harchol-Balter, Mor},
  year={2013},
  publisher={Cambridge University Press}
}

@article{rutten2023load,
  title={Load balancing under strict compatibility constraints},
  author={Rutten, Daan and Mukherjee, Debankur},
  journal={Mathematics of Operations Research},
  volume={48},
  number={1},
  pages={227--256},
  year={2023},
  publisher={INFORMS}
}

@inproceedings{tsitsiklis2013queueing,
  title={Queueing system topologies with limited flexibility},
  author={Tsitsiklis, John N and Xu, Kuang},
  booktitle={Proceedings of the ACM SIGMETRICS/international conference on Measurement and modeling of computer systems},
  pages={167--178},
  year={2013}
}

@article{weng2020optimal,
  title={Optimal load balancing with locality constraints},
  author={Weng, Wentao and Zhou, Xingyu and Srikant, Rayadurgam},
  journal={Proceedings of the ACM on Measurement and Analysis of Computing Systems},
  volume={4},
  number={3},
  pages={1--37},
  year={2020},
  publisher={ACM New York, NY, USA}
}

@article{mukherjee2018asymptotically,
  title={Asymptotically optimal load balancing topologies},
  author={Mukherjee, Debankur and Borst, Sem C and Van Leeuwaarden, Johan SH},
  journal={Proceedings of the ACM on Measurement and Analysis of Computing Systems},
  volume={2},
  number={1},
  pages={1--29},
  year={2018},
  publisher={ACM New York, NY, USA}
}

@article{harrison1999heavy,
  title={Heavy traffic resource pooling in parallel-server systems},
  author={Harrison, J Michael and L{\'o}pez, Marcel J},
  journal={Queueing systems},
  volume={33},
  pages={339--368},
  year={1999},
  publisher={Springer}
}

@article{hurtado2022heavy,
author = {Hurtado Lange, Daniela Andrea and Maguluri, Siva Theja},
title = {Heavy-Traffic Analysis of Queueing Systems with No Complete Resource Pooling},
journal = {Mathematics of Operations Research},
volume = {47},
number = {4},
pages = {3129-3155},
year = {2022},
doi = {10.1287/moor.2021.1248},
}

@ARTICLE{jhunjhunwala2022low,
  author={Jhunjhunwala, Prakirt Raj and Maguluri, Siva Theja},
  journal={IEEE/ACM Transactions on Networking}, 
  title={Low-Complexity Switch Scheduling Algorithms: Delay Optimality in Heavy Traffic}, 
  year={2022},
  volume={30},
  number={1},
  pages={464-473},
  doi={10.1109/TNET.2021.3116606}}

@article{schrage1968proof,
  title={A proof of the optimality of the shortest remaining processing time discipline},
  author={Schrage, Linus},
  journal={Operations Research},
  volume={16},
  number={3},
  pages={687--690},
  year={1968},
  publisher={INFORMS}
}

@article{mckeown2002achieving,
  title={Achieving 100\% throughput in an input-queued switch},
  author={McKeown, Nick and Mekkittikul, Adisak and Anantharam, Venkat and Walrand, Jean},
  journal={IEEE Transactions on Communications},
  volume={47},
  number={8},
  pages={1260--1267},
  year={2002},
  publisher={IEEE}
}

@inproceedings{tassiulas1990stability,
  title={Stability properties of constrained queueing systems and scheduling policies for maximum throughput in multihop radio networks},
  author={Tassiulas, Leandros and Ephremides, Anthony},
  booktitle={29th IEEE Conference on Decision and Control},
  pages={2130--2132},
  year={1990},
  organization={IEEE}
}

@article{chen2025improving,
  title={Improving multiresource job scheduling with markovian service rate policies},
  author={Chen, Zhongrui and Grosof, Isaac and Berg, Benjamin},
  journal={Proceedings of the ACM on Measurement and Analysis of Computing Systems},
  volume={9},
  number={2},
  pages={1--36},
  year={2025},
  publisher={ACM New York, NY, USA}
}

@article{anton2022scheduling,
  title={Scheduling under redundancy},
  author={Anton, Elene and Righter, Rhonda and Verloop, Ina Maria},
  journal={ACM SIGMETRICS Performance Evaluation Review},
  volume={50},
  number={2},
  pages={30--32},
  year={2022},
  publisher={ACM New York, NY, USA}
}

@inproceedings{meyn1993survey,
  title={A survey of Foster-Lyapunov techniques for general state space Markov processes},
  author={Meyn, Sean P and Tweedie, RL},
  booktitle={Proceedings of the Workshop on Stochastic Stability and Stochastic Stabilization, Metz, France},
  year={1993}
}

\begin{APPENDICES}
\section{Deferred Proofs}

\subsection{Proof of \cref{prop:WINE}}
\label{app:WINE}
\restate*\ref{prop:WINE}
\begin{proof}{\textit{Proof.}~}
We will work with $\mathbf{k}^{\boldsymbol{\nu}}$-mapped priority index here. Let $N^{\pi}$ be the number of jobs in the queueing system under $\pi$ in steady state. Let $u_1,\ldots, u_{N^{\pi}}$ be the remaining priority indices of these $N^{\pi}$ jobs. Then we have
\[
    W_{\remduration \leq x}^\pi = \sum_{i=1}^{N^{\pi}} u_i\mathbf{1}(u_i\leq x)
\]
It follows that
\[
    \int_0^\infty \frac{W_{\remduration \leq x}^\pi}{x^2}\,dx = \sum_{i=1}^{N^{\pi}}\int_0^\infty\frac{u_i\mathbf{1}(u_i\leq x)}{x^2}\,dx = \sum_{i=1}^{N^{\pi}}\int_{u_i}^\infty\frac{u_i}{x^2}\,dx=N^{\pi}
\]
Taking expectation with respect to the stationary distribution and applying Tonelli's theorem give us
\[
    \E[N^{\pi}] = \E\left[\int_0^\infty \frac{W_{\remduration \leq x}^\pi}{x^2}\,dx\right] \stackrel{(a)}{=} \int_{0}^{\infty} \frac{\E[W_{\remduration \leq x}^\pi]}{x^2} \, dx
\]
The proposition now follows from Little's law.
\hfill\Halmos
\end{proof}

\subsection{Proof of \cref{thm:stability}}
\label{app:stability}
Throughout this section, we consider the space $\mathbb{R}_{\geq0}^{n_b\times n_c}$ which contains all possible values of bucket work vectors $\{\mathbf{w}^{(1)},\ldots, \mathbf{w}^{(n_b)}\}$. For simplicity, we will use $\{\mathbf{w}^{(i)}\}$ to denote this collection of bucket work vectors. Moreover, we define, for a work vector $\mathbf{w}\in\mathbb{R}_{\geq0}^{n_c}$,
\[
H(\mathbf{w}):=\left(\|\mathbf{w}_{\perp {\boldsymbol{\rho}^{\boldsymbol{\nu}}}}\|_2-\|\mathbf{w}_{\parallel{\boldsymbol{\rho}^{\boldsymbol{\nu}}}}\|_2\tan\varphi\right)^+
\]
which roughly measures the distance between $\mathbf{w}$ and the cone $\mathcal{C}$. 

Before we establish the stability of the system under SEB, we first prove two lemmas. The first lemma implies that if the total work in system is large and all bucket work vectors are close to the cone, then at least one bucket is eligible for service.

\begin{lemma}
\label{lemma:bucket_service_eligibility}
For any $R>0$, if $H(\mathbf{w}^{(i)})\leq R$ for all $i=\{1,\ldots,n_b\}$ and 
\[
    \sum_{i=1}^{n_b}\sum_{j=1}^{n_c}w^{(i)}_j\geq n_b\cdot\max\left\{\sqrt{n_c}\left(\frac{(1+\tan\varphi)(n_s z_{\max} + R)}{\frac{\rho_{\min}}{\|{\boldsymbol{\rho}^{\boldsymbol{\nu}}}\|_2} - \tan\varphi} + R\right),\frac{n_s\sqrt{n_c} z_{\max}}{\frac{\rho_{\min}}{\|{\boldsymbol{\rho}^{\boldsymbol{\nu}}}\|_2}\cos\varphi-\sin\varphi}\right\}\eqqcolon L
\]
then at least one bucket is eligible for service.
\end{lemma}

\begin{proof}{\textit{Proof.}}
Since the total work in system is no less than $L$, there exists $j\in\{1,\ldots,n_b\}$ such that $\|\mathbf{w}^{(j)}\|_1\geq L/n_b$. We now show that bucket $j$ must be eligible for service. 

If $\mathbf{w}^{(j)}\not\in\mathcal{C}$, we have $\|\mathbf{w}^{(j)}_{\perp{\boldsymbol{\rho}^{\boldsymbol{\nu}}}}\|_2\leq R+\|\mathbf{w}^{(j)}_{\parallel{\boldsymbol{\rho}^{\boldsymbol{\nu}}}}\|_2\tan\varphi$. Since $\mathbf{w}^{(j)}=\mathbf{w}^{(j)}_{\|{\boldsymbol{\rho}^{\boldsymbol{\nu}}}}+\mathbf{w}^{(j)}_{\perp{\boldsymbol{\rho}^{\boldsymbol{\nu}}}}$, letting $\mathbf{e}_k$ be the $k$-th standard basis in $\mathbb{R}^{n_c}$, we have
\[
\inner{\mathbf{w}^{(j)}}{\mathbf{e}_k}&=\inner{\mathbf{w}^{(j)}_{\parallel{\boldsymbol{\rho}^{\boldsymbol{\nu}}}}}{\mathbf{e}_k}+\inner{\mathbf{w}^{(j)}_{\perp{\boldsymbol{\rho}^{\boldsymbol{\nu}}}}}{\mathbf{e}_i}\\
&\geq\frac{\inner{\mathbf{w}^{(j)}}{{\boldsymbol{\rho}^{\boldsymbol{\nu}}}}}{\|{\boldsymbol{\rho}^{\boldsymbol{\nu}}}\|_2^2}\rho_k-\|\mathbf{w}^{(j)}_{\perp{\boldsymbol{\rho}^{\boldsymbol{\nu}}}}\|_2\\
&\geq\frac{\inner{\mathbf{w}^{(j)}}{{\boldsymbol{\rho}^{\boldsymbol{\nu}}}}}{\|{\boldsymbol{\rho}^{\boldsymbol{\nu}}}\|_2^2}\rho_k-R-\frac{\inner{\mathbf{w}^{(j)}}{{\boldsymbol{\rho}^{\boldsymbol{\nu}}}}}{\|{\boldsymbol{\rho}^{\boldsymbol{\nu}}}\|_2^2}\|{\boldsymbol{\rho}^{\boldsymbol{\nu}}}\|_2\tan\varphi\\
&=\frac{\inner{\mathbf{w}^{(j)}}{{\boldsymbol{\rho}^{\boldsymbol{\nu}}}}}{\|{\boldsymbol{\rho}^{\boldsymbol{\nu}}}\|_2^2}(\rho_k-\|{\boldsymbol{\rho}^{\boldsymbol{\nu}}}\|_2\tan\varphi)-R\\
&\geq\|\mathbf{w}^{(j)}_{\parallel{\boldsymbol{\rho}^{\boldsymbol{\nu}}}}\|_2\underbrace{\left(\frac{\rho_{\min}}{\|{\boldsymbol{\rho}^{\boldsymbol{\nu}}}\|_2}-\tan\varphi\right)}_{\mathcal{T}}-R
\]
We note that $\mathcal{T}>0$ by assumption 3 on $\varphi$ in \Cref{sec:cone_SSC}. By norm inequality,
\[
\|\mathbf{w}^{(j)}_{\parallel{\boldsymbol{\rho}^{\boldsymbol{\nu}}}}\|_2&\geq\|\mathbf{w}^{(j)}\|_2-\|\mathbf{w}^{(j)}_{\perp{\boldsymbol{\rho}^{\boldsymbol{\nu}}}}\|_2\geq\frac{1}{\sqrt{n_c}}\|\mathbf{w}^{(j)}\|_1-R-\|\mathbf{w}^{(j)}_{\parallel{\boldsymbol{\rho}^{\boldsymbol{\nu}}}}\|_2\tan\varphi
\]
The following bound follows:
\[
\|\mathbf{w}^{(j)}_{\parallel{\boldsymbol{\rho}^{\boldsymbol{\nu}}}}\|_2\geq\frac{1}{1+\tan\varphi}\left(\frac{1}{\sqrt{n_c}}\|\mathbf{w}^{(j)}\|_1-R\right)
\]
Therefore, because
\[
    \|\mathbf{w}^{(j)}\|_1 \geq \frac{L}{n_b} \geq \sqrt{n_c}\left(\frac{(1+\tan\varphi)(n_s b_j + R)}{\frac{\rho_{\min}}{\|{\boldsymbol{\rho}^{\boldsymbol{\nu}}}\|_2} - \tan\varphi} + R\right)
\]
we have $\mathbf{w}^{(j)}_k\geq n_sb_{j}$ for all classes $k=\{1,\ldots,n_c\}$, which implies that bucket $j$ is eligible for service because there are at least $n_s$ jobs of each class in the bucket.

If $\mathbf{w}^{(j)}\in\mathcal{C}$, then by \Cref{thm:eligibility_work_amount}, bucket $j$ is eligible for service if
\[
\|\mathbf{w}^{(j)}\|_1\geq\frac{n_s\sqrt{n_c} b_{j}}{\frac{\rho_{\min}}{\|{\boldsymbol{\rho}^{\boldsymbol{\nu}}}\|_2}\cos\varphi-\sin\varphi}
\]
\hfill\Halmos
\end{proof}

The second lemma is first proved in \citet{eryilmaz2012asymptotically} under a slightly different setting. We show that this holds here as well.

\begin{lemma}\label{lemma:norm_change_bound}
For vectors $\mathbf{w}$ and $\mathbf{z}$ such that $\mathbf{w}_{\perp{\boldsymbol{\rho}^{\boldsymbol{\nu}}}}\neq{\boldsymbol{\rho}^{\boldsymbol{\nu}}}$, we have
\[
\|(\mathbf{w}+\mathbf{z})_{\perp{\boldsymbol{\rho}^{\boldsymbol{\nu}}}}\|_2-\|\mathbf{w}_{\perp{\boldsymbol{\rho}^{\boldsymbol{\nu}}}}\|_2\leq\frac{1}{2\|\mathbf{w}_{\perp{\boldsymbol{\rho}^{\boldsymbol{\nu}}}}\|_2}\left[(\|\mathbf{w}+\mathbf{z}\|_2^2-\|\mathbf{w}\|_2^2)-(\|(\mathbf{w}+\mathbf{z})_{\parallel{\boldsymbol{\rho}^{\boldsymbol{\nu}}}}\|_2^2-\|\mathbf{w}_{\parallel{\boldsymbol{\rho}^{\boldsymbol{\nu}}}}\|_2^2)\right]
\]
\end{lemma}
\begin{proof}{\textit{Proof.}}
Note that we have $\|\mathbf{x}\|_2=\sqrt{\|\mathbf{x}\|_2^2}$ and the square root function $x\mapsto\sqrt{x}$ is concave. Thus,
\[
\|(\mathbf{w}+\mathbf{z})_{\perp{\boldsymbol{\rho}^{\boldsymbol{\nu}}}}\|_2-\|\mathbf{w}_{\perp{\boldsymbol{\rho}^{\boldsymbol{\nu}}}}\|_2&=\sqrt{\|(\mathbf{w}+\mathbf{z})_{\perp{\boldsymbol{\rho}^{\boldsymbol{\nu}}}}\|_2^2}-\sqrt{\|\mathbf{w}_{\perp{\boldsymbol{\rho}^{\boldsymbol{\nu}}}}\|_2^2}\leq\frac{1}{2\|\mathbf{w}_{\perp{\boldsymbol{\rho}^{\boldsymbol{\nu}}}}\|_2}\left(\|(\mathbf{w}+\mathbf{z})_{\perp{\boldsymbol{\rho}^{\boldsymbol{\nu}}}}\|_2^2-\|\mathbf{w}_{\perp{\boldsymbol{\rho}^{\boldsymbol{\nu}}}}\|_2^2\right)
\]
The lemma now follows from the Pythagorean theorem
\[
\|(\mathbf{w}+\mathbf{z})_{\perp{\boldsymbol{\rho}^{\boldsymbol{\nu}}}}\|_2^2-\|\mathbf{w}_{\perp{\boldsymbol{\rho}^{\boldsymbol{\nu}}}}\|_2^2=\|\mathbf{w}+\mathbf{z}\|_2^2-\|(\mathbf{w}+\mathbf{z})_{\parallel{\boldsymbol{\rho}^{\boldsymbol{\nu}}}}\|_2^2-\|\mathbf{w}\|_2^2+\|\mathbf{w}_{\parallel{\boldsymbol{\rho}^{\boldsymbol{\nu}}}}\|_2^2
\]
\hfill\Halmos
\end{proof}

\restate*\ref{thm:stability}

\begin{proof}{\textit{Proof.}}
Our main tool is the continuous-time Foster-Lypunov theorem in a general state space (Theorem 6 in \citep{meyn1993survey}). The key is to find a differentiable nonnegative Lyapunov function $V$ that has bounded drift on a compact set $\mathcal{K}$ and negative drift outside $\mathcal{K}$.

Consider the following Lyapunov function:
\[
V(\{\mathbf{w}^{(i)}\})=\sum_{i=1}^{n_b}H(\mathbf{w}^{(i)})^2+M\sum_{i=1}^{n_b}\sum_{j=1}^{n_c}w^{(i)}_j
\]
That is, $V$ is a sum of total distance of bucket work vectors to cones and the total work in the system.

We define the following compact set:
\[
\mathcal{K}=\{\{\mathbf{w}^{(i)}\}: H(\mathbf{w}^{(i)})\leq R \text{ for all $i=\{1,\ldots,n_b\}$}\}\cap\left\{\{\mathbf{w}^{(i)}\}: \sum_{i=1}^{n_b}\sum_{j=1}^{n_c}w^{(i)}_j\leq L\right\}
\]
where $L$ is as specified in \Cref{lemma:bucket_service_eligibility}.

$M$ and $R$ will be determined at the end of the proof. As the notations quickly become cumbersome, we will introduce simplifying notations over the course of the proof.

Since the job priority index distribution is bounded, $V$ naturally has a bounded drift everywhere. We will show that $\mathcal{G}(\{\mathbf{w}^{(i)}\})\leq-1$ for $\{\mathbf{w}^{(i)}\}\in\mathcal{K}^c$.

We start by analyzing the change in $V$ due to a job arrival. The expected change is
\[
    \lambda\E[\Delta V(\{\mathbf{w}^{(i)}\})] = \sum_{i=1}^{n_b}\lambda_i\E[H(\mathbf{w}^{(i)}+\mathbf{Z}^{(i)})^2 - H(\mathbf{w}^{(i)})^2] + M
\]

We first focus on the change in $H(\mathbf{w}^{(i)})$ for a specific $i$. We have the following cases:

Case 1: $H(\mathbf{w}^{(i)}) > 0$ and $H(\mathbf{w}^{(i)}+\mathbf{z}^{(i)})=0$.

In this case, 
\[
    H(\mathbf{w}^{(i)}+\mathbf{z}^{(i)})^2 - H(\mathbf{w}^{(i)})^2=-H(\mathbf{w}^{(i)})^2
\]

As we will see, since drift terms from other cases are linear in $-H(\mathbf{w}^{(i)})$, we may enlarge $R$ so that this case can be dropped in the final drift bound.\\

Case 2: $H(\mathbf{w}^{(i)}) = 0$ and $H(\mathbf{w}^{(i)}+\mathbf{z}^{(i)})=0$.

In this case, 
\[
H(\mathbf{w}^{(i)}+\mathbf{z}^{(i)})^2 - H(\mathbf{w}^{(i)})^2=0
\]
\\

Case 3: $H(\mathbf{w}^{(i)}) = 0$ and $H(\mathbf{w}^{(i)}+\mathbf{z}^{(i)}) > 0$.

This case will be handled as we establish the overall arrival drift bound \cref{eq:arrival_crude_bound}.\\

Case 4: $H(\mathbf{w}^{(i)}) > 0$ and $H(\mathbf{w}^{(i)}+\mathbf{z}^{(i)})>0$.
\begin{align}
\lambda_i\E[H(\mathbf{w}^{(i)}+\mathbf{Z}^{(i)})^2 - H(\mathbf{w}^{(i)})^2] &= \lambda_i\E[(H(\mathbf{w}^{(i)}+\mathbf{Z}^{(i)}) + H(\mathbf{w}^{(i)}))(H(\mathbf{w}^{(i)}+\mathbf{Z}^{(i)}) - H(\mathbf{w}^{(i)}))] \nonumber \\
& \stackrel{(a)}{\leq}(2H(\mathbf{w}^{(i)})+\sqrt{n_c}z_{\max})\lambda_i\E[H(\mathbf{w}^{(i)}+\mathbf{Z}^{(i)}) - H(\mathbf{w}^{(i)})] \nonumber\\
& \stackrel{(b)}{\leq}(2\|\mathbf{w}^{(i)}_{\perp{\boldsymbol{\rho}^{\boldsymbol{\nu}}}}\|_2+\sqrt{n_c}z_{\max})\left(\frac{\lambda_i n_cz_{\max}^2}{2\|\mathbf{w}_{\perp{\boldsymbol{\rho}^{\boldsymbol{\nu}}}}^{(i)}\|_2}-\|{\boldsymbol{\rho}^{\boldsymbol{\nu}}}^{(i)}\|_2\tan\varphi\right) \nonumber\\
& \stackrel{(c)}{\leq}-\alpha H(\mathbf{w}^{(i)})+\frac{\beta}{H(\mathbf{w}^{(i)})}+\gamma   \label{eq:arrival_drift}
\end{align}

where (a) follows from
\[
    H(\mathbf{w}^{(i)}+\mathbf{z}^{(i)}) + H(\mathbf{w}^{(i)}) &= \|(\mathbf{w}^{(i)}+\mathbf{z}^{(i)})_{\perp{\boldsymbol{\rho}^{\boldsymbol{\nu}}}}\|_2 - \tan\varphi\left(\|\mathbf{w}^{(i)}_{\parallel{\boldsymbol{\rho}^{\boldsymbol{\nu}}}}\|_2+\|\mathbf{z}^{(i)}_{\parallel{\boldsymbol{\rho}^{\boldsymbol{\nu}}}}\|_2\right)+H(\mathbf{w}^{(i)})\\
    &\leq (\|\mathbf{w}_{\perp{\boldsymbol{\rho}^{\boldsymbol{\nu}}}}^{(i)}\|_2-\tan\varphi\|\mathbf{w}_{\parallel{\boldsymbol{\rho}^{\boldsymbol{\nu}}}}^{(i)}\|_2)+\|\mathbf{z}_{\perp{\boldsymbol{\rho}^{\boldsymbol{\nu}}}}^{(i)}\|_2-\tan\varphi\|\mathbf{z}_{\parallel{\boldsymbol{\rho}^{\boldsymbol{\nu}}}}^{(i)}\|_2+H(\mathbf{w}^{(i)})\\
    &\leq 2H(\mathbf{w}^{(i)}) + \sqrt{n_c}z_{\max}
\]
(b) follows from the following bound: By \cref{lemma:norm_change_bound},
\[
\|(\mathbf{w}^{(i)}+\mathbf{z}^{(i)})_{\perp{\boldsymbol{\rho}^{\boldsymbol{\nu}}}}\|_2-\|\mathbf{w}^{(i)}_{\perp{\boldsymbol{\rho}^{\boldsymbol{\nu}}}}\|_2&\leq\frac{1}{2\|\mathbf{w}^{(i)}_{\perp{\boldsymbol{\rho}^{\boldsymbol{\nu}}}}\|_2}\left(2\inner{\mathbf{w}^{(i)}}{\mathbf{z}^{(i)}}+\|\mathbf{z}^{(i)}\|_2^2-2\frac{\inner{\mathbf{w}^{(i)}}{{\boldsymbol{\rho}^{\boldsymbol{\nu}}}^{(i)}}\inner{\mathbf{z}^{(i)}}{{\boldsymbol{\rho}^{\boldsymbol{\nu}}}^{(i)}}}{\|{\boldsymbol{\rho}^{\boldsymbol{\nu}}}^{(i)}\|_2^2}-\|\mathbf{z}^{(i)}_{\parallel{\boldsymbol{\rho}^{\boldsymbol{\nu}}}}\|_2^2\right)\\
&\leq\frac{1}{\|\mathbf{w}^{(i)}_{\perp{\boldsymbol{\rho}^{\boldsymbol{\nu}}}}\|_2}\left(\inner{\mathbf{w}^{(i)}}{\mathbf{z}^{(i)}}-\frac{\inner{\mathbf{w}^{(i)}}{{\boldsymbol{\rho}^{\boldsymbol{\nu}}}^{(i)}}\inner{\mathbf{z}^{(i)}}{{\boldsymbol{\rho}^{\boldsymbol{\nu}}}^{(i)}}}{\|{\boldsymbol{\rho}^{\boldsymbol{\nu}}}^{(i)}\|_2^2}+\frac{1}{2}\|\mathbf{z}^{(i)}_{\perp{\boldsymbol{\rho}^{\boldsymbol{\nu}}}}\|_2^2\right)
\]
and
\[
\|(\mathbf{w}^{(i)}+\mathbf{z}^{(i)})_{\parallel{\boldsymbol{\rho}^{\boldsymbol{\nu}}}}\|_2-\|\mathbf{w}^{(i)}_{\parallel{\boldsymbol{\rho}^{\boldsymbol{\nu}}}}\|_2=\frac{\inner{\mathbf{z}^{(i)}}{{\boldsymbol{\rho}^{\boldsymbol{\nu}}}^{(i)}}}{\|{\boldsymbol{\rho}^{\boldsymbol{\nu}}}^{(i)}\|_2}
\]
Thus,
\[
H(\mathbf{w}^{(i)}+\mathbf{z}^{(i)}) - H(\mathbf{w}^{(i)})\leq&\,\frac{1}{\|\mathbf{w}^{(i)}_{\perp{\boldsymbol{\rho}^{\boldsymbol{\nu}}}}\|_2}\left(\inner{\mathbf{w}^{(i)}}{\mathbf{z}^{(i)}}-\frac{\inner{\mathbf{w}^{(i)}}{{\boldsymbol{\rho}^{\boldsymbol{\nu}}}}\inner{\mathbf{z}^{(i)}}{{\boldsymbol{\rho}^{\boldsymbol{\nu}}}^{(i)}}}{\|{\boldsymbol{\rho}^{\boldsymbol{\nu}}}^{(i)}\|_2^2}+\frac{1}{2}\|\mathbf{z}_{\perp{\boldsymbol{\rho}^{\boldsymbol{\nu}}}}\|_2^2\right)-\\
&\,\frac{\inner{\mathbf{z}^{(i)}}{{\boldsymbol{\rho}^{\boldsymbol{\nu}}}^{(i)}}}{\|{\boldsymbol{\rho}^{\boldsymbol{\nu}}}^{(i)}\|_2}\tan\varphi
\]
Since $\lambda_i\E[\mathbf{Z}^{(i)}]={\boldsymbol{\rho}^{\boldsymbol{\nu}}}^{(i)}$, we have
\[
\label{eq:H_arrival_drift_bound}
\lambda_i\E[H(\mathbf{w}^{(i)}+\mathbf{Z}^{(i)}) - H(\mathbf{w}^{(i)})]\leq \frac{\lambda\E[\|\mathbf{Z}_{\perp{\boldsymbol{\rho}^{\boldsymbol{\nu}}}}\|_2^2]}{2\|\mathbf{w}^{(i)}_{\perp{\boldsymbol{\rho}^{\boldsymbol{\nu}}}}\|_2}-\|{\boldsymbol{\rho}^{\boldsymbol{\nu}}}^{(i)}\|_2\tan\varphi
\]
(c) follows from $H(\mathbf{w}^{(i)})\leq \|\mathbf{w}^{(i)}_{\perp\rho}\|_2$ and defining
\[
    \alpha &:= 2(\min_{i}\|{\boldsymbol{\rho}^{\boldsymbol{\nu}}}^{(i)}\|_2)\tan\varphi\\
    \beta &:= \frac{\lambda n_c^{3/2}z_{\max}^3}{2}\\
    \gamma &:= \max\{\max_i\left\{n_cz_{\max}^2 - \sqrt{n_c}z_{\max}\|{\boldsymbol{\rho}^{\boldsymbol{\nu}}}^{(i)}\|_2\tan\varphi\}, 0\right\}
\]

We note that \cref{eq:arrival_drift} may not be desirable when $H(\mathbf{w}^{(i)})$ is small due to the $1/H(\mathbf{w}^{(i)})$ term. So we consider the set:
\[
    B:=\{i: H(\mathbf{w}^{(i)}) \geq \eta\}
\]
for some $\eta>0$. On $B^c$, we use a crude bound for the drift. Note that case 3 is included in the analysis of drift on $B^c$.

We first show that $H(\mathbf{w}^{(i)})$ is Lipschitz.
\begin{align}
\label{eq:H_Lipschitz}
    |H(\mathbf{w}^{(i)}+\mathbf{z}^{(i)})-H(\mathbf{w}^{(i)})|&\stackrel{(a)}{\leq}|\|(\mathbf{w}^{(i)}+\mathbf{z}^{(i)})_{\perp{\boldsymbol{\rho}^{\boldsymbol{\nu}}}}\|_2-\|(\mathbf{w}^{(i)}+\mathbf{z}^{(i)})_{\parallel{\boldsymbol{\rho}^{\boldsymbol{\nu}}}}\|_2\tan\varphi-(\|\mathbf{w}^{(i)}_{\perp{\boldsymbol{\rho}^{\boldsymbol{\nu}}}}\|_2-\|\mathbf{w}^{(i)}_{\parallel{\boldsymbol{\rho}^{\boldsymbol{\nu}}}}\|_2\tan\varphi)| \nonumber\\
    &\leq |\|(\mathbf{w}^{(i)}+\mathbf{z}^{(i)})_{\perp{\boldsymbol{\rho}^{\boldsymbol{\nu}}}}\|_2 - \|\mathbf{w}^{(i)}_{\perp{\boldsymbol{\rho}^{\boldsymbol{\nu}}}}\|_2| + \tan\varphi |\|(\mathbf{w}^{(i)}+\mathbf{z}^{(i)})_{\parallel{\boldsymbol{\rho}^{\boldsymbol{\nu}}}}\|_2-\|\mathbf{w}^{(i)}_{\parallel{\boldsymbol{\rho}^{\boldsymbol{\nu}}}}\|_2|\nonumber\\
    &\stackrel{(b)}{\leq}(1+\tan\varphi)\|\mathbf{z}^{(i)}\|_2
\end{align}

where (a) follows from the inequality $|x^+-y^+|\leq|x-y|$ for real numbers $x$ and $y$ and (b) follows from the reverse triangle inequality.

We have the following bound for a bucket $i$ in $B^c$:
\begin{align}
    H(\mathbf{w}^{(i)}+\mathbf{z}^{(i)})^2 - H(\mathbf{w}^{(i)})^2 &= (H(\mathbf{w}^{(i)}+\mathbf{z}^{(i)}) - H(\mathbf{w}^{(i)}))(H(\mathbf{w}^{(i)}+\mathbf{z}^{(i)}) + H(\mathbf{w}^{(i)}))\nonumber\\
    &\stackrel{(a)}{\leq}(1+\tan\varphi)\sqrt{n_c}z_{\max}(2H(\mathbf{w}^{(i)})+\sqrt{n_c}z_{\max})\nonumber\\
    &=\underbrace{2(1+\tan\varphi)\sqrt{n_c}z_{\max}}_{:=A}H(\mathbf{w}^{(i)})+\underbrace{(1+\tan\varphi)n_cz_{\max}^2}_{:=D}\label{eq:arrival_crude_bound}
\end{align}
where (a) follows from the Lipschitz proof and the bound on $H(\mathbf{w}^{(i)}+\mathbf{z}^{(i)}) + H(\mathbf{w}^{(i)})$ in case 4.
The overall drift bound due to arrival is
\begin{align}
    \lambda \E[\Delta V(\{\mathbf{w}^{(i)}\})]&\leq-\alpha\sum_{i\in B}H(\mathbf{w}^{(i)}) + |B|\left(\frac{\beta}{\eta}+\gamma\right)+\lambda A\sum_{i\in B^c}H(\mathbf{w}^{(i)}) + \lambda|B^c|D \nonumber \\
    &\leq-\alpha\sum_{i\in B}H(\mathbf{w}^{(i)}) + \underbrace{n_b\max\left\{\frac{\beta}{\eta}+\gamma,\lambda(A\eta+D)\right\}}_{:=E}+M \nonumber \\
    &\leq-\alpha\max_{i\in B}H(\mathbf{w}^{(i)}) + E + M \label{eq:overall_arrival_bound}
\end{align}
We set $\max_{i\in B}H(\mathbf{w}^{(i)})=0$ when $B=\varnothing$.

When bucket $i$ is in service\footnote{Here, we do not separate the case where there are more than one preferable service options and SEB implements a weighted-processor-sharing (see \Cref{sec:def-SEB}), which is essentially a convex combination of these preferable service options. As a result, the total service rate is still 1 and the inner product between $\nabla H(\mathbf{w}^{(i)})^2$,the processor-sharing is still negative, and \Cref{eq:overall_drift_bound} continues to hold.}, the total drift bound is:
\[
\mathcal{G}V(\{\mathbf{w}^{(i)}\})=&\,-\inner{\nabla H(\mathbf{w}^{(i)})^2}{\mathbf{r}^*}+\lambda \E[\Delta H(\mathbf{w}^{(i)})]-M(1-\rho^{\boldsymbol{\nu}})\\
=&\,-2H(\mathbf{w}^{(i)})\left(\frac{\inner{\mathbf{w}^{(i)}_{\perp{\boldsymbol{\rho}^{\boldsymbol{\nu}}}}}{\mathbf{r}^*}}{\|\mathbf{w}^{(i)}_{\perp{\boldsymbol{\rho}^{\boldsymbol{\nu}}}}\|_2}-\frac{\inner{{\boldsymbol{\rho}^{\boldsymbol{\nu}}}}{\mathbf{r}^*}}{\|{\boldsymbol{\rho}^{\boldsymbol{\nu}}}\|_2}\tan\varphi\right)+\lambda \E[\Delta H(\mathbf{w}^{(i)})]-M(1-\rho^{\boldsymbol{\nu}})\\
\stackrel{(a)}{\leq}&\,\lambda \E[\Delta H(\mathbf{w}^{(i)})]-M(1-\rho^{\boldsymbol{\nu}})\label{eq:overall_drift_bound}
\]
where (a) follows from the choice of $\varphi$ in \Cref{sec:bucket_SSC}.\\

We now show that the drift is negative on $\mathcal{K}^c$. 

Case 1: There exists a bucket $i$ such that $H(\mathbf{w}^{(i)})>R$. It suffices to consider only the arrival drift, as \cref{eq:overall_drift_bound} shows that service will only make the total drift smaller.

In this case, \cref{eq:overall_arrival_bound} implies
\[
    \mathcal{G}V(\{\mathbf{w}\})&\leq -\alpha R + E + M
\]
Thus, we choose 
\[
\label{eq:SEB_stability_R}
    R = \max\left\{\frac{E+M+1}{\alpha},\eta+1\right\}
\]
for all $i$, which leads to a total drift no larger than $-1$.\\

Case 2: $\sum_{i=1}^{n_b}\sum_{j=1}^{n_c}w_j^{(i)} > L$.

We note that if we also have $H(\mathbf{w}^{(i)})>R$ for a bucket $i$, then we reduce to case 1. So it suffices to look at when $H(\mathbf{w}^{(i)})\leq R$ for all $i=\{1,\ldots,n_b\}$. Notice that one bucket must be in service by \Cref{lemma:bucket_service_eligibility}. 

For $\mathbf{w}^{(i)}$ such that $\eta+1/2 < H(\mathbf{w}^{(i)}) \leq R$, we use \cref{eq:overall_arrival_bound} to bound the drift due to arrival

\[
    \lambda\E[\Delta H(\mathbf{w}^{(i)})]\leq -\alpha \left(\eta+\frac{1}{2}\right) + E + M
\]

For $\mathbf{w}^{(i)}$ such that $H(\mathbf{w}^{(i)}) \leq \eta+1/2$, we use \cref{eq:arrival_crude_bound} to bound the drift due to arrival

\[
    \lambda\E[\Delta H(\mathbf{w}^{(i)})]\leq \lambda n_b\left[A\left(\eta+\frac{1}{2}\right) + D\right] + M
\]

Therefore, the overall drift bound is given by \cref{eq:overall_drift_bound}:

\[
    \mathcal{G}V(\{\mathbf{w}^{(i)}\}) \leq -M(1-\rho^{\boldsymbol{\nu}}) + n_b\max\left\{-\alpha \left(\eta+\frac{1}{2}\right) + E, n_b\left[A\left(\eta+\frac{1}{2}\right) + D\right]\right\}
\]

We choose 
\[
\label{eq:SEB_stability_M}
    M = \frac{n_b\max\left\{-\alpha (\eta+1/2) + E, \lambda n_b[A(\eta+1/2) + D]\right\}+1}{1-\rho^{\boldsymbol{\nu}}}   
\]
which leads to a total drift no larger than $-1$.

The choices of $R$, $M$, and $L$ may seem circular. We now give the order in which these parameters are determined.
\begin{enumerate}
    \item Choose an arbitrary $\eta > 0$.
    \item Choose $M$ as in \cref{eq:SEB_stability_M}.
    \item Choose $R$ as in \cref{eq:SEB_stability_R}
    \item Choose $L$ according to \Cref{lemma:bucket_service_eligibility}.
\end{enumerate}
\hfill\Halmos
\end{proof}

\subsection{Proof of \cref{thm:SSC}}
\label{app:SSC}

The following lemma, which was proved in \citet{hajek1982hitting}, is foundational to discrete-time state-space collapse and will play a critical role in our proof for continuous-time state-space collapse.
\begin{lemma}[\citet{hajek1982hitting} Lemma 2.2]\label{Hajek_lemma}
Suppose that $X$ and $Y$ are random variables such that $|X|\leq_{st}Y$ and $\E[e^{\lambda Y}]<\infty$ for some $\theta>0$. Then for $0\leq\varepsilon\leq\theta$,
\begin{equation}\label{Hajek_inequality}
\E[e^{\varepsilon X}]\leq 1+\varepsilon\E[X]+\varepsilon^2c
\end{equation}
where $c$ is given by 
\[c=\sum_{k=2}^{\infty}\frac{\theta^{k-2}}{k!}\E[Y^k]\]
\end{lemma}

\restate*\ref{thm:SSC}
\begin{proof}{\textit{Proof.}}
Fix some positive integer $n$. Applying Rate Conservation Law \citep{miyazawa1994rate} to $e^{\eta(V(\mathbf{W})\wedge n)}$ gives us
\begin{equation}\label{RCL_SSC}
\E\left[\eta D_t(V(\mathbf{W})\wedge n)e^{\eta(V(\mathbf{W})\wedge n)}\right]+\lambda\E\left[e^{\eta(V(\mathbf{W}_+)\wedge n)}-e^{\eta(V(\mathbf{W})\wedge n)}\right]=0
\end{equation}
We first look into the second term on the LHS of \eqref{RCL_SSC} conditioned on $\mathbf{W}$. 
\begin{align*}
&\E\left[e^{\eta(V(\mathbf{W}_+)\wedge n)}-e^{\eta(V(\mathbf{W})\wedge n)}\mid \mathbf{W}\right]\\
=&\,\E\left[\left(e^{\eta(V(\mathbf{W}_+)\wedge n - V(\mathbf{W})\wedge n)}-1\right)e^{\eta(V(\mathbf{W})\wedge n)}\mid \mathbf{W}\right]\\
\stackrel{(a)}{\leq}&\,\E\left[\left(e^{\eta\Delta V(\mathbf{W})\cdot\mathbf{1}(V(\mathbf{W})\leq n)}-1\right)e^{\eta(V(\mathbf{W})\wedge n)}\mid \mathbf{W}\right]\\
=&\,\E\left[\left(e^{\eta\Delta V(\mathbf{W})}-1\right)e^{\eta V(\mathbf{W})}\mathbf{1}(V(\mathbf{W})\leq n)\mid \mathbf{W}\right]\\
=&\,\E\left[\eta\Delta V(\mathbf{W})e^{\eta V(\mathbf{W})}\mathbf{1}(V(\mathbf{W})\leq n)\mid \mathbf{W}\right]+\E\left[\left(e^{\eta\Delta V(\mathbf{W})}-\eta\Delta V(\mathbf{W})-1\right)e^{\eta V(\mathbf{W})}\mathbf{1}(V(\mathbf{W})\leq n)\mid \mathbf{W}\right]\\
\leq&\,\E\left[\eta\Delta V(\mathbf{W})\mid \mathbf{W}\right]e^{\eta V(\mathbf{W})}\mathbf{1}(V(\mathbf{W})\leq n)+\E\left[\left(e^{\eta\Delta V(\mathbf{W})}-\eta\Delta V(\mathbf{W})-1\right)\mid \mathbf{W}\right]e^{\eta V(\mathbf{W})}\mathbf{1}(V(\mathbf{W})\leq n)\\
\stackrel{(b)}{\leq}&\,\E\left[\eta\Delta V(\mathbf{W})\mid \mathbf{W}\right]e^{\eta V(\mathbf{W})}\mathbf{1}(V(\mathbf{W})\leq n)+\eta^2c(\mathbf{w})\,e^{\eta V(\mathbf{W})}\mathbf{1}(V(\mathbf{W})\leq n)
\end{align*}
where\\
(a) follows because we have, for all $\mathbf{w}$,
\[V(\mathbf{W}_+)\wedge n - V(\mathbf{W})\wedge n \leq \Delta V(\mathbf{w})\mathbf{1}(V(\mathbf{w})\leq n)\]
(b) follows from assumption (ii) and \eqref{Hajek_inequality} in \Cref{Hajek_lemma} in that for any $0<\eta\leq\theta$ we have
\[\E\left[e^{\eta\Delta V(\mathbf{W})}-\eta\Delta V(\mathbf{W})-1\mid \mathbf{W}\right]\leq\eta^2c(\mathbf{W})\]
where
\[
c(\mathbf{W})=\sum_{k=2}^\infty\frac{\theta^{k-2}}{k!}\E\left[|\Delta V(\mathbf{W})|^k\right]
\]
Here, $c(\mathbf{W})$ depends on $\mathbf{W}$. To obtain a bound independent of $\mathbf{W}$, we note that assumption (ii) gives
\[
c(\mathbf{W})=\frac{\E\left[e^{\theta|\Delta V(\mathbf{W})|}\right]-(1+\theta\E[|\Delta V(\mathbf{W})|])}{\theta^2}\leq\frac{\E\left[e^{\theta|\Delta V(\mathbf{W})|}\right]}{\theta^2}\leq\frac{D}{\theta^2}
\]

Taking expectation on both sides with respect to the stationary distribution of $\mathbf{W}(t)$ gives us 
\[
&\E\left[e^{\eta(V(\mathbf{W}_+)\wedge n)}-e^{\eta(V(\mathbf{W})\wedge n)}\right]\\
\leq&\,\E\left[\E\left[\eta\Delta V(\mathbf{W})\mid \mathbf{W}\right]e^{\eta V(\mathbf{W})}\mathbf{1}(V(\mathbf{W})\leq n)\right]+\frac{\eta^2D}{\theta^2}\,\E\left[e^{\eta V(\mathbf{W})}\mathbf{1}(V(\mathbf{W})\leq n)\right]
\]

Now we turn to analyze the first term on the LHS of \ref{RCL_SSC}. First observe that for any $\mathbf{w}$,
\[\eta D_t(V(\mathbf{w})\wedge n)e^{\eta(V(\mathbf{w})\wedge n)}\leq\eta D_t(V(\mathbf{w}))e^{\eta V(\mathbf{w})}\mathbf{1}(V(\mathbf{w})\leq n)\]
By assumption (i),
\[-\mathcal{G}V(\mathbf{w})=-D_tV(\mathbf{w})-\lambda\E[\Delta V(\mathbf{w})]\geq \alpha-\beta\cdot\mathbf{1}(V(\mathbf{w})\leq K)\]
Therefore, for any $\mathbf{w}$,
\begin{align*}
&-\eta D_t(V(\mathbf{w}))e^{\eta V(\mathbf{w})}\mathbf{1}(V(\mathbf{w})\leq n)\\
\geq&\,\eta(\alpha-\beta\mathbf{1}(V(\mathbf{w})\leq K)+\lambda\E[\Delta V(\mathbf{w})])e^{\eta V(\mathbf{w})}\mathbf{1}(V(\mathbf{w})\leq n)\\
=&\,(\alpha+\lambda\E[\Delta V(\mathbf{w})])\eta e^{\eta V(\mathbf{w})}\mathbf{1}(V(\mathbf{w})\leq n)-\beta\eta e^{\eta V(\mathbf{w})}\mathbf{1}(V(\mathbf{w})\leq n)\mathbf{1}(V(\mathbf{w})\leq K)\\
\geq&\,(\alpha+\lambda\E[\Delta V(\mathbf{w})])\eta e^{\eta V(\mathbf{w})}\mathbf{1}(V(\mathbf{w})\leq n)-\beta\eta e^{\eta K}
\end{align*}

Taking expectation on both sides with respect to $\pi$ gives us 
\begin{align*}
&-\eta\E\left[D_t(V(\mathbf{W}))e^{\eta V(\mathbf{W})}\mathbf{1}(V(\mathbf{W})\leq n)\right]\\
\geq
&\,\alpha\eta\E\left[e^{\eta V(\mathbf{W})}\mathbf{1}(V(\mathbf{W})\leq n)\right]+\lambda\E\left[\E\left[\eta\Delta V(\mathbf{W})\mid \mathbf{W}\right]e^{\eta V(\mathbf{W})}\mathbf{1}(V(\mathbf{W})\leq n)\right]-\beta\eta e^{\eta K}
\end{align*}
Rearranging \eqref{RCL_SSC} and applying all inequalities obtained above give us
\[
&\alpha\eta\E\left[e^{\eta V(\mathbf{W})}\mathbf{1}(V(\mathbf{W})\leq n)\right]+\lambda\E\left[\E\left[\eta\Delta V(\mathbf{W})\mid \mathbf{W}\right]e^{\eta V(\mathbf{W})}\mathbf{1}(V(\mathbf{W})\leq n)\right]-\beta\eta e^{\eta K}\\
\leq&\,-\eta\E\left[D_t(V(\mathbf{W}))e^{\eta V(\mathbf{W})}\mathbf{1}(V(\mathbf{W})\leq n)\right]\\
=&\,\lambda\E\left[e^{\eta(V(\mathbf{W}_+)\wedge n)}-e^{\eta(V(\mathbf{W})\wedge n)}\right]\\
\leq&\,\lambda\E\left[\E\left[\eta\Delta V(\mathbf{W})\mid \mathbf{W}\right]e^{\eta V(\mathbf{W})}\mathbf{1}(V(\mathbf{W})\leq n)\right]+\lambda\frac{\eta^2D}{\theta^2}\,\E\left[e^{\eta V(\mathbf{W})}\mathbf{1}(V(\mathbf{W})\leq n)\right]
\]
Rearranging,
\[\left(\alpha-\lambda\frac{\eta D}{\theta^2}\right)\E\left[e^{\eta V(\mathbf{W})}\mathbf{1}(V(\mathbf{W})\leq n)\right]\leq\beta e^{\eta K}\]
Take $\eta>0$ so small that $\alpha-\lambda\frac{\eta D}{\theta^2}>0$. We have
\[\E\left[e^{\eta V(\mathbf{W})}\mathbf{1}(V(\mathbf{W})\leq n)\right]\leq\frac{\theta^2\beta e^{\eta K}}{\theta^2\alpha-\lambda\eta D}\]
Letting $n\to\infty$ and invoking monotone convergence theorem complete the proof.
\hfill\Halmos
\end{proof}

\subsection{Proof of \Cref{thm:bucket_SSC}}
\label{app:bucket_SSC}
We would like to prove the bucket state-space collapse to the cone $\mathcal{C}$ result as discussed in \Cref{sec:cone_SSC,sec:bucket_SSC}. 

We first note that condition (ii) of \cref{thm:SSC} is satisfied because $|\Delta H_i(\mathbf{w})|\leq(1+\tan\varphi)\sqrt{n_c}b_i$ according to \cref{eq:H_Lipschitz}. 

\restate*\ref{thm:bucket_SSC}

\begin{proof}{\textit{Proof.}}
Under SEB, a bucket alternates between two modes: pure work accumulation mode, where no service is given to the bucket, and service mode, where all servers work on the bucket. We will analyze the drifts under both modes separately.

In the work accumulation mode, the generator for the work vector $\mathbf{w}$ is bounded by
\[
\mathcal{G}H_i(\mathbf{w})=\lambda\E[\Delta H_i(\mathbf{w})]\leq\min\left\{\frac{\lambda_i\E[\|\mathbf{Z}_{\perp{\boldsymbol{\rho}^{\boldsymbol{\nu}}}}\|_2^2]}{2H_i(\mathbf{w})}-\|{\boldsymbol{\rho}^{\boldsymbol{\nu}}}^{(i)}\|_2\tan\varphi, \lambda_i(1+\tan\varphi)\sqrt{n_c}b_i\right\}
\]
according to \cref{eq:H_arrival_drift_bound}.

Now we show that, in the service mode, $\mathbf{w}^{(i)}$ also collapses to cone $\mathcal{C}$. We note that $H_i(\mathbf{w})$ is not everywhere differentiable. Specifically, it is not differentiable when $\|\mathbf{w}_{\perp{\boldsymbol{\rho}^{\boldsymbol{\nu}}}}^{(i)}\|_2-\|\mathbf{w}_{\parallel{\boldsymbol{\rho}^{\boldsymbol{\nu}}}}^{(i)}\|_2\tan\varphi=0$ (i.e. on the surface of the cone $\mathcal{C}$) and when $\mathbf{w}=0$. Since the system idles when $\mathbf{w}=0$, nondifferentiability there does not concern us. Despite the nondifferentiability at the surface of $\mathcal{C}$, $H_i(\mathbf{w})$ remains one-sided differentiable, which is sufficient for generators. Below, we will first handle the differentiable cases, then discuss how to address the one-sided differentiable region.

The generator for the bucket, when $H_i$ is differentiable and the bucket is in service mode, is
\[
\mathcal{G}H_i(\mathbf{w})=-\inner{\nabla H_i(\mathbf{w})}{\mathbf{r}^*}+\lambda\E[\Delta H_i(\mathbf{w})]
\]

In light of the drift analysis of $\lambda\E[\Delta H_i(\mathbf{w})]$, it suffices to show that $\inner{\nabla H_i(\mathbf{w})}{\mathbf{r}^*}>0$. 
We first compute $\nabla H_i(\mathbf{w})$.
\begin{align*}
\nabla H_i(\mathbf{w})&=\nabla\left(\|\mathbf{w}^{(i)}_{\perp{\boldsymbol{\rho}^{\boldsymbol{\nu}}}}\|_2-\|\mathbf{w}^{(i)}_{\parallel{\boldsymbol{\rho}^{\boldsymbol{\nu}}}}\|_2\tan\varphi\right)
\end{align*}
We have
\begin{align*}
\nabla\|\mathbf{w}^{(i)}_{\parallel{\boldsymbol{\rho}^{\boldsymbol{\nu}}}}\|_2&=\frac{{\boldsymbol{\rho}^{\boldsymbol{\nu}}}}{\|{\boldsymbol{\rho}^{\boldsymbol{\nu}}}\|_2^2}\frac{\inner{\mathbf{w}^{(i)}}{{\boldsymbol{\rho}^{\boldsymbol{\nu}}}}}{\|\mathbf{w}^{(i)}_{\parallel{\boldsymbol{\rho}^{\boldsymbol{\nu}}}}\|_2}=\frac{{\boldsymbol{\rho}^{\boldsymbol{\nu}}}}{\|{\boldsymbol{\rho}^{\boldsymbol{\nu}}}\|_2}\\
\nabla\|\mathbf{w}^{(i)}_{\perp{\boldsymbol{\rho}^{\boldsymbol{\nu}}}}\|_2&=\frac{\mathbf{w}^{(i)}_{\perp{\boldsymbol{\rho}^{\boldsymbol{\nu}}}}}{\|\mathbf{w}^{(i)}_{\perp{\boldsymbol{\rho}^{\boldsymbol{\nu}}}}\|_2}
\end{align*}
Thus, by assumptions on $\varphi$ in \Cref{sec:cone_SSC},
\[
\inner{\nabla H_i(\mathbf{w})}{\mathbf{r}^*}=\frac{\inner{\mathbf{w}^{(i)}_{\perp{\boldsymbol{\rho}^{\boldsymbol{\nu}}}}}{\mathbf{r}^*}}{\|\mathbf{w}^{(i)}_{\perp{\boldsymbol{\rho}^{\boldsymbol{\nu}}}}\|_2}-\frac{\inner{{\boldsymbol{\rho}^{\boldsymbol{\nu}}}}{\mathbf{r}^*}}{\|{\boldsymbol{\rho}^{\boldsymbol{\nu}}}\|_2}\tan\varphi>0
\]

When $\mathbf{w}^{(i)}\in\mathcal{C}$, we have 
\[
\mathcal{G}H_i(\mathbf{w})=(-\inner{\nabla H_i(\mathbf{w})}{\mathbf{r}^*})^++\lambda\E[\Delta H_i(\mathbf{w})]
\]
Again by assumptions on $\varphi$ in \Cref{sec:cone_SSC}, 
\[
\mathcal{G}H_i(\mathbf{w})=\lambda\E[\Delta H_i(\mathbf{w})]
\]

Set
\[
K_i=\frac{\lambda_i\E\left[\|\mathbf{Z}^{(i)}_{\perp{\boldsymbol{\rho}^{\boldsymbol{\nu}}}}\|_2^2\right]}{\|{\boldsymbol{\rho}^{\boldsymbol{\nu}}}^{(i)}\|_2}\cot\varphi
\]
as in \Cref{thm:SSC}. We have
\[
\mathcal{G}H_i(\mathbf{w})\leq&-\frac{\|{\boldsymbol{\rho}^{\boldsymbol{\nu}}}^{(i)}\|_2}{2}\tan\varphi+(1+\tan\varphi)\sqrt{n_c}b_i\mathbf{1}(H(\mathbf{w})\leq K_i)
\]
by discarding the nonpositive drift term due to service. We set 
\[
\alpha_i&=\frac{\|{\boldsymbol{\rho}^{\boldsymbol{\nu}}}^{(i)}\|_2}{2}\tan\varphi\\
\beta_i&=\lambda_i(1+\tan\varphi)\sqrt{n_c}b_i\\
\theta_i&=\frac{2}{(1+\tan\varphi)\sqrt{n_c}b_{i}}
\]
as in \Cref{thm:SSC}. It follows that $D=e^2$.

With these parameters in hand, we are ready to bound $\E[e^{\eta H(\mathbf{w})}]$ in terms of $b_{i-1}$ and $b_{i}$. First note that we have
\[
\sqrt{n_c}\lambda_ib_{i-1}\leq \|{\boldsymbol{\rho}^{\boldsymbol{\nu}}}^{(i)}\|_2\leq \sqrt{n_c}\lambda_ib_{i}
\]
which immediately yields the following bounds
\[
\frac{\sqrt{n_c}\lambda_ib_{i-1}}{2}\tan\varphi \leq &\alpha_i \leq \frac{\sqrt{n_c}\lambda_ib_{i}}{2}\tan\varphi\\
&\beta_i \leq \lambda_i(1+\tan\varphi)\sqrt{n_c}b_i \\
&K_i\leq\sqrt{n_c}\frac{b_{i}^2}{b_{i-1}}\cot\varphi
\]
Recall that $\eta_i$ in \Cref{thm:SSC} must be taken so that $0<\eta_i<\min\left\{\frac{\alpha_i\theta_i^2}{\lambda_i D},\theta_i\right\}$. Since we have
\[
\frac{\alpha_i\theta_i^2}{\lambda_i D}\geq\frac{2\tan\varphi}{(1+\tan\varphi)^2}\frac{e^{-2}}{\sqrt{n_c}}\frac{b_{i-1}}{b_{i}^2},
\]
$\eta_i=\frac{\tan\varphi}{(1+\tan\varphi)^2}\frac{e^{-2}}{\sqrt{n_c}}\frac{b_{i-1}}{b_{i}^2}$ is a fine choice. Then we have
\[
\eta_iK_i&\leq\frac{\tan\varphi}{(1+\tan\varphi)^2}\frac{e^{-2}}{\sqrt{n_c}}\frac{b_{i-1}}{b_{i}^2}\cdot\sqrt{n_c}\frac{b_i^2}{b_{i-1}}\cot\varphi=\frac{e^{-2}}{(1+\tan\varphi)^2}\leq e^{-2}\\
\theta_i^2\beta_ie^{\eta_i K_i} &\leq \frac{4}{(1+\tan\varphi)^2n_cb_i^2}\lambda_i\sqrt{n_c}(1+\tan\varphi)b_ie^{e^{-2}}\leq\frac{1}{1+\tan\varphi}\frac{5}{\sqrt{n_c}}\frac{\lambda_i}{b_i}\\
\theta_i^2\alpha_i-\lambda_i\eta_iD &\geq \frac{4}{(1+\tan\varphi)^2n_cb_i^2}\frac{\sqrt{n_c}\lambda_ib_{i-1}}{2}\tan\varphi-\lambda_i\frac{\tan\varphi}{(1+\tan\varphi)^2}\frac{e^{-2}}{\sqrt{n_c}}\frac{b_{i-1}}{b_i^2}e^2\\
&=\lambda_i\frac{\tan\varphi}{(1+\tan\varphi)^2}\frac{1}{\sqrt{n_c}}\frac{b_{i-1}}{b_i^2}
\]
Finally, we invoke \Cref{thm:SSC} to obtain
\[
\E\left[e^{\eta_iH_i(\mathbf{W})}\right]&\leq\frac{5(1+\tan\varphi)}{\tan\varphi}\frac{b_i}{b_{i-1}}
\]
\hfill\Halmos
\end{proof}

\subsection{Proof of \Cref{thm:eligibility_work_amount}}
\label{app:eligibility_work_amount}
\restate*\ref{thm:eligibility_work_amount}

The theorem is a corollary of the following lemma, which relates the $\ell_1$ and $\ell_\infty$ norms for workload vectors in the cone $\mathcal{C}$.

\begin{lemma}\label{lemma:bucket_work_lower_bound}
If $\mathbf{w}^{(i)}\in\mathcal{C}$, then
\[
\min_{1 \leq j \leq n_c} \mathbf{w}^{(i)}_j \geq \|\mathbf{w}^{(i)}\|_1 \cdot \frac{1}{\sqrt{n_c}}\left(\frac{{\boldsymbol{\rho}^{\boldsymbol{\nu}}}_{\min}}{\|{\boldsymbol{\rho}^{\boldsymbol{\nu}}}\|_2}\cos\varphi-\sin\varphi\right)
\]
\end{lemma}
\begin{proof}{\textit{Proof.}}
Let $\mathbf{e}_j$ be the $j$-th standard basis in $\mathbb{R}^n$, we have
\begin{align*}
\inner{\mathbf{w}^{(i)}}{\mathbf{e}_j}&=\inner{\mathbf{w}^{(i)}_{\|{\boldsymbol{\rho}^{\boldsymbol{\nu}}}}}{\mathbf{e}_j}+\inner{\mathbf{w}^{(i)}_{\perp{\boldsymbol{\rho}^{\boldsymbol{\nu}}}}}{\mathbf{e}_j}\\
&\geq\frac{\inner{\mathbf{w}^{(i)}}{{\boldsymbol{\rho}^{\boldsymbol{\nu}}}}}{\|{\boldsymbol{\rho}^{\boldsymbol{\nu}}}\|_2^2}{\boldsymbol{\rho}^{\boldsymbol{\nu}}}_j-\|\mathbf{w}^{(i)}_{\perp{\boldsymbol{\rho}^{\boldsymbol{\nu}}}}\|_2\\
&\stackrel{(a)}{\geq}\frac{\inner{\mathbf{w}^{(i)}}{{\boldsymbol{\rho}^{\boldsymbol{\nu}}}}}{\|{\boldsymbol{\rho}^{\boldsymbol{\nu}}}\|_2^2}\rho_i-\|\mathbf{w}^{(i)}\|_2\sin\varphi\\
&\stackrel{(b)}{\geq}\frac{{\boldsymbol{\rho}^{\boldsymbol{\nu}}}_{j}}{\|{\boldsymbol{\rho}^{\boldsymbol{\nu}}}\|_2}\|\mathbf{w}^{(i)}\|_2\cos\varphi-\|\mathbf{w}^{(i)}\|_2\sin\varphi\\
&\stackrel{(c)}{\geq}\frac{1}{\sqrt{n_c}}\|\mathbf{w}^{(i)}\|_1\left(\frac{{\boldsymbol{\rho}^{\boldsymbol{\nu}}}_{\min}}{\|{\boldsymbol{\rho}^{\boldsymbol{\nu}}}\|_2}\cos\varphi-\sin\varphi\right)
\end{align*}
where (a) and (b) follow from the assumption that $\mathbf{w}^{(i)}\in\mathcal{C}$ and (c) follows from the norm inequality. Note further that 
\[
\frac{{\boldsymbol{\rho}^{\boldsymbol{\nu}}}_{\min}}{\|{\boldsymbol{\rho}^{\boldsymbol{\nu}}}\|_2}\cos\varphi-\sin\varphi>0
\]
because of assumption in \Cref{sec:cone_SSC} on $\varphi$.
\hfill\Halmos
\end{proof}

\subsection{Proof of \Cref{lem:H_lower_bound}}
\label{app:H_lower_bound}

\restate*\ref{lem:H_lower_bound}

\begin{proof}{\textit{Proof.}}
We begin by showing that, for any $\mathbf{w}^{(j)}\not\in\mathcal{C}$, there exists $\overline{\mathbf{w}}^{(j)}\in\mathcal{C}$ such that $H(\mathbf{w}^{(i)})=\|\mathbf{w}^{(j)} - \overline{\mathbf{w}}^{(j)}\|_2$. One can verify that
\[
    \overline{\mathbf{w}}^{(j)}=\mathbf{w}^{(j)}_{\parallel{\boldsymbol{\rho}^{\boldsymbol{\nu}}}}+\frac{\mathbf{w}^{(j)}_{\perp{\boldsymbol{\rho}^{\boldsymbol{\nu}}}}}{\|\mathbf{w}^{(j)}_{\perp{\boldsymbol{\rho}^{\boldsymbol{\nu}}}}\|_2}\|\mathbf{w}^{(j)}_{\parallel{\boldsymbol{\rho}^{\boldsymbol{\nu}}}}\|_2\tan\varphi
\]
is such a vector. 

Since $\mathbf{w}^{(j)}\not\in\mathcal{E}_i$, there exists $k=1,\ldots,n_c$ such that $\mathbf{w}^{(j)}_k < n_s b_j$. Thus,
\[
H(\mathbf{w}^{(i)})&=\|\mathbf{w}^{(j)} - \overline{\mathbf{w}}^{(j)}\|_2\geq \overline{\mathbf{w}}^{(j)}_k - \mathbf{w}^{(j)}_k \geq \frac{\tau}{\tan\varphi} \|\overline{\mathbf{w}}^{(j)}\|_1
\]
where the last inequality comes from \Cref{lemma:bucket_work_lower_bound}, which relates the $\ell_1$ and $\ell_\infty$ norms for workload vectors in the cone $\mathcal{C}$. It remains to lower bound $\|\overline{\mathbf{w}}^{(j)}\|_1$. Notice that we have
\[
\|\overline{\mathbf{w}}^{(j)}\|_1 = \left\|\mathbf{w}^{(j)}_{\parallel{\boldsymbol{\rho}^{\boldsymbol{\nu}}}}+\mathbf{w}^{(j)}_{\perp{\boldsymbol{\rho}^{\boldsymbol{\nu}}}}\frac{\|\mathbf{w}^{(j)}_{\parallel{\boldsymbol{\rho}^{\boldsymbol{\nu}}}}\|_2}{\|\mathbf{w}^{(j)}_{\perp{\boldsymbol{\rho}^{\boldsymbol{\nu}}}}\|_2}\tan\varphi\right\|_1\stackrel{(a)}{\geq} \min\{\|\mathbf{w}^{(j)}_{\parallel{\boldsymbol{\rho}^{\boldsymbol{\nu}}}}\|_1, \|\mathbf{w}^{(j)}\|_1\}
\]
where (a) follows from the following observations
\* Note that $\frac{\|\mathbf{w}^{(j)}_{\parallel{\boldsymbol{\rho}^{\boldsymbol{\nu}}}}\|_2}{\|\mathbf{w}^{(j)}_{\perp{\boldsymbol{\rho}^{\boldsymbol{\nu}}}}\|_2}$ can be viewed as cotangent of the angle between $\mathbf{w}^{(j)}$ and ${\boldsymbol{\rho}^{\boldsymbol{\nu}}}$. Since $\mathbf{w}^{(j)}\not\in\mathcal{C}$, this angle is larger than $\varphi$. This implies that $\frac{\|\mathbf{w}^{(j)}_{\parallel{\boldsymbol{\rho}^{\boldsymbol{\nu}}}}\|_2}{\|\mathbf{w}^{(j)}_{\perp{\boldsymbol{\rho}^{\boldsymbol{\nu}}}}\|_2}\tan\varphi\in(0,1)$.
\* The $\ell_1$ norm is a linear function, so the minimum is achieved whenever the constant $\frac{\|\mathbf{w}^{(j)}_{\parallel{\boldsymbol{\rho}^{\boldsymbol{\nu}}}}\|_2}{\|\mathbf{w}^{(j)}_{\perp{\boldsymbol{\rho}^{\boldsymbol{\nu}}}}\|_2}\tan\varphi$ is 0 or 1.
\*/

Analyzing $\|\mathbf{w}^{(j)}_{\parallel{\boldsymbol{\rho}^{\boldsymbol{\nu}}}}\|_1$ further gives us
\[
\|\mathbf{w}^{(j)}_{\parallel{\boldsymbol{\rho}^{\boldsymbol{\nu}}}}\|_1 = \frac{\inner{\mathbf{w}^{(j)}}{{\boldsymbol{\rho}^{\boldsymbol{\nu}}}}}{\|{\boldsymbol{\rho}^{\boldsymbol{\nu}}}\|_2^2}\|{\boldsymbol{\rho}^{\boldsymbol{\nu}}}\|_1\geq \|\mathbf{w}^{(j)}\|_1 \frac{{\boldsymbol{\rho}^{\boldsymbol{\nu}}}_{\min}}{\|{\boldsymbol{\rho}^{\boldsymbol{\nu}}}\|_2}\frac{\|{\boldsymbol{\rho}^{\boldsymbol{\nu}}}\|_1}{\|{\boldsymbol{\rho}^{\boldsymbol{\nu}}}\|_2}\stackrel{(b)}{\geq} \|\mathbf{w}^{(j)}\|_1 \frac{{\boldsymbol{\rho}^{\boldsymbol{\nu}}}_{\min}}{\|{\boldsymbol{\rho}^{\boldsymbol{\nu}}}\|_2}\stackrel{(c)}{\geq} \|\mathbf{w}^{(j)}\|_1 \tan\varphi
\]
where (b) follows from the norm inequality and (c) follows from the assumption on $\varphi$ in \Cref{sec:assumptions}. This completes the proof of the lemma.
\hfill\Halmos
\end{proof}

\subsection{Proof of \Cref{lemma:WINE_for_wait_time}}
\label{app:WINE_for_wait_time}

\restate*\ref{lemma:WINE_for_wait_time}

\begin{proof}{\textit{Proof.}}
To bound $\E[T^\pi_{\text{wait}}]$, we apply WINE (\Cref{prop:WINE}) to the queueing system consists only of jobs waiting to enter service. In this system, the remaining priority indices of jobs as in \Cref{prop:WINE} is their original priority indices, as jobs leave the system as soon as they receive any service. 

By \Cref{prop:WINE},
\begin{align*}
    \E[T^\pi_{\text{wait}}]
    &\leq \frac{1}{\lambda} \int_{0}^{\infty} \frac{\E[W^\pi_{\origduration \leq x}]}{x^2} \, dx \\
    &= \frac{1}{\lambda} \int_{0}^{b_{n_b}} \frac{\E[W^\pi_{\origduration \leq x}]}{x^2} \, dx
        + \frac{1}{\lambda} \int_{b_{n_b}}^{\infty} \frac{\E[W^\pi_{\origduration \leq x}]}{x^2} \, dx \\
    &= \frac{1}{\lambda} \sum_{i=1}^{n_b} \int_{b_{i-1}}^{b_{i}} \frac{\E[W^\pi_{\origduration \leq x}]}{x^2} \, dx
        + \frac{1}{\lambda} \frac{\E[W^\pi_{\origduration \leq b_{n_b}}]}{b_{n_b}} \\
    &\leq \frac{1}{\lambda} \sum_{i=1}^{n_b} \int_{b_{i-1}}^{b_{i}} \frac{\E[W^\pi_{\origduration \leq b_i}]}{x^2} \, dx
        + \frac{1}{\lambda} \frac{\E[W^\pi_{\origduration \leq b_{n_b}}]}{b_{n_b}} \\
    &= \frac{1}{\lambda} \sum_{i=1}^{n_b} \E[W^\pi_{\origduration \leq b_i}] \left(\frac{1}{b_{i-1}} - \frac{1}{b_{i}}\right)
        +\frac{1}{\lambda} \frac{\E[W^\pi_{\origduration \leq b_{n_b}}]}{b_{n_b}} \\
    &= \frac{1}{\lambda} \sum_{i=1}^{n_b} \frac{c-1}{b_0 c^i} \E[W^\pi_{\leq i}]
        + \frac{1}{\lambda} \frac{\E[W^\pi]}{b_{n_b}}.
\end{align*}
\hfill\Halmos
\end{proof}

\subsection{Proof of \Cref{thm:response_time_bound}}
\label{app:response_time_bound}

\restate*\ref{thm:response_time_bound}

\begin{proof}{\textit{Proof.}}
By \Cref{lemma:WINE_for_wait_time},
\[
    \E[T^\SEB_{\text{wait}}] \leq \underbrace{\frac{1}{\lambda} \sum_{i=1}^{n_b} \frac{c-1}{b_0 c^i} \E[W^\SEB_{\leq i}]}_{\mathcal{T}_1}
        + \underbrace{\frac{1}{\lambda} \frac{\E[W^\SEB]}{b_{n_b}}}_{\mathcal{T}_2}
\]

Recall from \Cref{prop:WDL} that for any $i=1,\ldots,n_b$,
\[
    \E[W_{\leq i}^\SEB] - \E[W_{\leq i}^{\mgone}] = \frac{\E[\mathcal{I}^\SEB_{\leq i}W^\SEB_{\leq i}]}{1-\rho^{\boldsymbol{\nu}}_{\leq i}}
\]

Using \Cref{prop:WDL} and the waste bound in \Cref{thm:waste_bound}, we obtain
\begin{align*}
\mathcal{T}_1&\leq \frac{1}{\lambda}\sum_{i=1}^{n_b}\frac{c-1}{b_0c^i}\E[W_{\leq i}^{\mgone}]+\frac{A_1+A_2}{\lambda}\sum_{i=1}^{n_b}\frac{c^{i}-1}{b_0c^i}+\frac{A_2}{\lambda}\sum_{i=1}^{n_b}\log\left(\frac{1}{1-\rho^{\boldsymbol{\nu}}_{\leq i}}\right)\frac{c^{i}-1}{b_0c^i}+\frac{A_2}{\lambda}\sum_{i=1}^{n_b}i\frac{c-1}{b_0c^i}\\
&\leq \frac{1}{\lambda}\sum_{i=1}^{n_b}\frac{c-1}{b_0c^i}\E[W_{\leq i}^{\mgone}]+\frac{1}{\lambda}\left(\frac{A_1+A_2}{b_0}+\frac{A_2}{b_0}\log\left(\frac{1}{1-\rho^{\boldsymbol{\nu}}}\right)\right)\sum_{i=1}^{n_b}\frac{c^{i}-1}{c^i}+\frac{1}{\lambda}\frac{A_2}{b_0}n_b\sum_{i=1}^{n_b}\frac{c-1}{c^i}\\
&\leq \frac{1}{\lambda}\sum_{i=1}^{n_b}\frac{c-1}{b_0c^i}\E[W_{\leq i}^{\mgone}]+\frac{1}{\lambda}\left(\frac{A_1+2A_2}{b_0}+\frac{A_2}{b_0}\log\left(\frac{1}{1-\rho^{\boldsymbol{\nu}}}\right)\right)n_b\\
\mathcal{T}_2&\leq\frac{1}{\lambda}\frac{\E[W^{\mgone}]}{b_{n_b}}+\frac{1}{\lambda}\left[\left(\frac{A_1+A_2}{b_{n_b}}+\frac{A_2}{b_{n_b}}\log\left(\frac{1}{1-\rho^{\boldsymbol{\nu}}}\right)\right)\frac{c^{n_b}-1}{c-1}+\frac{A_2}{b_{n_b}}n_b\right]
\end{align*}

Combining bounds for $\mathcal{T}_1$ and $\mathcal{T}_2$ and noting that $b_{n_b}-b_0=b_0(c^{n_b}-1)$ give us
\begin{align}
    \nonumber
    \E[T_{\text{wait}}^\SEB]\leq&\,\frac{1}{\lambda}\sum_{i=1}^{n_b}\frac{c-1}{b_0c^i}\E[W_{\leq i}^{\mgone}]+\frac{1}{\lambda}\frac{\E[W^{\mgone}]}{b_{n_b}}+\frac{1}{\lambda}\left(\frac{A_1+A_2}{b_0}+\frac{A_2}{b_0}\log\left(\frac{1}{1-\rho^{\boldsymbol{\nu}}}\right)+\frac{A_2}{b_0}\right)n_b+\\
    \nonumber
    &\,\frac{1}{\lambda}\left(\frac{A_1+A_2}{b_{n_b}}+\frac{A_2}{b_{n_b}}\log\left(\frac{1}{1-\rho^{\boldsymbol{\nu}}}\right)\right) \frac{b_{n_b}-b_0}{b_0}\frac{1}{c-1}+\frac{1}{\lambda}\frac{A_2}{b_{n_b}}n_b\\
    \label{eq:PSJF-second-to-last}
    \leq&\, \frac{1}{\lambda}\sum_{i=1}^{n_b}\frac{c-1}{b_0c^i}\E[W_{\leq i}^{\mgone}]+\frac{1}{\lambda}\frac{\E[W^{\mgone}]}{b_{n_b}}+\frac{A}{\lambda} \log\left(\frac{1}{1-\rho^{\boldsymbol{\nu}}}\right) \left(2n_b + \frac{1}{c - 1}\right)
\end{align}

Next, we show that the first two terms involving $\E[W^{M/G/1}]$ are upper bounded by $c \E{T^\PSJF}$.

We now consider a single-server Preemptive-Shortest-Job-First (PSJF) policy. Let $T^\PSJF$ be its response time, then we have
\begin{align}\label{eq:PSJF_RT_bound}
\E[T^\PSJF]\stackrel{(a)}{=}&\frac{1}{\lambda}\int_{0}^{\infty}\frac{\E[W^\PSJF_{\remduration \leq x}]}{x^2}\,dx \nonumber \\
\stackrel{(b)}{\geq}& \frac{1}{\lambda}\int_{0}^{\infty}\frac{\E[W^\PSJF_{\origduration \leq x}]}{x^2}\,dx \nonumber\\
\stackrel{(c)}{=}& \frac{1}{\lambda}\frac{1}{c}\int_{0}^{\infty}\frac{\E[W^\PSJF_{\origduration \leq cx}]}{x^2}\,dx \nonumber\\
\stackrel{(d)}{=}& \frac{1}{\lambda}\frac{1}{c}\sum_{i=1}^{n_b}\int_{b_{i-1}}^{b_i}\frac{\E[W^\PSJF_{\origduration \leq cx}]}{x^2}\,dx+\frac{1}{\lambda}\frac{1}{c}\int_{b_{n_b}}^\infty\frac{\E[W^\PSJF_{\origduration \leq cx}]}{x^2}\,dx \nonumber\\
\stackrel{(e)}{\geq}& \frac{1}{\lambda}\frac{1}{c}\sum_{i=1}^{n_b}\int_{b_{i-1}}^{b_i} \frac{\E[W^\PSJF_{\origduration \leq b_i}]}{x^2}\,dx+\frac{1}{\lambda}\frac{1}{c}\int_{b_{n_b}}^\infty\frac{\E[W^\PSJF_{\origduration \leq cx}]}{x^2}\,dx \nonumber\\
\stackrel{(f)}{=}& \frac{1}{\lambda}\frac{1}{c}\sum_{i=1}^{n_b}\int_{b_{i-1}}^{b_i} \frac{\E[W^{\mgone}_{\leq i}]}{x^2}\,dx+\frac{1}{\lambda}\frac{1}{c}\int_{b_{n_b}}^\infty\frac{\E[W^{\mgone}]}{x^2}\,dx \nonumber\\
=& \frac{1}{\lambda}\frac{1}{c}\sum_{i=1}^{n_b}\frac{c-1}{b_0c^i}\E[W^{\mgone}_{\leq i}]+\frac{1}{\lambda}\frac{1}{c}\frac{1}{b_{n_b}}\E[W^{\mgone}]
\end{align}

where (a) follows from the WINE identity. (b) follows from the fact that 
$W^\PSJF_{\remduration \leq x}$ includes work from jobs with original priority indices less than $x$ (i.e. $W^\PSJF_{\origduration \leq x}$) and jobs with original priority indices larger than $x$ but have remaining priority indices less than $x$ because they have received some service. (c) follows from a change of variable $x\mapsto cx$. To better compare $\E[T^\PSJF]$ and $\E[T^\SEB]$, we divide the integral into $b_0, b_1,\ldots,b_{n_b},\infty$ in (d) the same way we define the buckets in MSJ scheduling. (e) follows from $b_{i-1}\leq x\leq b_i\leq cx$. (f) follows from the fact that PSJF is a work-conserving policy.

Comparing \eqref{eq:PSJF_RT_bound} with \eqref{eq:PSJF-second-to-last} completes the proof.
\hfill\Halmos
\end{proof}








    

\end{APPENDICES}

\end{document}